\documentclass[11pt]{article}

\usepackage[margin=1in]{geometry}
\usepackage[T1]{fontenc}

\usepackage{babel}
\usepackage[spacing=true,kerning=true,babel=true,tracking=true]{microtype}

\usepackage{amsmath,amsfonts,amssymb,amsthm,dcolumn,bm,qcircuit,appendix,mathtools,thmtools,thm-restate,mathrsfs,pgfplots,soul,lipsum,graphicx,braket,enumitem,caption,subcaption,ragged2e}

\usepackage[linesnumbered,ruled,vlined]{algorithm2e}
\SetKwInput{KwInput}{Input}
\SetKwInput{KwOutput}{Output}
\SetKwInOut{Promise}{Promise}
\SetKwInput{Goal}{Goal}
\SetKwProg{Fn}{function}{}{}
\SetKwFor{RepTimes}{repeat}{times}{}
\SetKwFunction{Ver}{Verify}
\SetKwFunction{Prep}{PrepareState}
\usepackage{algorithmic}

\usepackage[most]{tcolorbox}
\usepackage{tikz,tkz-euclide,tikz-3dplot}
\usetikzlibrary{matrix,fit,backgrounds,3d,arrows, decorations.pathreplacing,shapes.geometric,3d, calc}

\usepackage[dvipsnames]{xcolor}
\definecolor{blueviolet}{rgb}{0.2, 0.2, 0.6}
\definecolor{citered}{rgb}{0.70,0.13,0.13}

\usepackage[colorlinks]{hyperref}
\usepackage[capitalize,nameinlink]{cleveref}
\hypersetup{
	colorlinks = true,
	linkcolor = [rgb]{0.70,0.13,0.13},
	citecolor = [rgb]{0.13,0.55,0.13},
	urlcolor = [rgb]{0.25, 0.41, 0.88}}
\usepackage{placeins}
\usepackage{comment}
\usepackage{textgreek}

\newtheorem{definition}{Definition}
\newtheorem{theorem}{Theorem}
\newtheorem{lemma}{Lemma}

\newtheorem{result}{Main Result}

\numberwithin{proposition}{section}
\numberwithin{theorem}{section}
\numberwithin{lemma}{section}
\numberwithin{corollary}{section}
\numberwithin{remark}{section}
\numberwithin{definition}{section}
\numberwithin{problem}{section}
\numberwithin{equation}{section}

\DeclareMathOperator{\Tr}{Tr}
\newcommand{\xbf}{\textbf{x}}
\newcommand{\ubf}{\textbf{u}}
\newcommand{\vbf}{\textbf{v}}
\newcommand{\wbf}{\textbf{w}}

\newcommand{\Ibb}{\mathbb{I}}

\newcommand{\Rbb}{\mathbb{R}}

\newcommand{\Ccal}{\mathcal{C}}
\newcommand{\Xcal}{\mathcal{X}}

\begin{document}

\title{Toward speedup without quantum coherent access}

\author{
        Nhat A. Nghiem\thanks{C. N. Yang Institute for Theoretical Physics and Department of Physics and Astronomy, State University of New York at Stony Brook.~\href{mailto:nhatanh.nghiemvu@stonybrook.edu}{nhatanh.nghiemvu@stonybrook.edu}} \quad 
        }
\date{ }

\maketitle

\begin{abstract}
Along with the development of quantum technology, finding useful applications of quantum computers has been a central pursuit. Despite various quantum algorithms have been developed, many of them often require strong input assumptions, which is hardware demanding. In particular, recent advances on dequantization have revealed that the quantum advantage is more of a mere artifact of strong input assumptions. In this work, we propose a variant of these algorithms, leveraging both classical and quantum resources. Provided the classical knowledge (the entries) of the matrix/vector of interest, a classical procedure is used to pre-process this information. Then they are fed into a quantum circuit which is shown to be a block encoding of the matrix of interest. From this block-encoding, we show how to use it to tackle a wide range of problems, including principal component analysis, linear equation solving, Hamiltonian simulation, preparing ground state, and data fitting. We also analyze our protocol, showing that both the classical and quantum procedure can achieve logarithmic complexity in the input dimension, thus implying its potential for near-term realization. In particular, we will show that these complexities, especially the classical pre-processing time, can be significantly improved if the matrix/vector of interest admits certain structures, and we give concrete criteria for this. 

Several implications, byproducts, and corollaries are discussed. First, we reveal another particular type of quantum state that can be efficiently prepared, which can potentially find application elsewhere. Second, we show that a Hamiltonian $H$ with classically known rows/columns can be efficiently simulated, thus providing another model in addition to the well-known sparse access and linear combination of unitary models. Third, our results suggest there are certain matrices/Hamiltonians where the quantum linear solver and quantum simulation algorithm can achieve logarithmical complexity with respect to the sparsity parameter. As a result, our method provides exponential improvement compared to the existing ones in these scenarios. In particular, regarding dense linear systems, our method achieves exponential speed-up with respect to the inverse of error tolerance, compared to the best previously known quantum algorithm for dense systems. Last, and most importantly, regarding quantum data fitting, we show how the output of our quantum algorithms can be leveraged for useful purposes, e.g., predicting unseen data. Thus, it provides an end-to-end application, which has been an open aspect of the previous quantum data fitting algorithm. 

\end{abstract}

\newpage
\tableofcontents

\newpage

\section{Introduction}
Quantum computing has rapidly advanced since its early proposals~\cite{manin1980computable, benioff1980computer, feynman2018simulating, grover1996fast, shor1999polynomial}, with significant progress across diverse applications. Quantum algorithm for principal component analysis (PCA) -- a widely used tool in statistics and machine learning -- was proposed in~\cite{lloyd2014quantum}, marking an early milestone in quantum machine learning. Several extensions and variants of quantum PCA algorithm have been developed in subsequent attempts~\cite{gordon2022covariance, rodriguez2025quantum, bellante2022quantum, bellante2023quantum, nghiem2025new, tang2018quantum, nghiem2025quantum1}. Concurrently, the quantum algorithm for solving linear systems -- the backbone of many areas of science and engineering -- was introduced in~\cite{harrow2009quantum}, where 'solving' refers to preparing a quantum state proportional to $A^{-1} \textbf{b}$. This algorithm not only achieves exponential speed-up in the system's dimension but also proved matrix inversion to be BQP-complete, suggesting that it cannot be efficiently simulated classically~\cite{tang2018quantum, tang2021quantum}. Since then, many improvements and extensions have been proposed~\cite{childs2017quantum, clader2013preconditioned, wossnig2018quantum, nghiem2025new2, zhang2022quantum}. In particular, building upon the quantum linear solver, quantum algorithm for data fitting -- a very important tool in qualitative science -- was proposed in~\cite{wiebe2012quantum}. Quantum simulation is also a major topic of the field. Tremendous progress has been made in this direction~\cite{berry2007efficient,berry2012black,berry2014high, aharonov2003adiabatic, childs2010relationship, low2017optimal,low2019hamiltonian, childs2018toward, lloyd1996universal,berry2015hamiltonian,berry2015simulating,tran2021faster, childs2021theory,childs2019nearly,zhao2022hamiltonian}. Not only providing a compelling example for quantum speedup, recent advance in quantum simulation also underlies the namely quantum singular value transformation framework \cite{gilyen2019quantum}-- which has been shown to unify many quantum algorithms.  Aside from the time evolution, the properties of Hamiltonian, for example, its ground state, is also of fundamental importance. There is a rich body of works in this topic ~\cite{lin2020near, ding2024single, dong2022ground, ge2019faster, bespalova2021hamiltonian, motta2020determining}. 

The examples above feature exciting progress in the field of quantum algorithms. However, there are some open aspects that still put the quantum computing application in doubt. 

\begin{center}
    \textbf{Open Aspect 1: } Many of the above algorithms, e.g., \cite{harrow2009quantum, wiebe2012quantum, zlokapa2021quantum}, etc, hinge on strong input assumption, meaning that they require the classical data to be in quantum-accessible form. As mentioned in these works, quantum random access memory (QRAM) can provide the desired input; however, large-scale, fault-tolerance QRAM is still in its infancy, and generally it is hardware demanding. 
\end{center}
\begin{center}
    \textbf{Open Aspect 2: } Tang’s dequantization results~\cite{tang2018quantum, tang2021quantum} show that under a comparable query-and-access model—which also assumes efficient $l_2$-norm sampling—classical algorithms can match the performance of many existing quantum algorithms, e.g., PCA, up to polynomial factors. Therefore, even if quantum computers are available, whether they can still deliver meaningful usage is of great interest.  
\end{center}

\begin{center}
    \textbf{Open Aspect 3: } In the quantum data fitting context \cite{wiebe2012quantum}, their complexity has a factor $\mathcal{O}(s^6)$ where $s$ is the sparsity of some matrix (will be shown below). Thus, in practice, their algorithm is only effective when $s = \mathcal{O}(1)$. However, this is barely the case in reality, as we will elaborate via a simple example subsequently. This severely limits the potential of algorithm in \cite{wiebe2012quantum}, and generally of quantum computers to deliver real-world impacts, given that data fitting is arguably a standard tool in many areas. 
\end{center}

\begin{center}
    \textbf{Open Aspect 4:} Many of the outputs of existing quantum algorithms, such as the quantum data fitting algorithm \cite{wiebe2012quantum}, quantum neural tangent algorithm \cite{zlokapa2021quantum}, are quantum states. Performing tomography on these states to obtain any useful quantities is costly. Whether these states can be of practical value is more or less open.
\end{center}

In this work, inspired by recent advances in quantum algorithms \cite{gilyen2019quantum} and state preparation \cite{zhang2022quantum,grover2000synthesis,grover2002creating,plesch2011quantum, schuld2018supervised, nakaji2022approximate,marin2023quantum,zoufal2019quantum,prakash2014quantum}, we propose variants of all the aforementioned algorithms. The information we need is the classical knowledge of the input, e.g., matrix or vector entries. This information is first being preprocessed by a classical procedure, before feeding into a quantum circuit. We then show that the problems of interest, including PCA, linear equations, simulating quantum systems, preparing the ground state, can then be solved by executing this quantum circuit appropriately and also combining with other quantum algorithms. As a remark, the structure of the quantum circuit used in our work is explicit. In other words, our algorithm bypasses the need for QRAM, which, to some extent, answer the \textbf{Aspect 1} and \textbf{Aspect 2}. For the last two aspects, we defer the discussion to Section \ref{sec: overviewquantumdatafitting}. In the following, we proceed to describe the key technique behind our proposal, with our main results being the applications and corollaries of the technique.

%

\paragraph{Organization.} The rest of the paper is organized as follows. \cref{sec: overview} summarizes our main contributions. \cref{sec: overviewPCA} and \cref{sec: overviewlinearsystem} detail our QPCA and QLSA improvements, with comparisons in \cref{tab: pca} and \cref{tab: lineareq}. \cref{sec: overviewquantumsimulation} discusses our new quantum simulation model. Detailed algorithms and analysis appear in the appendix. 

\section{Overview of key technique}
\label{sec: keytechnique}
The first ingredient of our work is the state preparation protocols \cite{grover2000synthesis,grover2002creating,plesch2011quantum, schuld2018supervised, nakaji2022approximate,marin2023quantum,zoufal2019quantum,prakash2014quantum, zhang2022quantum}. 
Next, an essential ingredient of our work is based on block-encoding, which was recently introduced \cite{low2017optimal,low2019hamiltonian, gilyen2019quantum}. Roughly speaking, a unitary $U$ is said to exactly block-encode an operator $A$ (of lower dimension) if $A = (\bra{\bf 0} \otimes \Ibb ) U ( \Ibb \otimes \ket{\bf 0})$. A peculiar feature of block-encoding is that, given $U$ (and its transpose $U^\dagger$), one can perform many polynomial transformation, to transform $A$ to $P(A)$ where $P(.)$ is the polynomial. Besides, given two unitaries $U_1,U_2$ that block-encodes $A_1,A_2$, then we can form a block-encoding of their product, linear combinations, tensor product, etc. For our purpose, we defer the details to the Appendix \ref{sec: prelim} and \ref{sec: proofofLemmablockencodingknownmatrix}, and recapitulate the key tool in the following lemma, which is a corollary of Lemma \ref{lemma: stateprepration} and \ref{lemma: improveddme} in the same appendix.
\begin{lemma}[Block-encoding known matrix; Appendix \ref{sec: proofofLemmablockencodingknownmatrix}]
\label{lemma: blockencodingknownmatrix}
Provided the classical knowledge of entries of a matrix $\mathcal{A} \in \Rbb^{m \times n}$ having $s_\mathcal{A}$ nonzero entries with a promise $||\mathcal{A}|| \leq 1$ where $||.||$ is the operator norm. Defining $ \ket{\mathcal{A} } =\frac{1}{||\mathcal{A}||_F} \sum_{i=1}^n \sum_{j=1}^m \mathcal{A}_{ji} \ket{j}\ket{i } =   \frac{1}{||\mathcal{A}||_F} \sum_{i=1}^n  \mathcal{A}^i \ket{i}$ where $\mathcal{A}^i$ is the $i$-th column, $||\mathcal{A}||_F$ is the Frobenius norm, and $\kappa$ is the condition number of $\mathcal{A}$. Then:
\begin{itemize}
    \item (General case) The $\epsilon$-approximated block-encoding of $ \mathcal{A}$ can be constructed with a $\log mn$-qubits quantum circuit of depth $\mathcal{O}\left( ||\mathcal{A}||_F \log(s_\mathcal{A}) \kappa \log^2 \frac{\kappa}{\epsilon}  \right)$, using totally $\mathcal{O}\left( s_\mathcal{A} \right)$ ancilla qubits, and classical pre-processing of time $\mathcal{O}\left( \log mn \right)$, respectively \cite{zhang2022quantum}. The classical preprocessing cost can be improved to $\mathcal{O}(1)$ if $\{a_i\}_{i=1}^n$ can be partitioned into subsets in which each subset contains similar entries. 
    
    \item (Structured case 1) If the entries of $\ket{\mathcal{A}}$ has structure as any of the following works \cite{grover2000synthesis,grover2002creating,plesch2011quantum, schuld2018supervised, nakaji2022approximate,marin2023quantum,zoufal2019quantum,prakash2014quantum},  then $\mathcal{A}$ can be $\epsilon$-approximated block-encoded with a $\log mn$-qubits circuit of depth $\mathcal{O}\left( ||\mathcal{A}||_F\log (nm) \kappa \log^2 \frac{\kappa}{\epsilon}  \right)$. The total number of ancilla qubit can be $\mathcal{O}(\log s_{\mathcal{A}}$) or $\mathcal{O}(1)$.
    
    \item (Structured case 2) If $\ket{\mathcal{A}} = \sum_{i=1}^M \alpha_i \ket{\psi_{i_1}} \otimes \cdots \otimes \ket{\psi_{i_k}}$, and for all $i$, the entries of $ \{  \ket{\psi_{i_j}} \}_{j=1}^k$ as well as $\{ \alpha_i \}_{i=1}^M$ are classically known. Let $d_j \equiv \dim \ket{\psi_{i_j}}$ denote the dimension of $\ket{\psi_{i_j}} $ and $s_j$ denote the sparsity of $ \ket{\psi_{i_j}}$. Defining $d = \max \{ d_j \}_{j=1}^k$, and $s_{\max} = \max \{ s_j \}_{j=1}^k$, then $\mathcal{A}$ can be $\epsilon$-approximated block-encoded using a $\log mn$-qubits quantum circuit depth $\mathcal{O}\left( M \log (d) \kappa \log^2 \frac{\kappa}{\epsilon}   \right)$, $\mathcal{O}\left( k s_{\max}   \right)$ ancilla qubits, and a classical preprocessing of time $\mathcal{O}\left( \log d\right)$. If for all $j$, $d_j, s_j \in \mathcal{O}(1)$, then the circuit depth is $\mathcal{O}\left( ||\mathcal{A}||_FM \kappa \log^2 \frac{\kappa}{\epsilon}  \right)$, the number of extra ancila qubits is $\mathcal{O}\left( \log s_\mathcal{A} \right)$. 
\end{itemize}
\end{lemma}

\noindent
\textbf{A few remarks.} First, using QRAM, the quantum state must be sampled each time it is accessed within the quantum algorithm. Therefore, the QRAM access cost scales proportionally with the total runtime (i.e., it multiplies the overall cost). In contrast, the classical preprocessing cost required to construct the block-encoding quantum circuit is incurred only once. This cost does not scale with the number of algorithmic runs, it is simply added to the total cost as a one-time overhead. Next, as $s_{\mathcal{A}}$ is the number of nonzero entries of $\mathcal{A}$, it holds that $s_{\mathcal{A}} \leq mn$. Therefore, in the complexities above, it is safe to replace $\log s_{\mathcal{A}} \in \mathcal{O}(\log mn) $. In the discussion below, we will use this fact to avoid unnecessary extra notation. 

\section{Main Results}
\label{sec: overview}
In this section, we provide an overview of key objectives, key results as well as the underlying method to achieve them.

\subsection{Another type of state that admits efficient preparation procedure}
\label{sec: overviewstatepreparation}
The problem of state preparation is to construct a procedure that prepares $\ket{\Phi} = \sum_{i=1}^n a_i \ket{i}$ be the state of interest with classically known entries $\{a_i\}_{i=1}^n$ (assuming normalization $\sum_{i=1}^n |a_i|^2 =1$). This is a very common subroutine, especially in quantum machine learning algorithms \cite{schuld2019quantum, schuld2018supervised, schuld2020circuit, lloyd2013quantum}. In these contexts, typically one first needs to load the given classical data of interest (e.g., in supervised learning problems, one is provided with input feature vectors) into a quantum state, before executing a further algorithm. As such, the circuit complexity of this loading process is highly important, contributing significantly to the overall complexity of quantum machine learning algorithms. Although it has been shown that most quantum states would require an exponentially large (in the number of qubits) circuit to prepare, recent advances have revealed certain classes of state that admits efficient preparing complexity \cite{grover2000synthesis,grover2002creating,plesch2011quantum, schuld2018supervised, nakaji2022approximate,marin2023quantum,zoufal2019quantum,prakash2014quantum, zhang2022quantum, mcardle2022quantum}. 

To this end, we comment that the procedure underlying the last bullet point in Lemma \ref{lemma: blockencodingknownmatrix} (\textit{Structured case 2}) is based on the following result (which also appears in Lemma \ref{lemma: stateprepration} in Appendix \ref{sec: proofofLemmablockencodingknownmatrix}):
\begin{result}
If $\ket{\Phi}$ admits the following structure  $\ket{\Phi}= \sum_{i=1}^M \alpha_i \ket{\psi_{i_1}} \otimes \cdots \otimes \ket{\psi_{i_k}}$, and for all $i$, the entries of $ \{  \ket{\psi_{i_j}} \}_{j=1}^k$ as well as $\{ \alpha_i \}_{i=1}^M$ are classically known. Let $d_j \equiv \dim \ket{\psi_{i_j}}$ denote the dimension of $\ket{\psi_{i_j}} $ and $s_j$ denote the sparsity of $ \ket{\psi_{i_j}}$. Defining $d = \max \{ d_j \}_{j=1}^k$, and $s_{\max} = \max \{ s_j \}_{j=1}^k$, then it can be prepared using the $\log(n)$-qubits quantum circuit of depth $\mathcal{O}\left( M \log d  \right)$, $\mathcal{O}\left( k s_{\max}   \right)$ ancilla qubits, and a classical preprocessing of time $\mathcal{O}\left( \log n\right)$. If for all $j$, $d_j, s_j \in \mathcal{O}(1)$, then the circuit depth is $\mathcal{O}\left( M \right)$, the number of ancilla qubits is $\mathcal{O}\left( \log s \right)$.
\end{result}
The proof will be provided in the Appendix \ref{sec: improvingstatepreparation}. Here, we point out a few aspects. From the knowledge of entries of $\{  \ket{\psi_{i_j}} \}_{j=1}^k$, we can classically compute the entries of $\ket{\Phi}$. At the same time, if we know the entries $\{a_i\}_{i=1}^n$ of the general decomposition $\sum_{i=1}^n a_i \ket{i}$, then the method in \cite{zhang2022quantum} can be used to prepare $\ket{\Phi}$. However, this method uses the number of ancilla proportional to the number of nonzero entries of $\ket{\Phi}$, thus in practice, it is only effective $\ket{\Phi}$ is sparse. To circumvent this, we extend and adapt the method in \cite{zhang2022quantum} as follows. By exploiting the structure $\sum_{i=1}^M \alpha_i \ket{\psi_{i_1}} \otimes \cdots \otimes \ket{\psi_{i_k}} $ (as well as the classical knowledge), we first use \cite{zhang2022quantum} to construct the unitaries that prepare the states  $\{  \ket{\psi_{i_j}} \}_{j=1}^k$. Then via a combination of arithmetic tools from block-encoding, such as Lemma \ref{lemma: tensorproduct} \ref{lemma: sumencoding} \ref{lemma: scale}  (see Appendix \ref{sec: prelim}), the unitary that prepares the state $\ket{\Phi}$ can be obtained. As indicated above, the the number of ancilla qubits and circuit depth is significantly improved. Thus, it suggests one way to expand the capability of the technique in \cite{zhang2022quantum} in practice. At the same time, this result reveals one particular state structure that admits an efficient preparation circuit, adding another example to the existing literature \cite{grover2000synthesis,grover2002creating,plesch2011quantum, schuld2018supervised, nakaji2022approximate,marin2023quantum,zoufal2019quantum,prakash2014quantum, zhang2022quantum}.

\subsection{Principal component analysis}
\label{sec: overviewPCA}
Principal Component Analysis (PCA) is a dimensionality reduction technique widely used in statistics and machine learning. Let the dataset have $m$ points $\xbf^1,\xbf^2,...,\xbf^m$ where each $\xbf^i \in \Rbb^n$ is a $n$-dimensional vector. Let $\mathcal{X}$ be a $m\times n$ matrix with $\{\xbf^i\}_{i=1}^n$ be its columns. The centroid of given dataset is defined as $\mu = \sum_{i=1}^m ({\xbf^i}/{m})$. The covariance matrix is defined as: 
\begin{align}
    \mathcal{C}  = \sum_{i=1}^m \frac{1}{m}\xbf^i(\xbf^i)^T - \mu \mu^T = \frac{1}{m} \mathcal{X}^T \mathcal{X} - \mu\mu^T
\end{align}
The essential step of the PCA is to diagonalize the above matrix and find the largest eigenvalues with corresponding eigenvectors -- which are called principal components. The projection of given data points $\xbf^1,\xbf^2,...,\xbf^m$ along the top eigenvectors is the newly low-dimensional representation of these points, thus providing a compactification of the given data set. 

Our proposal for PCA is to combine Lemma \cref{lemma: blockencodingknownmatrix} and the (quantum) power method recently introduced in \cite{nghiem2023improved, nghiem2022quantum, chen2025quantum}. Briefly speaking, the power method aims to find the largest eigenvalue (and corresponding eigenvector) by performing a certain algebraic operation on the power of such a matrix and a randomly initiated vector. By an appropriate choice of the power, the desired eigenvalue can be approximated. A more detailed description can be found in the appendix. We recapitulate the main result in the following:
\begin{result}[PCA via Power Method]
    Given a dataset $\mathcal{X}$ with $m$ samples and $n$ features, with the covariance matrix $\mathcal{C}$ as defined above. Let the eigenvectors of $\mathcal{C}$ be $\ket{\lambda_1},\ket{\lambda_2},...,\ket{\lambda_n}$ and corresponding eigenvalues be $\lambda_1 > \lambda_2 > ... > \lambda_n$. Define $\Delta = \max_i \{  |\lambda_i - \lambda_{i+1}| \}_{i=1}^{r}$. The $r$ principal components $\ket{\lambda_1},\ket{\lambda_2},...,\ket{\lambda_r} $ of $\mathcal{X}$ can be obtained in complexity
\begin{equation}
    \mathcal{O}\left( \frac{\log(mn) ||\mathcal{C}||_F \log^r ({n}/{\epsilon}) \log^r ({1}/{\epsilon})}{\Delta^r\cdot \gamma^r}\right).
\end{equation}
The eigenvalues $\lambda_1,\lambda_2,...,\lambda_r $ can be estimated with complexity 
\begin{equation}
    \mathcal{O}\left( \frac{\log(mn)||\mathcal{C}||_F \log^r ({n}/{\epsilon}) \log^r ({1}/{\epsilon})}{\epsilon\cdot\Delta^r\cdot \gamma^r}\right).
\end{equation}
\end{result}

\begin{table*}[t]
    \centering
    \begin{tabular}{|c|l|}
    \hline
    \textbf{Method} & \textbf{Complexity} \\
    \hline
    Our approach (\cref{sec: PCApowermethod}) 
    & $ \mathcal{O}(\log(mn)\log^2(n/\epsilon)\log^2(1/\epsilon)/( \Delta^2 \gamma^2)$ \\
    \hline
    Ref.~\cite{lloyd2014quantum} 
    & $\mathcal{O}(\log(mn)/\epsilon^3)$ \\
    \hline
    Ref.~\cite{nghiem2025new} 
    & $\mathcal{O}(m\log(n)\log^6(n/\epsilon)/(\epsilon\Delta)^4)$ \\
    \hline
    Ref.~\cite{tang2021quantum} 
    & $\mathcal{O}(1/\epsilon^6 + \log(mn)/\epsilon^4)$ \\
    \hline
    \end{tabular}
    \caption{Table summarizing our result and relevant works of \cite{lloyd2014quantum, nghiem2025new, tang2021quantum}. As we can see, our first approach achieves exponential speed-up with respect to $1/\epsilon$ compared to previous works, meanwhile further exponential speed-up with respect to $m$ (the number of sample data) compared to \cite{nghiem2025new}. We note that in the works \cite{lloyd2014quantum,nghiem2025new, tang2021quantum}, the Frobenius norm $||\mathcal{C}||_F$ is assumed to be 1, so we import such condition in the comparison.  }
    \label{tab: pca}
\end{table*}

To demonstrate the exponential improvement of our proposal compared to existing works, we provide the following table summarizing the relevant complexity in finding the top $2$ eigenvalues/eigenvectors of covariance matrix $\mathcal{C}$ defined above.

\subsection{System of linear algebraic equations}
\label{sec: overviewlinearsystem}
A $n \times n$ linear system is defined as $A\xbf = \textbf{b}$ with $A$ is some $n\times n$ matrix and $\textbf{b}$ is $n$-dimensional vector. The goal is to find $\xbf$ that satisfies such an equation. In quantum context, the goal is to obtain the quantum state $\ket{\xbf}$ corresponding to the solution $\xbf$. Similar to existing works \cite{harrow2009quantum, childs2017quantum}, assuming $\textbf{b}$ is normalized for simplicity plus a known preparation procedure. Provided the classical knowledge of $A$ and $\textbf{b}$, we can leverage Lemma \ref{lemma: blockencodingknownmatrix} to block-encode $A$, and then use any of existing methods \cite{harrow2009quantum, childs2017quantum, gilyen2019quantum} to invert $A$, obtaining the block-encoding of $A^{-1}/\kappa$. Then we apply this unitary block-encoding to the state $\ket{\bf 0}\textbf{b}$, followed by measuring the ancilla and post-select $\ket{\bf 0}$, which results in $\ket{\xbf}$. We then achieve the following result:
\noindent
\begin{result}[Quantum Linear Solving Algorithm]
    Let the linear system be $A\xbf = \textbf{b}$ where $A$ is an $s$-sparse, Hermitian matrix of size $n \times n$, with condition number $\kappa$, and  $\textbf{b}$ is assumed to be unit vector.  Then there is a quantum algorithm outputting the state $\ket{\xbf } \varpropto  A^{-1}\textbf{b} $ in complexity
    \begin{equation}
        \mathcal{O}\left( ||A||_F \kappa^2 \log(sn) \log^2\left( \frac{\kappa^2}{\epsilon}\right) \log^2 \frac{1}{\epsilon}  \right).
    \end{equation}
In the case $A$ is positive-semidefinite, the complexity is:
    \begin{equation}
        \mathcal{O}\left( ||A||_F \kappa^3 \log (sn) \log^2 \left(\frac{\kappa^{3/2}}{\epsilon}\right) \right).
    \end{equation}
    The number of ancilla qubits and classical pre-processing time follows Lemma \ref{lemma: blockencodingknownmatrix}, which depends on the structure of $A$.
\end{result}
 
We provide Table~\ref{tab: lineareq} for comparison of our new proposals versus existing results in the context. To comment, for those matrices admitting structure as in Lemma \ref{lemma: blockencodingknownmatrix} and having $||A||_F$ upper bounded (or independent of sparsity $s$), our result achieves exponential improvement with respect to $s$, compared to \cite{harrow2009quantum, childs2017quantum}. In the dense regime $s \sim n$, our method admits a superpolynomial speed-up in the inverse of error tolerance compared to \cite{wossnig2018quantum}. To our knowledge, this is the first quantum algorithm that achieves polylogarithmic scaling in $1/\epsilon$ in solving dense linear equation. 

\begin{table*}
    \centering
    {
    \begin{tabular}{|c|l|c|}
    \hline
    \textbf{Method} & \textbf{Complexity}  & \textbf{Need of QRAM} \\
    \hline
    Our method 
    & $\mathcal{O}(||A||_F \kappa^2 \log(sn) \log^2({\kappa^2}/{\epsilon}) \log^2({1}/{\epsilon}) )$ 
    & No \\
    \hline
    Ref.~\cite{nghiem2025new2} 
    & $\mathcal{O}( s^2 (s^2 + \log(n)) \log^{3.5}({s}/{\epsilon})/\epsilon)$
    & No \\
    \hline
    Ref.~\cite{harrow2009quantum} 
    & $\mathcal{O}(s \kappa \log (n) / \epsilon)$ 
    & Yes \\
    \hline
    Ref.~\cite{childs2017quantum} 
    & $\mathcal{O}( s \kappa^2 \log^{2.5} ({\kappa}/{\epsilon}) (\log(n) + \log^{2.5}({\kappa}/{\epsilon})))$ 
    & Yes \\
    \hline
    Ref.~\cite{clader2013preconditioned} 
    & $\mathcal{O}(s^7 \log (n) / \epsilon)$ 
    & Yes \\
    \hline
    Ref.~\cite{wossnig2018quantum} & $\mathcal{O}\left( \kappa^2 ||A||_F \frac{1}{\epsilon}\rm polylog(n)   \right) $ & Yes \\
    \hline
    \end{tabular}
    }
    \caption{Table summarizing our result and relevant works on system of linear algebraic equations. Our result achieves exponential improvement with respect to $s$ -- the sparsity of $A$.}
    \label{tab: lineareq}
\end{table*}

\subsection{Quantum simulation}
\label{sec: overviewquantumsimulation}
The dynamic of a quantum system obeys Schr\"odinger's equation (we set $\hbar = 1$):
\begin{align}
    \frac{\partial \ket{\psi}}{\partial t} = -i H \ket{\psi}.
\end{align}

The goal is to construct the so-called evolution operator $\exp(-i Ht)$, up to some additive error. In this case, if we know the rows of $H$ explicitly, then we can use Lemma \ref{lemma: blockencodingknownmatrix} to construct the block encoding of $H/||H||_F$, from which the simulation $\exp(-i H t)$ can be constructed in a manner similar to~\cite{low2017optimal,low2019hamiltonian, gilyen2019quantum}, as we approximate $\exp(-i H t)$ by the Jacobi-Anger polynomial expansion, and use \cref{lemma: theorem56} to transform $H$ into such a polynomial. This completes a new quantum simulation algorithm with the input model being the classical knowledge of Hamiltonian of interest, and we summarize as follows.
\noindent
\begin{result}[Quantum Simulation of Classically Known Hamiltonian]
     Provided the classical knowledge of the time-independent Hamiltonian $H$ of size $n \times n$, the evolution operator $\exp(-i Ht )$ can be constructed, within a precision $\epsilon$, using a quantum circuit of complexity 
     \begin{equation}
        \mathcal{O}\left(\log (sn) \log^2 \left(\frac{1}{\epsilon}\right) \left( t ||H||_{F} + \frac{\log (1/\epsilon)}{\log(e + \log(1/\epsilon)/t))}  \right) \right)    
     \end{equation}
     where $s$ is the sparsity of $H$. 
\end{result}
Despite working on a different model to ours, we remark that most existing results in quantum simulation have complexity being linear in the sparsity $s$ \cite{aharonov2003adiabatic,berry2007efficient,berry2012black,berry2014high,berry2015hamiltonian,berry2015simulating, childs2010relationship, low2017optimal,low2019hamiltonian}. Thus, if the Frobenius norm $||H||_F$ is bounded (independent of $s$), our result above indicates that there is an exponential improvement in $s$, highlighting the capability of our method.


\subsection{Preparing ground state}
\label{sec: preparinggroundstate}
Although technically we can use the quantum PCA algorithm above to find the ground state, there is another, more efficient way of doing so. 
Again, from the classical knowledge of $H$, we can block-encode it via Lemma \ref{lemma: blockencodingknownmatrix}. Let $\ket{\psi}$ be some initially random state, then provided that there is a finite gap, we have the following: 
\begin{align}
   \lim_{t\longrightarrow \infty} \frac{1}{\bra{\psi} e^{-2Ht} \ket{\psi} } e^{-Ht} \ket{\psi} = \ket{\lambda_{\rm ground}}.
\end{align}
The above formula underlies the so-called quantum imaginary time evolution, which is a very popular classical method for finding the ground state. From the classical knowledge of $H$, Lemma \ref{lemma: blockencodingknownmatrix} allows us to block-encode $H$. Our algorithm for preparing the ground state of $H$ obeys the following procedure. From the block-encoding of $H$, we use Lemma \ref{lemma: theorem56} to transform $H$ into the block-encoding of $\approx \exp\Big( - \beta \big( \Ibb -  \frac{1}{2}(\Ibb - H) \big) \Big) = \exp\Big( - \frac{1}{2} \beta \big( \Ibb +  H  \big)   \Big)$. We then take the block encoding of $\exp\Big( - t \big( \Ibb +  H  \big)   \Big) $ and apply to $\ket{\bf 0} \ket{\psi}$. Measuring the ancilla and post-select on $\ket{\bf 0}$, we obtain the normalization of $\exp\Big( -t \big( \Ibb + H \big) \Big) \ket{\psi} $ , which is exactly $\ket{\psi_t}$. 
In the Appendix \ref{sec: proofite}, we will prove that by choosing $t = \mathcal{O}\Big( \frac{1}{\Delta} \log \frac{n }{\epsilon \gamma^2 } \Big)$ (where $\Delta$ is the gap, or the (absolute) difference between ground state and first excited state energy, $n$ is the dimension, $\gamma = |\braket{\psi, \lambda_{\rm ground}} | $ is the overlap between initial state $\ket{\psi}$ and the targeted ground state, and $\epsilon$ is the error tolerance), then it holds that $||\ket{\psi_t} - \ket{\lambda_{\rm ground}}  || \leq \epsilon$. However, there is an issue with this approach, as the probability of measuring $\ket{\bf 0}$ on the ancilla is $\bra{\psi} \exp(-t (\Ibb + H) t)\ket{\psi}$ can be very small. To overcome this, instead of measuring the ancilla, we add an extra step that incorporates Lemma \ref{lemma: improveddme} and \ref{lemma: exponentialapproximation} to ``boost'' the exponentially small term. For our purpose, we state our main result for ground state preparation as follows and defer the full details to the Appendix \ref{sec: groundstatepreparation}.
\noindent
\begin{result}[Ground State Preparation]
     Provided the classical knowledge of the time-independent Hamiltonian $H$ of size $n \times n$ and an efficient circuit to prepare an arbitrary state $\ket{\psi}$, its ground state $\ket{\lambda_{\rm ground}}$ can be prepared, up to an $\epsilon$ precision, in complexity:
$$ \mathcal{O}\Big( \frac{1}{\gamma} ||H||_F  \sqrt{\frac{1}{\Delta} \log \big(\frac{n}{\epsilon \gamma } }\big) \log (n) \log^{5/2} \frac{1}{\epsilon}  \Big)  $$
where $\Delta$ is the gap between first excited state and ground state energy; and $\gamma = |\braket{\psi, \lambda_{\rm ground}} | $. 
\end{result}

Provided that $||H||_F$ is bounded, our quantum imaginary time evolution algorithm achieves an quadratic improvement in $\frac{1}{\Delta}$ compared to~\cite{dong2022ground, lin2020near}, and an exponential improvement in (almost) all parameters compared to the relevant work~\cite{motta2020determining}, which also is based on imaginary time evolution. 

\subsection{Quantum data fitting}
\label{sec: overviewquantumdatafitting}
Data fitting is a pivotal tool in quantitative science. Typically, given that the dataset are $\{x_i, y_i\}_{i=1}^M$ where $x_i \in \Rbb^N$ is the data of interest, $y_i \in \Rbb$ is the scalar value, we aim to fit a function of the form $ f(x,\lambda) = \sum_{j=1}^N f_j(x)  \lambda_j$ where $f_j: \Rbb^N \longrightarrow \Rbb$ is a continuous function, and $\lambda = (\lambda_1,\lambda_2,...,\lambda_N)^T$. A common way of doing so is minimizing the so-called cost function $C = \sum_{i=1}^M |f(x_i,\lambda) - y_i|^2 $. As pointed out in \cite{wiebe2012quantum}, the desired parameters $\lambda$ can be found as:
\begin{align}
    \lambda = ( F^\dagger F )^{-1} F^\dagger y,
\end{align}
where the matrix $F$ is defined as $F_{ij} = f_j(x_i)$ and $y = (y_1,y_2,...,y_M)^T$. 

For simplicity, we assume $y$ has unit norm, and its entries are known. Besides, assuming the classical knowledge of entries of $F$, then the block-encoding of $F^\dagger F$ can be obtained via Lemma \ref{lemma: blockencodingknownmatrix}. The block-encoding of $F^\dagger$ can also be obtained with a small modification (see Appendix \ref{sec: datafitting}). Then we can use the Lemma \ref{lemma: product} to construct the block-encoding of $( F^\dagger F )^{-1} F^\dagger $, and apply it to the state $\ket{\bf 0}\ket{y}$. The resultant state contains the fit parameters vector $\lambda$ as a sub-vector. In principle, we can measure the ancilla and post-select $\ket{\bf 0}$ to obtain the quantum state $\ket{\lambda} \varpropto \lambda$. However, in practice, such the step is not necessary. The reason is that, once we have $\lambda$, it is typically desired to use it to predict some unseen data $\tilde{x}$, i.e., estimating $f(\tilde{x},\lambda)$. In the Appendix \ref{sec: datafitting}, we will show that we can leverage the the block-encoding of $( F^\dagger F )^{-1} F^\dagger $ and also the state $\ket{\bf 0}\ket{y}$ directly to estimate $f(\tilde{x},\lambda)$, which help reduce a significant amount of time for measurement and post-selection. We summarize our quantum data fitting algorithm as well as predicting unseen input in the following.
\begin{result}[Quantum Data Fitting]
    Provided the data set $\{x_i, y_i\}_{i=1}^M$ and a fit model $f(x,\lambda) = \sum_{ j=1}^N f_j(x) \lambda_j$. Then the state which includes $\lambda = (\lambda_1,\lambda_2,...,\lambda_N)^T$ as a sub-vector that minimizes the cost function $C = \sum_{i=1}^M |f(x_i,\lambda) - y_i|^2 $ can be obtained in complexity
    $$ \mathcal{O}\left(  ||F||_F \log(MN) \kappa_F^2 \log^2 \frac{1}{\epsilon}  \right)$$
    In particular, from such the state and a given unseen input $\tilde{x}$, the value $f(\tilde{x},\lambda)$ can be estimated with a total complexity $\mathcal{O}\left(  ||F||_F \log(MN) \kappa_F^2 \frac{1}{\epsilon}\log^2 \frac{1}{\epsilon}  \right)$.
\end{result}
We recall that the complexity of \cite{wiebe2012quantum} is $ \mathcal{O}\left( \frac{\kappa_F^3 s^6}{\epsilon}  \log MN \right)$. Therefore, our result achieves a superpolynomial speed-up in the inverse of the error tolerance, $\frac{1}{\epsilon}$, a polynomial speed-up in $\kappa_F$. In addition, we do not require the oracle/black-box access to the entries of $F$. Now we point out that the polynomial dependence on sparsity $s$ is a severe limitation of the method in \cite{wiebe2012quantum}.

\smallskip\noindent
\textbf{An example with univariate polynomial. } For simplicity, suppose that the fit function is univariate polynomial $f(x) = \lambda_1 x + \lambda_2 x^2 + \lambda_3x^3 + \lambda_4 x^4 + ... + \lambda_N x^N $. The data set is $\{ x_i,y_i\}_{i=1}^M$. In this case, we have that $f_1(x) = x, f_2(x) = x^2, f_3(x) = x^3, f_4(x) = x^4, ... , f_N(x) = x^N$. Recall that the matrix $F$ is defined by $F_{ij} = f_j(x_i)$, which is $F_{ij} = x_i^j$ in this case. Therefore, it is only zero when $x_i = 0$. In other words, the matrix $F$ has low sparsity $s$ only when many of the values among $\{x_i\}_{i=1}^M$ are zero -- an unlikely possibility. Thus, in practice, the value of $s$ is typically as large as $M$. It means that the method of~\cite{wiebe2012quantum} achieves polynomial scaling in $M$ -- the number of data points. At the same time, it shows that our method achieves exponential speedup compared to~\cite{wiebe2012quantum} in the number of data points.

\section{Outlook and Conclusion}
In this work, we have proposed variants of some quantum algorithms for a wide range of problems. Our algorithms are largely motivated by the caveats faced by existing methods, which were identified and improved in our work. More specifically, for the PCA, we have pointed out that prior constructions suffered from both strong input assumption and poor scaling in certain parameters, which severely limit their impact and potential realization, to some extent. We introduced a way to block-encode the matrix of interest and combine with the power method, which is a simple yet highly efficient tool for dealing with top eigenvalues/eigenvectors. As we have seen, the top eigenvalues/eigenvectors, also called the principal components of the covariance matrix, could be revealed within (poly)logarithmic complexity in all parameters. This approach surpasses the previous results~\cite{lloyd2014quantum} and~\cite{nghiem2025new} in terms of complexity scaling in the inverse of error tolerance. We then extend the technique from PCA to the context of solving linear equations and show that a highly efficient quantum linear solver can be achieved. The complexity turns out to scale (poly)logarithmically in most parameters. This is an exponential improvement over the previous results \cite{harrow2009quantum, childs2017quantum, nghiem2025new2, clader2013preconditioned, wossnig2018quantum}. In particular, we have shown that the techniques of our new QPCA/QLSA can be applied to the context of quantum simulation, preparing ground state, and data fitting. All results, especially our quantum data fitting algorithm, exhibit significant improvement, in terms of both complexity and applicability, compared to existing algorithms. At the same time, we have seen that our quantum algorithms can be executed without oracle/black-box access to the classical information, thus implying a message that quantum computers can be advantageous without resorting to a strong input assumption. This assumption has been a major roadblock to the realization of quantum advantage and imposing great concern on the practicability of quantum algorithm as as whoel. Our results thus pave a new route for quantum computer application, highlighting its potential toward practical problems.

\section*{Acknowledgement}
We acknowledge the discussion with Junseo Lee, Myeongjin Shin and Kabgyun Jeong. This work was partly supported by the U.S. Department of Energy, Office
of Science, National Quantum Information Science Research
Centers, Co-design Center for Quantum Advantage (C2QA)
under Contract No. DE-SC0012704. N.A.N. and T.-C.W. also acknowledge support from the Center for Distributed Quantum Processing, Stony Brook University. Part of this work is done when N.A.N. is an intern at QuEra Computing Inc. 


\onecolumn
\addcontentsline{toc}{section}{References}
\bibliographystyle{alpha}
\bibliography{citation}

@article{guo2024nonlinear,
  title={Nonlinear transformation of complex amplitudes via quantum singular value transformation},
  author={Guo, Naixu and Mitarai, Kosuke and Fujii, Keisuke},
  journal={Physical Review Research},
  volume={6},
  number={4},
  pages={043227},
  year={2024},
  publisher={APS}
}

@inproceedings{chen2025quantum,
  title={A quantum speed-up for approximating the top eigenvectors of a matrix},
  author={Chen, Yanlin and Gily{\'e}n, Andr{\'a}s and de Wolf, Ronald},
  booktitle={Proceedings of the 2025 Annual ACM-SIAM Symposium on Discrete Algorithms (SODA)},
  pages={994--1036},
  year={2025},
  organization={SIAM}
}

@article{bellante2022quantum,
  title={Quantum algorithms for SVD-based data representation and analysis},
  author={Bellante, Armando and Luongo, Alessandro and Zanero, Stefano},
  journal={Quantum Machine Intelligence},
  volume={4},
  number={2},
  pages={20},
  year={2022},
  publisher={Springer}
}

@inproceedings{bellante2023quantum,
  title={Quantum eigenfaces: Linear feature mapping and nearest neighbor classification with outlier detection},
  author={Bellante, Armando and Bonvini, William and Vanerio, Stefano and Zanero, Stefano},
  booktitle={2023 IEEE International Conference on Quantum Computing and Engineering (QCE)},
  volume={1},
  pages={196--207},
  year={2023},
  organization={IEEE}
}

@book{nesterov2013introductory,
  title={Introductory lectures on convex optimization: A basic course},
  author={Nesterov, Yurii},
  volume={87},
  year={2013},
  publisher={Springer Science \& Business Media}
}

@article{boyd2004convex,
  title={Convex optimization},
  author={Boyd, Stephen},
  journal={Cambridge UP},
  year={2004}
}

@article{nghiem2025quantum1,
  title={Quantum Computer Does Not Need Coherent Quantum Access for Advantage},
  author={Nghiem, Nhat A},
  journal={arXiv preprint arXiv:2503.02515},
  year={2025}
}

@article{mcardle2022quantum,
  title={Quantum state preparation without coherent arithmetic},
  author={McArdle, Sam and Gilyen, Andras and Berta, Mario},
  journal={arXiv preprint arXiv:2210.14892},
  year={2022}
}

@article{nghiem2025new2,
  title={New Quantum Algorithm For Solving Linear System of Equations},
  author={Nghiem, Nhat A},
  journal={arXiv preprint arXiv:2502.13630},
  year={2025}
}

@article{zlokapa2021quantum,
  title={A quantum algorithm for training wide and deep classical neural networks},
  author={Zlokapa, Alexander and Neven, Hartmut and Lloyd, Seth},
  journal={arXiv preprint arXiv:2107.09200},
  year={2021}
}

@article{rodriguez2025quantum,
  title={Quantum approximated cloning-assisted density matrix exponentiation},
  author={Rodriguez-Grasa, Pablo and Ibarrondo, Ruben and Gonzalez-Conde, Javier and Ban, Yue and Rebentrost, Patrick and Sanz, Mikel},
  journal={Physical Review Research},
  volume={7},
  number={1},
  pages={013264},
  year={2025},
  publisher={APS}
}

@article{rattew2023non,
  title={Non-linear transformations of quantum amplitudes: Exponential improvement, generalization, and applications},
  author={Rattew, Arthur G and Rebentrost, Patrick},
  journal={arXiv preprint arXiv:2309.09839},
  year={2023}
}

@article{nghiem2024improved,
  title={Improved Quantum Power Method and Numerical Integration Using Quantum Singular Value Transformation},
  author={Nghiem, Nhat A and Sukeno, Hiroki and Zhang, Shuyu and Wei, Tzu-Chieh},
  journal={arXiv preprint arXiv:2407.11744},
  year={2024}
}

@book{golub2013matrix,
  title={Matrix computations},
  author={Golub, Gene H and Van Loan, Charles F},
  year={2013},
  publisher={JHU press}
}

@inproceedings{aharonov2003adiabatic,
  title={Adiabatic quantum state generation and statistical zero knowledge},
  author={Aharonov, Dorit and Ta-Shma, Amnon},
  booktitle={Proceedings of the thirty-fifth annual ACM symposium on Theory of computing},
  pages={20--29},
  year={2003}
}

@article{chakraborty2018power,
  title={The power of block-encoded matrix powers: improved regression techniques via faster Hamiltonian simulation},
  author={Chakraborty, Shantanav and Gily{\'e}n, Andr{\'a}s and Jeffery, Stacey},
  journal={arXiv preprint arXiv:1804.01973},
  year={2018}
}

@article{childs2018toward,
  title={Toward the first quantum simulation with quantum speedup},
  author={Childs, Andrew M and Maslov, Dmitri and Nam, Yunseong and Ross, Neil J and Su, Yuan},
  journal={Proceedings of the National Academy of Sciences},
  volume={115},
  number={38},
  pages={9456--9461},
  year={2018},
  publisher={National Academy of Sciences}
}

@article{berry2015simulating,
  title={Simulating Hamiltonian dynamics with a truncated Taylor series},
  author={Berry, Dominic W and Childs, Andrew M and Cleve, Richard and Kothari, Robin and Somma, Rolando D},
  journal={Physical review letters},
  volume={114},
  number={9},
  pages={090502},
  year={2015},
  publisher={APS}
}

@article{tran2021faster,
  title={Faster digital quantum simulation by symmetry protection},
  author={Tran, Minh C and Su, Yuan and Carney, Daniel and Taylor, Jacob M},
  journal={PRX Quantum},
  volume={2},
  number={1},
  pages={010323},
  year={2021},
  publisher={APS}
}

@article{childs2021theory,
  title={Theory of trotter error with commutator scaling},
  author={Childs, Andrew M and Su, Yuan and Tran, Minh C and Wiebe, Nathan and Zhu, Shuchen},
  journal={Physical Review X},
  volume={11},
  number={1},
  pages={011020},
  year={2021},
  publisher={APS}
}

@article{childs2019nearly,
  title={Nearly optimal lattice simulation by product formulas},
  author={Childs, Andrew M and Su, Yuan},
  journal={Physical review letters},
  volume={123},
  number={5},
  pages={050503},
  year={2019},
  publisher={APS}
}

@article{zhao2022hamiltonian,
  title={Hamiltonian simulation with random inputs},
  author={Zhao, Qi and Zhou, You and Shaw, Alexander F and Li, Tongyang and Childs, Andrew M},
  journal={Physical Review Letters},
  volume={129},
  number={27},
  pages={270502},
  year={2022},
  publisher={APS}
}

@article{childs2020quantum,
  title={Quantum spectral methods for differential equations},
  author={Childs, Andrew M and Liu, Jin-Peng},
  journal={Communications in Mathematical Physics},
  volume={375},
  number={2},
  pages={1427--1457},
  year={2020},
  publisher={Springer}
}

@article{berry2017quantum,
  title={Quantum algorithm for linear differential equations with exponentially improved dependence on precision},
  author={Berry, Dominic W and Childs, Andrew M and Ostrander, Aaron and Wang, Guoming},
  journal={Communications in Mathematical Physics},
  volume={356},
  pages={1057--1081},
  year={2017},
  publisher={Springer}
}

@article{berry2012black,
  title={Black-box Hamiltonian simulation and unitary implementation},
  author={Berry, Dominic W and Childs, Andrew M},
  journal={Quantum Information and Computation}, volume={12}, pages={29-62},
  year={2009}
}

@book{schuld2018supervised,
  title={Supervised learning with quantum computers},
  author={Schuld, Maria and Petruccione, Francesco},
  volume={17},
  year={2018},
  publisher={Springer}
}

@article{schuld2020circuit,
  title={Circuit-centric quantum classifiers},
  author={Schuld, Maria and Bocharov, Alex and Svore, Krysta M and Wiebe, Nathan},
  journal={Physical Review A},
  volume={101},
  number={3},
  pages={032308},
  year={2020},
  publisher={APS}
}

@article{shor1999polynomial,
  title={Polynomial-time algorithms for prime factorization and discrete logarithms on a quantum computer},
  author={Shor, Peter W},
  journal={SIAM review},
  volume={41},
  number={2},
  pages={303--332},
  year={1999},
  publisher={SIAM}
}

@inproceedings{grover1996fast,
  title={A fast quantum mechanical algorithm for database search},
  author={Grover, Lov K},
  booktitle={Proceedings of the twenty-eighth annual ACM symposium on Theory of computing},
  pages={212--219},
  year={1996}
}

@article{lloyd2013quantum,
  title={Quantum algorithms for supervised and unsupervised machine learning},
  author={Lloyd, Seth and Mohseni, Masoud and Rebentrost, Patrick},
  journal={arXiv preprint arXiv:1307.0411},
  year={2013}
}

@article{schuld2019quantum,
  title={Quantum machine learning in feature Hilbert spaces},
  author={Schuld, Maria and Killoran, Nathan},
  journal={Physical review letters},
  volume={122},
  number={4},
  pages={040504},
  year={2019},
  publisher={APS}
}

@article{harrow2009quantum,
  title={Quantum algorithm for linear systems of equations},
  author={Harrow, Aram W and Hassidim, Avinatan and Lloyd, Seth},
  journal={Physical review letters},
  volume={103},
  number={15},
  pages={150502},
  year={2009},
  publisher={APS}
}

@article{plesch2011quantum,
  title={Quantum-state preparation with universal gate decompositions},
  author={Plesch, Martin and Brukner, {\v{C}}aslav},
  journal={Physical Review A},
  volume={83},
  number={3},
  pages={032302},
  year={2011},
  publisher={APS}
}

@book{prakash2014quantum,
  title={Quantum algorithms for linear algebra and machine learning},
  author={Prakash, Anupam},
  year={2014},
  publisher={University of California, Berkeley}
}

@article{brassard2002quantum,
  title={Quantum amplitude amplification and estimation},
  author={Brassard, Gilles and Hoyer, Peter and Mosca, Michele and Tapp, Alain},
  journal={Contemporary Mathematics},
  volume={305},
  pages={53--74},
  year={2002},
  publisher={Providence, RI; American Mathematical Society; 1999}
}

@incollection{feynman2018simulating,
  title={Simulating physics with computers},
  author={Feynman, Richard P},
  booktitle={Feynman and computation},
  pages={133--153},
  year={2018},
  publisher={CRC Press}
}

@article{wiebe2012quantum,
  title={Quantum algorithm for data fitting},
  author={Wiebe, Nathan and Braun, Daniel and Lloyd, Seth},
  journal={Physical review letters},
  volume={109},
  number={5},
  pages={050505},
  year={2012},
  publisher={APS}
}

@article{tang2021quantum,
  title={Quantum principal component analysis only achieves an exponential speedup because of its state preparation assumptions},
  author={Tang, Ewin},
  journal={Physical Review Letters},
  volume={127},
  number={6},
  pages={060503},
  year={2021},
  publisher={APS}
}

@article{lloyd2014quantum,
  title={Quantum principal component analysis},
  author={Lloyd, Seth and Mohseni, Masoud and Rebentrost, Patrick},
  journal={Nature Physics},
  volume={10},
  number={9},
  pages={631--633},
  year={2014},
  publisher={Nature Publishing Group}
}

@article{gordon2022covariance,
  title={Covariance matrix preparation for quantum principal component analysis},
  author={Gordon, Max Hunter and Cerezo, Marco and Cincio, Lukasz and Coles, Patrick J},
  journal={PRX Quantum},
  volume={3},
  number={3},
  pages={030334},
  year={2022},
  publisher={APS}
}

@article{nghiem2025new,
  title={New Quantum Algorithm for Principal Component Analysis},
  author={Nghiem, Nhat A},
  journal={arXiv preprint arXiv:2501.07891},
  year={2025}
}

@article{zhang2022quantum,
  title={Quantum state preparation with optimal circuit depth: Implementations and applications},
  author={Zhang, Xiao-Ming and Li, Tongyang and Yuan, Xiao},
  journal={Physical Review Letters},
  volume={129},
  number={23},
  pages={230504},
  year={2022},
  publisher={APS}
}

@article{grover2000synthesis,
  title={Synthesis of quantum superpositions by quantum computation},
  author={Grover, Lov K},
  journal={Physical review letters},
  volume={85},
  number={6},
  pages={1334},
  year={2000},
  publisher={APS}
}

@article{grover2002creating,
  title={Creating superpositions that correspond to efficiently integrable probability distributions},
  author={Grover, Lov and Rudolph, Terry},
  journal={arXiv preprint quant-ph/0208112},
  year={2002}
}

@inproceedings{nesterov1983method,
  title={A method for solving the convex programming problem with convergence rate O (1/k2)},
  author={Nesterov, Yurii},
  booktitle={Dokl akad nauk Sssr},
  volume={269},
  pages={543},
  year={1983}
}

@article{nakaji2022approximate,
  title={Approximate amplitude encoding in shallow parameterized quantum circuits and its application to financial market indicators},
  author={Nakaji, Kouhei and Uno, Shumpei and Suzuki, Yohichi and Raymond, Rudy and Onodera, Tamiya and Tanaka, Tomoki and Tezuka, Hiroyuki and Mitsuda, Naoki and Yamamoto, Naoki},
  journal={Physical Review Research},
  volume={4},
  number={2},
  pages={023136},
  year={2022},
  publisher={APS}
}

@article{marin2023quantum,
  title={Quantum algorithms for approximate function loading},
  author={Marin-Sanchez, Gabriel and Gonzalez-Conde, Javier and Sanz, Mikel},
  journal={Physical Review Research},
  volume={5},
  number={3},
  pages={033114},
  year={2023},
  publisher={APS}
}

@article{zoufal2019quantum,
  title={Quantum generative adversarial networks for learning and loading random distributions},
  author={Zoufal, Christa and Lucchi, Aur{\'e}lien and Woerner, Stefan},
  journal={npj Quantum Information},
  volume={5},
  number={1},
  pages={103},
  year={2019},
  publisher={Nature Publishing Group UK London}
}

@article{childs2017quantum,
  title={Quantum algorithm for systems of linear equations with exponentially improved dependence on precision},
  author={Childs, Andrew M and Kothari, Robin and Somma, Rolando D},
  journal={SIAM Journal on Computing},
  volume={46},
  number={6},
  pages={1920--1950},
  year={2017},
  publisher={SIAM}
}

@article{berry2007efficient,
  title={Efficient quantum algorithms for simulating sparse Hamiltonians},
  author={Berry, Dominic W and Ahokas, Graeme and Cleve, Richard and Sanders, Barry C},
  journal={Communications in Mathematical Physics},
  volume={270},
  number={2},
  pages={359--371},
  year={2007},
  publisher={Springer}
}

@article{childs2010relationship,
  title={On the relationship between continuous-and discrete-time quantum walk},
  author={Childs, Andrew M},
  journal={Communications in Mathematical Physics},
  volume={294},
  number={2},
  pages={581--603},
  year={2010},
  publisher={Springer}
}

@article{motta2020determining,
  title={Determining eigenstates and thermal states on a quantum computer using quantum imaginary time evolution},
  author={Motta, Mario and Sun, Chong and Tan, Adrian TK and O’Rourke, Matthew J and Ye, Erika and Minnich, Austin J and Brandao, Fernando GSL and Chan, Garnet Kin-Lic},
  journal={Nature Physics},
  volume={16},
  number={2},
  pages={205--210},
  year={2020},
  publisher={Nature Publishing Group UK London}
}

@article{lin2020near,
  title={Near-optimal ground state preparation},
  author={Lin, Lin and Tong, Yu},
  journal={Quantum},
  volume={4},
  pages={372},
  year={2020},
  publisher={Verein zur F{\"o}rderung des Open Access Publizierens in den Quantenwissenschaften}
}

@article{ding2024single,
  title={Single-ancilla ground state preparation via Lindbladians},
  author={Ding, Zhiyan and Chen, Chi-Fang and Lin, Lin},
  journal={Physical Review Research},
  volume={6},
  number={3},
  pages={033147},
  year={2024},
  publisher={APS}
}

@article{dong2022ground,
  title={Ground-state preparation and energy estimation on early fault-tolerant quantum computers via quantum eigenvalue transformation of unitary matrices},
  author={Dong, Yulong and Lin, Lin and Tong, Yu},
  journal={PRX quantum},
  volume={3},
  number={4},
  pages={040305},
  year={2022},
  publisher={APS}
}

@article{ge2019faster,
  title={Faster ground state preparation and high-precision ground energy estimation with fewer qubits},
  author={Ge, Yimin and Tura, Jordi and Cirac, J Ignacio},
  journal={Journal of Mathematical Physics},
  volume={60},
  number={2},
  year={2019},
  publisher={AIP Publishing}
}

@article{bespalova2021hamiltonian,
  title={Hamiltonian operator approximation for energy measurement and ground-state preparation},
  author={Bespalova, Tatiana A and Kyriienko, Oleksandr},
  journal={PRX Quantum},
  volume={2},
  number={3},
  pages={030318},
  year={2021},
  publisher={APS}
}

@inproceedings{berry2015hamiltonian,
  title={Hamiltonian simulation with nearly optimal dependence on all parameters},
  author={Berry, Dominic W and Childs, Andrew M and Kothari, Robin},
  booktitle={2015 IEEE 56th annual symposium on foundations of computer science},
  pages={792--809},
  year={2015},
  organization={IEEE}
}

@article{nghiem2022quantum,
  title={Quantum Algorithm For Estimating Eigenvalue},
  author={Nghiem, Nhat A and Wei, Tzu-Chieh},
  journal={arXiv preprint arXiv:2211.06179},
  year={2022}
}

@inproceedings{gilyen2019quantum,
  title={Quantum singular value transformation and beyond: exponential improvements for quantum matrix arithmetics},
  author={Gily{\'e}n, Andr{\'a}s and Su, Yuan and Low, Guang Hao and Wiebe, Nathan},
  booktitle={Proceedings of the 51st Annual ACM SIGACT Symposium on Theory of Computing},
  pages={193--204},
  year={2019}
}

@article{friedman1998error,
  title={Error bounds on the power method for determining the largest eigenvalue of a symmetric, positive definite matrix},
  author={Friedman, Joel},
  journal={Linear algebra and its applications},
  volume={280},
  number={2-3},
  pages={199--216},
  year={1998},
  publisher={Elsevier}
}

@article{lloyd1996universal,
  title={Universal quantum simulators},
  author={Lloyd, Seth},
  journal={Science},
  volume={273},
  number={5278},
  pages={1073--1078},
  year={1996},
  publisher={American Association for the Advancement of Science}
}

@article{low2017optimal,
  title={Optimal Hamiltonian simulation by quantum signal processing},
  author={Low, Guang Hao and Chuang, Isaac L},
  journal={Physical review letters},
  volume={118},
  number={1},
  pages={010501},
  year={2017},
  publisher={APS}
}

@article{low2019hamiltonian,
  title={Hamiltonian simulation by qubitization},
  author={Low, Guang Hao and Chuang, Isaac L},
  journal={Quantum},
  volume={3},
  pages={163},
  year={2019},
  publisher={Verein zur F{\"o}rderung des Open Access Publizierens in den Quantenwissenschaften}
}

@article{camps2020approximate,
  title={Approximate quantum circuit synthesis using block encodings},
  author={Camps, Daan and Van Beeumen, Roel},
  journal={Physical Review A},
  volume={102},
  number={5},
  pages={052411},
  year={2020},
  publisher={APS}
}

@article{berry2014high,
  title={High-order quantum algorithm for solving linear differential equations},
  author={Berry, Dominic W},
  journal={Journal of Physics A: Mathematical and Theoretical},
  volume={47},
  number={10},
  pages={105301},
  year={2014},
  publisher={IOP Publishing}
}

@article{childs2021high,
  title={High-precision quantum algorithms for partial differential equations},
  author={Childs, Andrew M and Liu, Jin-Peng and Ostrander, Aaron},
  journal={Quantum},
  volume={5},
  pages={574},
  year={2021},
  publisher={Verein zur F{\"o}rderung des Open Access Publizierens in den Quantenwissenschaften}
}

@article{wossnig2018quantum,
  title={Quantum linear system algorithm for dense matrices},
  author={Wossnig, Leonard and Zhao, Zhikuan and Prakash, Anupam},
  journal={Physical review letters},
  volume={120},
  number={5},
  pages={050502},
  year={2018},
  publisher={APS}
}

@article{clader2013preconditioned,
  title={Preconditioned quantum linear system algorithm},
  author={Clader, B David and Jacobs, Bryan C and Sprouse, Chad R},
  journal={Physical review letters},
  volume={110},
  number={25},
  pages={250504},
  year={2013},
  publisher={APS}
}

@article{childs2017lecture,
  title={Lecture notes on quantum algorithms},
  author={Childs, Andrew M},
  journal={Lecture notes at University of Maryland},
  year={2017}
}

@article{nghiem2023improved,
  title={Improved Quantum Algorithms for Eigenvalues Finding and Gradient Descent},
  author={Nghiem, Nhat A and Wei, Tzu-Chieh},
  journal={arXiv preprint arXiv:2312.14786},
  year={2023}
}

@article{tang2018quantum,
  title={Quantum-inspired classical algorithms for principal component analysis and supervised clustering},
  author={Tang, Ewin},
  journal={arXiv preprint arXiv:1811.00414},
  volume={4},
  year={2018}
}

@article{manin1980computable,
  title={Computable and uncomputable},
  author={Manin, Yuri},
  journal={Sovetskoye Radio, Moscow},
  volume={128},
  pages={15},
  year={1980}
}

@article{benioff1980computer,
  title={The computer as a physical system: A microscopic quantum mechanical Hamiltonian model of computers as represented by Turing machines},
  author={Benioff, Paul},
  journal={Journal of statistical physics},
  volume={22},
  pages={563--591},
  year={1980},
  publisher={Springer}
}

@article{huang2019near,
  title={Near-term quantum algorithms for linear systems of equations},
  author={Huang, Hsin-Yuan and Bharti, Kishor and Rebentrost, Patrick},
  journal={arXiv preprint arXiv:1909.07344},
  year={2019}
}

@article{sachdeva2014faster,
  title={Faster algorithms via approximation theory},
  author={Sachdeva, Sushant and Vishnoi, Nisheeth K and others},
  journal={Foundations and Trends{\textregistered} in Theoretical Computer Science},
  volume={9},
  number={2},
  pages={125--210},
  year={2014},
  publisher={Now Publishers, Inc.}
}

\newpage
\appendix

In this section, we provide a more detailed versions of those quantum algorithms introduced in the main text. 
\section{Preliminaries}
\label{sec: prelim}
To begin, we summarize the main recipes of our work, mostly derived from the seminal QSVT work \cite{gilyen2019quantum}. We keep the statements brief and precise for simplicity, with their proofs/ constructions referred to in their original works.

\begin{definition}[Block Encoding Unitary]~\cite{low2017optimal, low2019hamiltonian, gilyen2019quantum}
\label{def: blockencode} 
Let $A$ be some Hermitian matrix of size $N \times N$ whose matrix norm $|A| < 1$. Let a unitary $U$ have the following form:
\begin{align*}
    U = \begin{pmatrix}
       A & \cdot \\
       \cdot & \cdot \\
    \end{pmatrix}.
\end{align*}
Then $U$ is said to be an exact block encoding of matrix $A$. Equivalently, we can write $U = \ket{ \bf{0}}\bra{ \bf{0}} \otimes A + (\cdots)$, where $\ket{\bf 0}$ refers to the ancilla system required for the block encoding purpose. In the case where the $U$ has the form $ U  =  \ket{ \bf{0}}\bra{ \bf{0}} \otimes \Tilde{A} + (\cdots) $, where $|| \Tilde{A} - A || \leq \epsilon$ (with $||.||$ being the matrix norm), then $U$ is said to be an $\epsilon$-approximated block encoding of $A$. Furthermore, the action of $U$ on some quantum state $\ket{\bf 0}\ket{\phi}$ is:
\begin{align}
    \label{eqn: action}
    U \ket{\bf 0}\ket{\phi} = \ket{\bf 0} A\ket{\phi} +  \ket{\rm Garbage},
\end{align}
where $\ket{\rm Garbage }$ is a redundant state that is orthogonal to $\ket{\bf 0} A\ket{\phi}$. The above definition has multiple natural \textbf{corollaries}: 
\begin{itemize}
    \item First, an arbitrary unitary $U$ block encodes itself
    \item Second, suppose that $A$ is block encoded by some matrix $U$, then $A$ can be block encoded in a larger matrix by simply adding any ancilla (supposed to have dimension $m$), then note that $\Ibb_m \otimes U$ contains $A$ in the top-left corner, which is block encoding of $A$ again by definition 
    \item Third, it is almost trivial to block encode identity matrix of any dimension. For instance, we consider $\sigma_z \otimes \Ibb_m$ (for any $m$), which contains $\Ibb_m$ in the top-left corner. 
\end{itemize}
\end{definition}

\begin{lemma}[\cite{gilyen2019quantum} Product]
\label{lemma: product}
    Given the unitary block encoding of two matrices $A_1$ and $A_2$, then there exists an efficient procedure that constructs a unitary block encoding of $A_1 A_2$ using each block encoding of $A_1,A_2$ one time. 
\end{lemma}

\begin{lemma}[\cite{camps2020approximate} Tensor Product]
\label{lemma: tensorproduct}
    Given the unitary block encoding $\{U_i\}_{i=1}^m$ of multiple operators $\{M_i\}_{i=1}^m$ (assumed to be exact encoding), then, there is a procedure that produces the unitary block encoding operator of $\bigotimes_{i=1}^m M_i$, which requires paralle single uses of 
    $\{U_i\}_{i=1}^m$ and $\mathcal{O}(1)$ SWAP gates. 
\end{lemma}
The above lemma is a result in \cite{camps2020approximate}. 
\begin{lemma}[\cite{gilyen2019quantum} Block Encoding of a  Matrix]
\label{lemma: As}
    Given oracle access to $s$-sparse matrix $A$ of dimension $n\times n$, then an $\epsilon$-approximated unitary block encoding of $A/s$ can be prepared with gate/time complexity $\mathcal{O}\Big(\log n + \log^{2.5}(\frac{s^2}{\epsilon})\Big).$
\end{lemma}
This is presented in~\cite{gilyen2019quantum} (see their Lemma 48), and one can also find a review of the construction in~\cite{childs2017lecture}. We remark further that the scaling factor $s$ in the above lemma can be reduced by the preamplification method with further complexity $\mathcal{O}({s})$~\cite{gilyen2019quantum}.

\begin{lemma}[\cite{gilyen2019quantum} Linear combination ]
    Given unitary block encoding of multiple operators $\{M_i\}_{i=1}^m$. Then, there is a procedure that produces a unitary block encoding operator of \,$\sum_{i=1}^m \pm M_i/m $ in complexity $\mathcal{O}(m)$, e.g., using block encoding of each operator $M_i$ a single time. 
    \label{lemma: sumencoding}
\end{lemma}

\begin{lemma}[Scaling Block encoding] 
\label{lemma: scale}
    Given a block encoding of some matrix $A$ (as in~\cref{def: blockencode}), then the block encoding of $A/p$ where $p > 1$ can be prepared with an extra $\mathcal{O}(1)$ cost.  
\end{lemma}
To show this, we note that the matrix representation of RY rotational gate is
\begin{align}
   R_Y(\theta) = \begin{pmatrix}
        \cos(\theta/2) & -\sin(\theta/2) \\
        \sin(\theta/2) & \cos(\theta/2) 
    \end{pmatrix}.
\end{align}
If we choose $\theta$ such that $\cos(\theta/2) = 1/p$, then~\cref{lemma: tensorproduct} allows us to construct block encoding of $R_Y(\theta) \otimes \mathbb{I}_{{\rm dim}(A)}$  (${\rm dim}(A)$ refers to dimension of matirx $A$), which contains the diagonal matrix of size ${\rm dim}(A) \times {\rm dim}(A)$ with entries $1/p$. Then~\cref{lemma: product} can construct block encoding of $(1/p) \ \mathbb{I}_{{\rm dim}(A)} \cdot A = A/p$.  \\

The following is called amplification technique:
\begin{lemma}[\cite{gilyen2019quantum} Theorem 30 ]\label{lemma: amp_amp}
Let $U$, $\Pi$, $\widetilde{\Pi} \in {\rm End}(\mathcal{H}_U)$ be linear operators on $\mathcal{H}_U$ such that $U$ is a unitary, and $\Pi$, $\widetilde{\Pi}$ are orthogonal projectors. 
Let $\gamma>1$ and $\delta,\epsilon \in (0,\frac{1}{2})$. 
Suppose that $\widetilde{\Pi}U\Pi=W \Sigma V^\dagger=\sum_{i}\varsigma_i\ket{w_i}\bra{v_i}$ is a singular value decomposition. 
Then there is an $m= \mathcal{O} \Big(\frac{\gamma}{\delta}
\log \left(\frac{\gamma}{\epsilon} \right)\Big)$ and an efficiently computable $\Phi\in\mathbb{R}^m$ such that
\begin{align}
& \left(\bra{+}\otimes\widetilde{\Pi}_{\leq\frac{1-\delta}{\gamma}}\right)U_\Phi \left(\ket{+}\otimes\Pi_{\leq\frac{1-\delta}{\gamma}}\right) \\ &=\sum_{i\colon\varsigma_i\leq \frac{1-\delta}{\gamma} }\tilde{\varsigma}_i\ket{w_i}\bra{v_i} , \text{ where } \Big|\!\Big|\frac{\tilde{\varsigma}_i}{\gamma\varsigma_i}-1 \Big|\!\Big|\leq \epsilon.
\end{align}
Moreover, $U_\Phi$ can be implemented using a single ancilla qubit with $m$ uses of $U$ and $U^\dagger$, $m$ uses of C$_\Pi$NOT and $m$ uses of C$_{\widetilde{\Pi}}$NOT gates and $m$ single qubit gates.
Here,
\begin{itemize}
\item C$_\Pi$NOT$:=X \otimes \Pi + I \otimes (I - \Pi)$ and a similar definition for C$_{\widetilde{\Pi}}$NOT; see Definition 2 in \cite{gilyen2019quantum},
\item $U_\Phi$: alternating phase modulation sequence; see Definition 15 in \cite{gilyen2019quantum},
\item $\Pi_{\leq \delta}$, $\widetilde{\Pi}_{\leq \delta}$: singular value threshold projectors; see Definition 24 in \cite{gilyen2019quantum}.
\end{itemize}
\end{lemma}

\begin{lemma}[Projector]
\label{lemma: projector}
The block encoding of a projector $\ket{j-1}\bra{j-1}$ (for any $j=1,2, ...,n$) by a circuit of depth $\mathcal{O}\big( \log n \big)$ 
\end{lemma}
\noindent
\textit{Proof.} First we note that it takes a circuit of depth $\mathcal{O}(1)$ to generate $\ket{j-1}$ from $\ket{0}$. Then \cref{lemma: improveddme} can be used to construct the block encoding of $\ket{j-1}\bra{j-1}$. 
\begin{lemma}[\cite{guo2024nonlinear}, or Theorem 2 in \cite{rattew2023non}]
\label{lemma: diagonal}
     Given an n-qubit quantum state specified by a state-preparation-unitary $U$, such that $\ket{\psi}_n=U\ket{0}_n=\sum^{N-1}_{k=0}\psi_k \ket{k}_n$ (with $\psi_k \in \mathbb{C}$ and $N=2^n$), we can prepare an exact block-encoding $U_A$ of the diagonal matrix $A = {\rm diag}(\psi_0, ...,\psi_{N-1})$ with $\mathcal{O}(n)$ circuit depth and a total of $\mathcal{O}(1)$ queries to a controlled-$U$ gate  with $n+3$ ancillary qubits.
\end{lemma}
\begin{lemma}
\label{lemma: qsvt}[\cite{gilyen2019quantum} Theorem 56]
\label{lemma: theorem56}  
Suppose that $U$ is an
$(\alpha, a, \epsilon)$-encoding of a Hermitian matrix $A$. (See Definition 43 of~\cite{gilyen2019quantum} for the definition.)
If $P \in \mathbb{R}[x]$ is a degree-$d$ polynomial satisfying that
\begin{itemize}
\item for all $x \in [-1,1]$: $|P(x)| \leq \frac{1}{2}$,
\end{itemize}
then, there is a quantum circuit $\tilde{U}$, which is an $(1,a+2,4d \sqrt{\frac{\epsilon}{\alpha}})$-encoding of $P(A/\alpha)$ and
consists of $d$ applications of $U$ and $U^\dagger$ gates, a single application of controlled-$U$ and $\mathcal{O}((a+1)d)$
other one- and two-qubit gates.
\end{lemma}
\begin{lemma}[Negative Power Exponent \cite{gilyen2019quantum}, \cite{chakraborty2018power}]
\label{lemma: negative}
    Given a block encoding of a positive matrix $\mathcal{M}$ such that 
    $$ \frac{\Ibb}{\kappa_M} \leq \mathcal{M} \leq \Ibb. $$
    then we can implement an $\epsilon$-approximated block encoding of $A^{-c}/(2\kappa_M^c)$ in complexity $\mathcal{O}( \kappa_M T_M (1+c) \log^2(  \frac{\kappa_M^{1+c}}{\epsilon} ) )$ where $T_M$ is the complexity to obtain the block encoding of $\mathcal{M}$. 
\end{lemma}

\begin{lemma}[Positive Power Exponent \cite{gilyen2019quantum},\cite{chakraborty2018power}]
\label{lemma: positive}
    Given a block encoding of a positive matrix $\mathcal{M}$ such that 
    $$ \frac{\Ibb}{\kappa_M} \leq \mathcal{M} \leq \Ibb. $$
   Let $c \in (0,1)$. Then we can implement an $\epsilon$-approximated block encoding of $\mathcal{M}^c/2$ in time complexity $\mathcal{O}( \kappa_M T_M \log^2 (\frac{\kappa_M}{\epsilon})  )$, where $T_M$ is the complexity to obtain the block encoding of $\mathcal{M}$. 
\end{lemma}

\begin{lemma}[Corollary 64 of \cite{gilyen2019quantum}  ]
\label{lemma: exponentialapproximation}
   Let $\beta \in \mathbb{R}_+$ and $\epsilon \in (0,1/2]$. There exists an efficiently constructible polynomial $P \in \mathbb{R}[x]$ such that 
   $$ \Big|\!\Big| e^{ -\beta ( 1-x ) } - P(x)  \Big|\!\Big|_{x\in[-1,1]} \leq \epsilon. $$
   Moreover, the degree of $P$ is $\mathcal{O}\Big( \sqrt{\max[\beta, \log(\frac{1}{\epsilon})] \log(\frac{1}{\epsilon}}) \Big).$
\end{lemma}

\subsection{Proof of Lemma \ref{lemma: blockencodingknownmatrix}}
\label{sec: proofofLemmablockencodingknownmatrix}
Given the above preliminaries, our subsequent algorithms are built on the following recipes. The first one is the quantum state preparation protocols, i.e., preparing a state $\ket{\Phi} = \sum_{i=1}^n a_i \ket{i-1}$ provided known entries $\{a_i\}$. The method in \cite{zhang2022quantum} uses a quantum circuit of logarithmical depth, a number of ancilla qubits, and a classical preprocessing of logarithmical time. This is probably the most universal method for state preparation, with ``universal'' indicating that this method can be applied to any state. Despite this, the number of ancilla qubits required depends on the number of nonzero element of $\ket{\Phi}$, which implies that this method is most efficient for sparse states. On the other hand, if the state $\ket{\Phi}$ has the structures as indicated in any of the works \cite{grover2000synthesis,grover2002creating,plesch2011quantum, schuld2018supervised, nakaji2022approximate,marin2023quantum,zoufal2019quantum,prakash2014quantum, zhang2022quantum}, then it can be prepared using a quantum circuit of logarithmical depth, using a modest amount of ancilla qubits and no classical preprocessing. In principle, any of these aforementioned methods can be used in our subsequent construction. However, in this work, we point out another type of structure of $\ket{\Phi}$ that can potentially  help reduce the ancilla qubits and circuit depth, which is when $\ket{\Phi}$ admits the decomposition $ \sum_{i=1}^M \ket{\psi_{i_1}} \otimes \cdots \otimes \ket{\psi_{i_k}}$ with the entries of $ \{\ket{\psi_{i_j}}  \}_{j=1}^k$ are classically known (or preprocessed). For convenience, we recapitulate the essential information for subsequent uses as follows, with the detailed proof regarding the last statement shall be provided in the Appendix \ref{sec: improvingstatepreparation}.
\begin{lemma}[Efficient state preparation \cite{grover2000synthesis,grover2002creating,plesch2011quantum, schuld2018supervised, nakaji2022approximate,marin2023quantum,zoufal2019quantum,prakash2014quantum, zhang2022quantum, mcardle2022quantum}]
\label{lemma: stateprepration}
Let $\ket{\Phi} = \sum_{i=1}^N a_i \ket{i-1}$ be $N$-dimensional quantum state with known entries, with $s$ non-zero entries. Then: 
\begin{itemize}
    \item (General case) $\ket{\Phi}$ can be prepared with a $\log(n)$-qubits circuit of depth $\mathcal{O}\big( \log (s\log N)\big)$, using $\mathcal{O}(s)$ ancilla qubits and a classical preprocessing of complexity $\mathcal{O}( \log N)$ \cite{zhang2022quantum}. The classical preprocessing cost can be improved to $\mathcal{O}(1)$ if $\{a_i\}_{i=1}^N$ can be partitioned into subsets in which each subset contains similar entries. 
    \item (Structured cases 1) If $\ket{\Phi}$ has structure as any of the following works \cite{grover2000synthesis,grover2002creating,plesch2011quantum, schuld2018supervised, nakaji2022approximate,marin2023quantum,zoufal2019quantum,prakash2014quantum}, then $\ket{\Phi}$ can be prepared with a $\log(N)$-qubits circuit of depth $\mathcal{O}\left( \log N\right)$, $\mathcal{O}(1)$ ancilla qubit and no classical preprocessing.
    \item  (Structured cases 2) If $\ket{\Phi} = \sum_{i=1}^M \alpha_i \ket{\psi_{i_1}} \otimes \cdots \otimes \ket{\psi_{i_k}}$, and for all $i$, the entries of $ \{  \ket{\psi_{i_j}} \}_{j=1}^k$ as well as $\{ \alpha_i \}_{i=1}^M$ are classically known. Let $d_j \equiv \dim \ket{\psi_{i_j}}$ denote the dimension of $\ket{\psi_{i_j}} $ and $s_j$ denote the sparsity of $ \ket{\psi_{i_j}}$. Defining $d = \max \{ d_j \}_{j=1}^k$, and $s_{\max} = \max \{ s_j \}_{j=1}^k$, then it can be prepared using the $\log(n)$-qubits quantum circuit of depth $\mathcal{O}\left( M \log d  \right)$, $\mathcal{O}\left( k s_{\max}   \right)$ ancilla qubits, and a classical preprocessing of time $\mathcal{O}\left( \log n\right)$. If for all $j$, $d_j, s_j \in \mathcal{O}(1)$, then the circuit depth is $\mathcal{O}\left( M \right)$, the number of ancilla qubits is $\mathcal{O}\left( \log s \right)$.
\end{itemize}
\end{lemma}

The next recipe is the block-encoding of density operator. 
\begin{lemma}[\cite{gilyen2019quantum} Block Encoding Density Matrix]
\label{lemma: improveddme}
Let $\rho = \Tr_A \ket{\Phi}\bra{\Phi}$, where $\rho \in \mathbb{H}_B$, $\ket{\Phi} \in  \mathbb{H}_A \otimes \mathbb{H}_B$. Given unitary $U$ that generates $\ket{\Phi}$ from $\ket{\bf 0}_A \otimes \ket{\bf 0}_B$, then there exists a highly efficient procedure that constructs an exact unitary block encoding of $\rho$ using $U$ and $U^\dagger$ a single time, respectively.
\end{lemma}
Consider a matrix $\mathcal{A} \in \Rbb^{m \times n}$, and denote $\mathcal{A}^i \in \Rbb^{m}$ as its $i$-th column. Provided the classical knowledge of entries of $\mathcal{A}$, then Lemma \ref{lemma: stateprepration}, once classically pre-processing the entries of $\mathcal{A}$, allows us to prepare the following state:
\begin{align}
    \ket{\mathcal{A} } = \frac{1}{||\mathcal{A}||_F} \sum_{i=1}^n  \mathcal{A}^i\ket{i} 
\end{align}
where $||\mathcal{A}||_F $ denotes its Frobenius norm. The quantum circuit depth required is $\mathcal{O}\left( \log (s_{\mathcal{A}} \log m)\right)$ and the number of ancilla qubits required is $s_{\mathcal{A}}$ where $s_{\mathcal{A}}$ is the number of nonzero elements in $\mathcal{A}$.  If we trace out the first ancilla system, then we obtain the density operator $\frac{1}{||\mathcal{A}||_F^2|| }  \mathcal{A}^\dagger \mathcal{A} $ , which can be block-encoded via Lemma \ref{lemma: improveddme}. Next, we can use Lemma \ref{lemma: positive} with the choice $c=\frac{1}{2}$, we can transform the block-encoding of $ \frac{1}{||\mathcal{A}||_F^2|| }  \mathcal{A}^\dagger \mathcal{A} $  to the $\epsilon$-approximated block-encoding of $ \frac{1}{||\mathcal{A}||_F } \mathcal{A}$. The factor $ ||\mathcal{A}||_F $ can be removed using Lemma \ref{lemma: amp_amp} (if $||\mathcal{A}||_F \geq 1$) or \ref{lemma: scale} (if $||\mathcal{A}||_F  \leq 1$). Thus, Lemma \ref{lemma: blockencodingknownmatrix} is proved.

\section{More on quantum state preparation (Lemma \ref{lemma: stateprepration})}
\label{sec: improvingstatepreparation}
We recall that within Lemma \ref{lemma: stateprepration}, we stated:
\begin{itemize}
    \item If $\ket{\Phi} = \sum_{i=1}^M \alpha_i \ket{\psi_{i_1}} \otimes \cdots \otimes \ket{\psi_{i_k}}$, and for all $i$, the entries of $ \{  \ket{\psi_{i_j}} \}_{j=1}^k$ as well as $\{ \alpha_i \}_{i=1}^M$ are classically known. Let $d_j \equiv \dim \ket{\psi_{i_j}}$ denote the dimension of $\ket{\psi_{i_j}} $ and $s_j$ denote the sparsity of $ \ket{\psi_{i_j}}$. Defining $d = \max \{ d_j \}_{j=1}^k$, and $s_{\max} = \max \{ s_j \}_{j=1}^k$, then it can be prepared using the quantum circuit depth $\mathcal{O}\left( M \log d  \right)$, $\mathcal{O}\left( k s_{\max}   \right)$ ancilla qubits, and a classical preprocessing of time $\mathcal{O}\left( \log n\right)$. If for all $j$, $d_j, s_j \in \mathcal{O}(1)$, then the circuit depth is $\mathcal{O}\left( M \right)$, the number of ancilla qubits is $\mathcal{O}\left( \log s \right)$.
\end{itemize}
We now prove this statement. Consider a fixed value of $i$, for each $\ket{\psi_{i_j}}  $, define $d_j$ as its dimension and $s_j$ as its sparsity. It can be seen that, given the dimension of $\ket{\Phi}$ is $n$ and the sparsity is $s$, we have $ \prod_{j=1}^k d_j =n$ and $\prod_{j=1}^k s_j =s $. The method of \cite{zhang2022quantum} allows us to prepare this state using a quantum circuit of depth $\mathcal{O}\left( \log d_j   \right)$, the number of ancilla qubits $\mathcal{O}(s_i)$ and classical preprocessing of time $\mathcal{O}\left( \log d_j \right)$. Let $U_{ { i_j}}$ denote the unitary that prepares this state. Then it can be seen that the first column of $U_{ { i_j}}$ is $ \ket{\psi_{i_j}}$. 

Next, we use Lemma \ref{lemma: tensorproduct} to construct the block-encoding of $U_{ { i_1}} \otimes U_{ { i_2}} \otimes \cdots \otimes U_{ { i_k}} $. As the first column of each $U_{ { i_j}}$ is $ \ket{\psi_{i_j}}$, then the first column of this block-encoding $U_{ { i_1}} \otimes U_{ { i_2}} \otimes \cdots \otimes U_{ { i_k}} $ is $ \ket{\psi_{i_1}} \otimes \cdots \otimes \ket{\psi_{i_k}}$. Because the method of Lemma \ref{lemma: tensorproduct} uses all these circuits $\{ U_{ { i_j}} \}_{j=1}^k$ in parallel, the circuit depth of this step is $\mathcal{O}\left( \log d\right)$, the number of ancilla used is $\mathcal{O}\left( \sum_{j=1}^k s_i \right)$, the classical preprocessing time is $\mathcal{O}\left( \sum_{j=1}^k \log d_i  \right) = \mathcal{O}\left( \log n\right)$. However, this classical preprocessing time can be improved to $\mathcal{O}\left( \log d \right)$ if we use parallel computation, for which all the preprocessing of $  \{  \ket{\psi_{i_j}} \}_{j=1}^k$ can be done in parallel.

Repeating this construction for all $i =1,2,...,M$, we then use Lemma \ref{lemma: sumencoding} to construct the block-encoding of: 
\begin{align}
    \frac{1}{\alpha} \sum_{i=1}^M \alpha_i U_{ { i_1}} \otimes U_{ { i_2}} \otimes \cdots \otimes U_{ { i_k}} 
\end{align}
where $\alpha = \sum_{i=1}^M \alpha_i$. Because the first column of the unitary $U_{ { i_1}} \otimes U_{ { i_2}} \otimes \cdots \otimes U_{ { i_k}} $ is $ \ket{\psi_{i_1}} \otimes \cdots \otimes \ket{\psi_{i_k}}$, then the first column of the above operator is $  \frac{1}{\alpha}\ket{\Phi}$. This block-encoding has circuit depth $\mathcal{O}\left( M \log d\right) $. Taking this block-encoding and applying it to the state $\ket{\bf 0}\ket{0}_{n}$ where $\ket{\bf 0}$ denotes the ancilla qubits for block-encoding purpose, and $\ket{0}_n$ denotes the first computational basis state of the $n$-dimensional Hilbert space. According to Definition \ref{def: blockencode}, we obtain the state:
\begin{align}
    \ket{\bf 0}   \frac{1}{\alpha}\ket{\Phi} + \ket{\rm Garbage}
\end{align}
Measuring the ancilla and post-select $\ket{\bf 0}$, we obtain the state $\ket{\Phi}$ with probability $\frac{1}{\alpha^2}$, which can be quadratically improved by using amplitude amplification.

We now take a closer look at the complexity of all parts, investigating their most efficient regimes as well as the trade-off, if any, among them. 
\begin{itemize}
    \item Circuit depth: The procedure above has quantum circuit depth $\mathcal{O}\left( M \log d\right)$. The value of $d$, by definition, is $d = \max \{ d_j\}_{j=1}^k$. If all $\{ d_j \}_{j=1}^k $ are of the same order of magnitude, then $ \prod_{j=1}^k d_j =n$ implies that each $d_j \in \mathcal{O}\left( n^{1/k} \right)$, if the value of $k$ is fixed. In this case, the quantum circuit depth is $\mathcal{O}\left( M \log n^{1/k} \right)$. If instead, the value for all $\{ d_j\}$ are $\sim \mathcal{O}(1)$, then $k$ is of order $\mathcal{O}\left( \log n\right)$. Thus, the depth of the quantum circuit is most optimized when $\{d_j\}_{j=1}^k$ is $\in \mathcal{O}(1)$, and $M = \mathcal{O}(1)$. In this case, the depth is $\mathcal{O}\left(  1 \right)$.
    \item Ancilla qubits: The number of ancilla qubits $\mathcal{O}\left( \sum_{j=1}^k s_j \right)= \mathcal{O}\left( k s_{\max} \right)$. Similarly to above, because $ \prod_{j=1}^k s_j \leq s$, so if all $\{s_j\}_{j=1}^k$ are of the same magnitude, then each $s_j$ has magnitude $\sim \mathcal{O}\left( s^{1/k} \right)$ if we fix a value for $k$. In this case the ancilla qubits is $\mathcal{O}\left( k s^{1/k} \right)$. Otherwise, if for all $\{s_j\}_{j=1}^k$, their values are $\sim \mathcal{O}(1)$, then $k$ needs to be $\mathcal{O}\left( \log s\right)$, resulting in the total number of anilla qubits being $\mathcal{O}\left( \log s\right)$.
\end{itemize}
The above analysis clearly implies that the cost of preparing $\ket{\Phi}$ can be significantly reduced if $\ket{\Phi}$ has a certain structure. The degree of reduction depends greatly on how the value of $\{d_j, s_j\}_{j=1}^k$ behaves. In the best case, as we see above, the circuit depth required to prepare the state can be $\mathcal{O}(1)$, and the number of ancilla qubits is $\mathcal{O}\left( \log s\right)$.

\section{Finding principal components based on the power method}
\label{sec: PCApowermethod}
By a slight abuse of notation, we define the following $m \times n$ matrix:
\begin{align}
     \mathcal{X} = \begin{pmatrix}
        \xbf^1_1 & \xbf^1_2 & \cdots & \xbf^1_n \\
        \xbf^2_1 & \xbf^2_2 & \cdots & \xbf^2_n \\
        \vdots & \vdots & \ddots & \vdots \\
        \xbf^m_1 & \xbf^m_2 & \cdots & \xbf_n^m
    \end{pmatrix}
\end{align}
Without loss of generalization, we assume that their sum of norms $\sum_{i=1}^m ||\xbf^i ||^2 = ||\mathcal{X}||_F = 1$, for simplicity. Otherwise, we scale the matrix as $\mathcal{X} \longrightarrow \frac{1}{||\mathcal{X}||_F}  \mathcal{X}$ and consider this new matrix instead. This scaling essentially maintains the spectrum, but only the eigenvalues are scaled. Therefore, we will assume this condition through the remaining. 

Using the Lemma \ref{lemma: stateprepration}, we prepare the following state with a circuit $U_\Phi$ of depth $\mathcal{O}(\log mn)$:
\begin{align}
    \ket{\Phi} = \sum_{i=1}^m \sum_{j=1}^n  \xbf^i_j \ket{i} \ket{j}
\end{align}
which is essentially $\sum_{j=1}^n \big( \sum_{i=1}^m \xbf^i_j \ket{i} \big) \ket{j} $. The density state $\ket{\Phi}\bra{\Phi}$ is:
\begin{align}
     \ket{\Phi}\bra{\Phi} = \sum_{j=1}^n  \sum_{k=1}^n \big( \sum_{i=1}^m \xbf^i_j \ket{i}  \big)  \big( \sum_{p=1}^m \xbf^p_k \bra{p}   \big) \otimes  \ket{j} \bra{k}
\end{align}
If we trace out the first register (that holds $\ket{i}$ index), then we obtain the following density state: 
\begin{align}
     \sum_{j=1}^n  \sum_{k=1}^n \big( \sum_{p=1}^m \xbf^p_k \bra{p}   \big)  \big( \sum_{i=1}^m \xbf^i_j \ket{i}  \big) \ket{j}\bra{k}
\end{align}
It is not hard to see that the above matrix is indeed $\mathcal{X}^T \mathcal{X}$ because $\sum_{i=1}^m \xbf^i_j \ket{i}  $ is the $j$-th row of $\mathcal{X}^T$, and $\sum_{p=1}^m \xbf^p_k \bra{p}   $ is the $k$-th column of $\mathcal{X}$. Given that we have $U_\Phi$ that generates $\ket{\Phi}$, by virtue of the following lemma (see their Lemma 45 of \cite{gilyen2019quantum}): 
\begin{lemma}[\cite{gilyen2019quantum} Block Encoding Density Matrix]
\label{lemma: improveddme}
Let $\rho = \Tr_A \ket{\Phi}\bra{\Phi}$, where $\rho \in \mathbb{H}_B$, $\ket{\Phi} \in  \mathbb{H}_A \otimes \mathbb{H}_B$. Given unitary $U$ that generates $\ket{\Phi}$ from $\ket{\bf 0}_A \otimes \ket{\bf 0}_B$, then there exists a highly efficient procedure that constructs an exact unitary block encoding of $\rho$ using $U$ and $U^\dagger$ a single time, respectively.
\end{lemma}
it is possible to block-encode $\mathcal{X}^T \mathcal{X}$, with a total circuit of depth $\mathcal{O}( \log mn )$. We can use this block encoding and use \cref{lemma: scale} with scaling factor $m$, to obtain the block encoding of $ \frac{1}{m}\mathcal{X}^T \mathcal{X}$. \\

The next step is to obtain the block encoding of $\mu \mu^T$. Recall from the above that thanks to \cref{lemma: stateprepration}, we can prepare the state $ \ket{\Phi}$, which is also $ \sum_{i=1}^m \sum_{j=1}^n  \xbf^i_j \ket{i} \ket{j} = \sum_{i=1}^m \ket{i} \otimes  \xbf^i $. Since $ H^{\otimes \log m} \otimes \Ibb_n $ is simple to prepare, applying it to $\ket{\Phi}$ yields the state $\ket{\Phi'}$, which is:
\begin{align}
   \ket{\Phi'} &=  H^{\otimes \log m} \otimes \Ibb_n  \cdot \sum_{i=1}^m  \ket{i} \otimes \xbf^i \\& = \frac{1}{\sqrt{m}} \ket{0}_{m} \otimes  \big( \sum_{i=1}^m \xbf^i  \big)   + \ket{\rm Redundant}
\end{align} 
where $\ket{0}_m$ specifically denote the first computational basis state of the $m$-dimensional Hilbert space, and $\ket{\rm Redundant}$ is the irrelevant state that is orthogonal to $\ket{0}_{m}\otimes \big( \sum_{i=1}^m \xbf^i  \big)  $. Then \cref{lemma: improveddme} allows us to construct the block encoding of the density state $\ket{\Phi'}\bra{\Phi'}$, which is equivalent to:
\begin{align}
    \ket{\Phi'}\bra{\Phi'} = \frac{1}{m} \ket{0}_m\bra{0}_m \otimes \big( \sum_{i=1}^m \xbf^i  \big) \big( \sum_{i=1}^m \xbf^i  \big)^T + (...)
\end{align}
where $(...)$ denotes the remaining irrelevant terms. According to \cref{def: blockencode}, the above density operator is again a block encoding of $\frac{1}{m} \big( \sum_{i=1}^m \xbf^i  \big) \big( \sum_{i=1}^m \xbf^i  \big)^T $, which can be combined with \cref{lemma: scale} (with scaling factor $m$) to transform it into $\frac{1}{m^2} \big( \sum_{i=1}^m \xbf^i  \big) \big( \sum_{i=1}^m \xbf^i  \big)^T  $.  Recall that we have defined the centroid $\mu = \sum_{i=1}^m \frac{\xbf^i}{m} $, so the above procedure allows us to construct the block encoding of $\mu \mu^T$. The complexity of the above procedure is mainly coming from an application of \cref{lemma: stateprepration} to prepare $\ket{\Phi}$, and of \cref{lemma: improveddme} to prepare the block encoding of $ \ket{\Phi'}\bra{\Phi'}$, resulting in total complexity  $\mathcal{O}( \log mn)$. 

The block encoding of $\frac{1}{m} \mathcal{X}^T \mathcal{X}$ and of $\mu \mu^T $ allows us to construct the block encoding of $\frac{1}{2}\big( \frac{1}{m} \mathcal{X}^T \mathcal{X} - \mu \mu^T \big)$, which is exactly $\frac{1}{2} \mathcal{C} $ where $\mathcal{C}$ is the covariance matrix. The next goal is to find the principal components -- the largest eigenvalues and the corresponding eigenvectors of $\mathcal{C}$. In the following, we describe our quantum PCA algorithm based on power method. 

Power method has appeared in a series of works \cite{nghiem2022quantum, nghiem2024improved, nghiem2023improved}, in which they proposed quantum algorithms for finding the largest eigenvalues based on the classical power method \cite{friedman1998error, golub2013matrix}.  In fact, in a recent attempt \cite{nghiem2025new}, the author  also proposed a new quantum PCA algorithm based on power method, however,  as we mentioned previously, their method has linear scaling in the number of sample $m$, and polynomial in the inverse of error.  For the purpose at hand, we refer the interested readers to these original works and recapitulate their main results as follows:
\begin{lemma}[Ref.~\cite{nghiem2023improved}, Appendix \ref{sec: reviewpowermethod}]
\label{lemma: largestsmallest}
    Given the block encoding of a positive semidefinite Hermitian matrix $A$ of size $n\times n$. Let the eigenvectors of a $A$ be $\ket{A_1}, \ket{A_2},...,\ket{A_n}$ and eigenvalues be $A_1,A_2,...,A_n$. Suppose without loss of generalization that $A_1 > A_2 > ... > A_n $. Let $\ket{\psi}$ be some random initial state with a known quantum circuit of negligible depth, and define $\gamma = | \braket{\psi, A_1} |$.  Then the largest eigenvalue $A_1$ can be estimated up to additive precision $\epsilon$ in complexity $\mathcal{O}\Big(  T_A \big( \frac{1}{ |A_1-A_2| \gamma \epsilon}\big) \log \big(\frac{n}{\epsilon} \big) \log \frac{1}{\epsilon} \Big)$ where $T_A$ is the complexity of producing block encoding of $A$. Additionally, the eigenvector $\ket{A_1}$ corresponding to this eigenvalue can be obtained with complexity $\mathcal{O}\Big(  T_A \frac{1}{|A_1-A_2| \gamma}\log \big(\frac{n}{\epsilon} \big) \log \frac{1}{\epsilon}\Big) $
\end{lemma}
For the purpose of presentation, we denote the eigenvectors of $\mathcal{C}$ as $\ket{\lambda_1}, \ket{\lambda_2},..., \ket{\lambda_n}$ and the corresponding (ordered) eigenvalues are $\lambda_1 > \lambda_2 > ...> \lambda_n$. The application of the above lemma to our case is straightforward, because the covariance matrix $\mathcal{C}$ is positive semidefinite (see in the previous section, we had $\mathcal{C} = \mathcal{X}_{\rm center}^T\mathcal{X}_{\rm center}$, which is apparently positive semidefinite).  The complexity for obtaining the block encoding of $\frac{1}{2} \mathcal{C}$ is the sum of complexity for obtaining the block encoding of $ \mathcal{X}^T\mathcal{X}$ and of $\mu \mu^T$, so totally it is $\mathcal{O}(\log mn)$. Thus, the complexity in obtaining the first principal component, $\ket{\lambda_1}$, is 
$$\mathcal{O}\Big( \frac{1}{|\lambda_1-\lambda_2| \gamma} \log(mn) \log \big(\frac{n}{\epsilon} \big) \log \frac{1}{\epsilon}  \Big)$$
To find the second largest eigenvalue and the corresponding eigenvector, we need the following result, which is an extension of the above \cref{lemma: largestsmallest}:
\begin{lemma}
\label{lemma: extensionlemmalargestsmallest}
    In the context of \cref{lemma: largestsmallest}, there is a quantum procedure of complexity $\mathcal{O}\Big( T_A \frac{1}{|A_1-A_2| \gamma}\textcolor{black}{\log \big(\frac{n}{\epsilon} \big)} \log\frac{1}{\epsilon}\Big) $ that outputs an $\epsilon$-approximated block encoding of $ A_1 \ket{A_1}\bra{A_1}$. 
\end{lemma}
Details of \cref{lemma: largestsmallest} and the above \cref{lemma: extensionlemmalargestsmallest} will be provided in the \cref{sec: reviewpowermethod}. Given that we have the block encoding of $\frac{1}{2}\mathcal{C}$, an application of the above lemma with $A$ replaced by $\mathcal{C}/2$ yields the block encoding of $ \frac{1}{2}\lambda_1 \ket{\lambda_1}\bra{\lambda_1}$. Then we take the block encoding of $\frac{1}{2}\mathcal{C}$, a block encoding of $  \frac{1}{2}\lambda_1\ket{\lambda_1}\bra{\lambda_1}$ and \cref{lemma: sumencoding} to construct the block encoding of:
\begin{align}
    \frac{1}{4}  \Big( \mathcal{C} - \lambda_1 \ket{\lambda_1}\bra{\lambda_1} \Big)  \equiv \mathcal{C}_1
\end{align}
Because the complexity to obtain the block encoding of $\frac{1}{2}\mathcal{C}$ is $\mathcal{O}(\log mn)$, the complexity to obtain the block encoding of $ \frac{1}{2}\lambda_1\ket{\lambda_1}\bra{\lambda_1}$ is $\mathcal{O}\Big( \log(mn) \frac{1}{|\lambda_1-\lambda_2|}  \log \big(\frac{n}{\epsilon}\big) \log\frac{1}{\epsilon} \Big)$. So the complexity to obtain the block encoding of the above operator, $\mathcal{C}_1$, is 
$$ \mathcal{O}\Big( \log(mn) \frac{1}{|\lambda_1-\lambda_2| \gamma}  \textcolor{black}{\log \big(\frac{n}{\epsilon} \big)} \log\frac{1}{\epsilon} \Big)$$
This matrix $\mathcal{C}_1$ is apparently has the largest eigenvalue to be $\lambda_2/4$ and the corresponding eigenvector is $\ket{\lambda_2}$. So we can repeat an application of \cref{lemma: largestsmallest} to find them, thus revealing the second principal component. Provided the complexity for block-encoding $\mathcal{C}_1$  as above, the complexity for an application of \cref{lemma: largestsmallest} is then:
$$ \mathcal{O}\Big( \log(mn) \frac{1}{|\lambda_1-\lambda_2| |\lambda_2-\lambda_3|\textcolor{black}{\gamma^2}}  \textcolor{black}{\log \big(\frac{n}{\epsilon} \big)} \log^2 \frac{1}{\epsilon} \Big)$$
for obtaining $ \lambda_2/4$ and $\ket{\lambda_2}$.  In a similar manner, we use \cref{lemma: extensionlemmalargestsmallest} again and repeat the same procedure to obtain the block encoding of $\frac{1}{8} \Big( C - \lambda_1\ket{\lambda_1}\bra{\lambda_1} - \lambda_2 \ket{\lambda_2}\bra{\lambda_2}  \Big) \equiv \mathcal{C}_2$. This operator has $\lambda_3/8$ as largest eigenvalue and corresponding eigenvector is $\ket{\lambda_3}$, which can be found by applying \cref{lemma: largestsmallest}, resulting in a total complexity:
$$ \mathcal{O}\Big( \log(mn) \frac{1}{|\lambda_1-\lambda_2| |\lambda_2-\lambda_3| |\lambda_3 -\lambda_4  |\gamma^3 }  \log^3 \big(\frac{n}{\epsilon}\big) \log^3 \frac{1}{\epsilon} \Big)$$
Continuing the procedure, say, $r$ times to find the $r$ principal components that we desire, then we complete the refined quantum PCA algorithm. \\

\begin{algorithm}[H]
\caption{Quantum PCA via covariance estimation, power method, and gradient refinement}
\label{alg:qpca}
\DontPrintSemicolon
\SetAlgoLined
\KwIn{Classical data $\{ \xbf^1,\ldots,\xbf^m\}\subset\mathbb{R}^n$ (or $\mathbb{C}^n$), target rank $k$, power iterations $T_{\mathrm{pow}}$, gradient steps $T_{\mathrm{gd}}$, step size $\eta>0$.}
\KwOut{Approximate top-$k$ principal components $\{ \ubf_1,\ldots,\ubf_k\}$.}

\textbf{State preparation} (\cref{lemma: stateprepration}): Prepare the data superposition
\[
\ket{\Psi_X} \propto \sum_{i=1}^m \ket{i}\,\xbf^i .
\]

\textbf{Second-moment estimation} (\cref{lemma: improveddme}): Using density-matrix estimation on the index register of $\ket{\Psi_X}$, obtain an implicit linear-operator access to
\[
\frac{1}{m}\,\Xcal^\top \Xcal .
\]

\textbf{Mean estimation} (\cref{lemma: product} + \cref{lemma: improveddme}): Estimate the sample mean
$\mu := \frac{1}{m}\sum_{i=1}^m \xbf^i$ and enable operator application of $\mu^\top\mu$.

\textbf{Sum/difference encoding} (\cref{lemma: sumencoding}): Combine the two operators to realize
\[
\frac{1}{2}\Ccal \;\equiv\; \frac{1}{2}\Big(\tfrac{1}{m}\Xcal^\top\Xcal - \mu^\top\mu\Big),
\]
which is proportional to the sample covariance operator $\Ccal$.

\BlankLine
\textbf{Option A: Power method for extremal eigenpairs} (\cref{lemma: largestsmallest})\;
Initialize a random unit vector $\vbf^{(0)}$.\;
\For{$t=1,2,\ldots,T_{\mathrm{pow}}$}{
  $\wbf^{(t)} \leftarrow \big(\tfrac{1}{2}\Ccal\big)\,\vbf^{(t-1)}$ \tcp*{apply estimated covariance}
  $\vbf^{(t)} \leftarrow \wbf^{(t)} / \|\wbf^{(t)}\|_2$ \tcp*{normalize}
}
Set $\ubf_1 \leftarrow \vbf^{(T_{\mathrm{pow}})}$.\;
(Optional) For $\ell=2$ to $k$, use deflation or block power iterations to extract $\ubf_\ell$.\;

\BlankLine
\textbf{Option B: Gradient-descent refinement (iterated $T_{\mathrm{gd}}$ times)}\;
Choose an initial unit vector $\xbf^{(0)}$ (e.g., $\ubf_1$ from Option A).\;
\For{$t=0,1,\ldots,T_{\mathrm{gd}}-1$}{
  $\xbf^{(t+1)} \leftarrow \xbf^{(t)} - \eta\,(2\Ibb_n - \Ccal)\,\xbf^{(t)}$ \tcp*{gradient step}
  $\xbf^{(t+1)} \leftarrow \xbf^{(t+1)} / \|\xbf^{(t+1)}\|_2$
}
Set $\ubf_1 \leftarrow \xbf^{(T_{\mathrm{gd}})}$; repeat with orthogonal initialization to obtain $\ubf_2,\ldots,\ubf_k$.\;

\BlankLine
\Return $\{\ubf_1,\ldots,\ubf_k\}$, with eigenvalue estimates via Rayleigh quotients $\hat{\lambda}_j=\ubf_j^\top \Ccal\,\ubf_j$.
\end{algorithm}

We state the main result in the following theorem:
\begin{theorem}
\label{thm: pcapowermethod}
    Given a dataset with $m$ samples and $n$ features 
    \begin{align*}
    \mathcal{X} = \begin{pmatrix}
        \xbf^1_1 & \xbf^1_2 & \cdots & \xbf^1_n \\
        \xbf^2_1 & \xbf^2_2 & \cdots & \xbf^2_n \\
        \vdots & \vdots & \ddots & \vdots \\
        \xbf^m_1 & \xbf^m_2 & \cdots & \xbf_n^m
    \end{pmatrix}
\end{align*}
with the covariance matrix $\mathcal{C}$ as defined above. Let the eigenvectors of $\mathcal{C}$ be $\ket{\lambda_1},\ket{\lambda_2},...,\ket{\lambda_n}$ and corresponding eigenvalues be $\lambda_1 > \lambda_2 > ... > \lambda_n$. Define $\Delta = \max_i \{  |\lambda_i - \lambda_{i+1}| \}_{i=1}^{r}$. The $r$ principal components $\ket{\lambda_1},\ket{\lambda_2},...,\ket{\lambda_r} $ of $\mathcal{X}$ can be obtained in complexity 
$$ \mathcal{O}\Big( \log(mn) \frac{1}{(\Delta \gamma) ^r}   \textcolor{black}{\log^r \big(\frac{n}{\epsilon} \big)} \log^r \frac{1}{\epsilon}    \Big)$$
The eigenvalues $\lambda_1,\lambda_2,...,\lambda_r $ can be estimated with complexity 
$$ \mathcal{O}\Big( \log(mn) \frac{1}{ ( \Delta \gamma)^r} \frac{1}{\epsilon}  \textcolor{black}{\log^r \big(\frac{n}{\epsilon} \big)} \log^r \frac{1}{\epsilon}    \Big)$$
\end{theorem}

In reality, the value of $r$ is typically small, for example $r=2,3$ is common, so our method achieves a polylogarithmic running time on all parameters, providing exponential speed-up compared to previous works \cite{lloyd2014quantum, tang2018quantum,tang2021quantum}.  A crucial factor presented above is $\Delta$, which depends on the gap between eigenvalues. The best regime for this power method-based framework is apparently when $\Delta = \mathcal{O}(1)$. For $\Delta$ being $\frac{1}{\rm polylog(n)}$, our complexity is still efficient. 

The method introduced above achieves polylogarithmic scaling in all parameters, which is major improvement over existing results. However, as we pointed out, the complexity depends on $\Delta$, and if $\Delta$ is polynomially small in the inverse of dimension $n$, the advantage would vanish. In the following section, we introduce another approach, based on redefining the PCA problem as a convex optimization problem, thus can be solved by gradient descent. The complexity of this approach does not depend on the gap $\Delta$, which provides a supplementary framework to this section.  We note that in recent work \cite{nghiem2025quantum1}, the author proposed a new algorithm for PCA, which is also based on gradient descent. However, their approach is technically different from ours, as they encode a vector, say $\xbf = \sum_{i=1}^n x_i \ket{i-1}$ in a diagonal operator $\bigoplus_{i=1}^n x_i$. Here, instead, we embed the vector $\xbf$ into a density matrix-like operator $\xbf \xbf^\dagger$. This strategy has also appeared in recent work \cite{nghiem2025new2}, where the author introduced a new quantum linear solver, also built on gradient descent. In particular, it also appeared in the relevant work \cite{nghiem2023improved}, where they outlined an improved quantum algorithm for gradient descent, aiming at polynomial optimization. In fact, as will be shown below, the function we are going to optimize has the same form as those considered in \cite{nghiem2023improved}, therefore, we can use the same line of reasoning to analyze the complexity. 
\begin{table*}[t]
    \centering
    \begin{tabular}{|c|l|}
    \hline
    \textbf{Method} & \textbf{Complexity} \\
    \hline
    First approach (\cref{sec: PCApowermethod}) 
    & $ \mathcal{O}(\log(mn)\log^2(n/\epsilon)\log^2(1/\epsilon)/\Delta^2)$ \\
    \hline
    Second approach (\cref{sec: PCAgradientdescent}) 
    & $ \mathcal{O}(\log(mn)\log^3(1/\epsilon) / \epsilon^2)$ \\
    \hline
    Ref.~\cite{lloyd2014quantum} 
    & $\mathcal{O}(\log(mn)/\epsilon^3)$ \\
    \hline
    Ref.~\cite{nghiem2025new} 
    & $\mathcal{O}(m\log(n)\log^6(n/\epsilon)/(\epsilon\Delta)^4)$ \\
    \hline
    Ref.~\cite{tang2021quantum} 
    & $\mathcal{O}(1/\epsilon^6 + \log(mn)/\epsilon^4)$ \\
    \hline
    \end{tabular}
    \caption{Table summarizing our result and relevant works of \cite{lloyd2014quantum, nghiem2025new, tang2021quantum}. As we can see, our first approach achieves exponential speed-up with respect to $1/\epsilon$ compared to previous works, meanwhile further exponential speed-up with respect to $m$ (the number of sample data) compared to \cite{nghiem2025new}. }
    \label{tab: pca}
\end{table*}

\section{Finding principal components based on gradient descent}
\label{sec: PCAgradientdescent}
To begin, we remind our readers that we are first interested in the top eigenvector of the covariance matrix $\mathcal{C}$, the eigenvector that corresponds to the largest eigenvalue. Since $\mathcal{C}$ is positive semidefinite and without loss of generalization, we assume that its eigenvalues $\lambda_1,\lambda_2,...,\lambda_n$ are bounded between $0 $ and $1$.  
To find it, we define $f(\xbf) =  - \frac{1}{2}\xbf^T C \xbf $ and consider the following optimization problem:
\begin{align}
    \min_{\xbf} \ f(\xbf) 
\end{align}
which can be solved by gradient descent algorithm -- a very popular method widely used in many domains of science and engineering. Its execution is simple as the following. First we randomize an initial point $\xbf_0$, then at $t$-th step, iterate the following procedure:
\begin{align}
    \xbf_{t+1} = \xbf_t - \eta \bigtriangledown f(\xbf_t)
\end{align}
where $\eta$ is the hyperparameter. The total iteration step $T$ is typically user-dependent. Some results \cite{nesterov1983method,nesterov2013introductory,boyd2004convex} have established convergence guarantee for the gradient descent algorithm. If the given function is convex, a local minima is also global minima, so by choosing $T = \mathcal{O}\big(  \frac{1}{\epsilon} \big)$ suffices to ensure that $\xbf_T$ is $\epsilon$ close to the true minima of $f(\xbf)$. Meanwhile, for strongly convex functions, $T$ is further improved to $\mathcal{O}\big( \log \frac{1}{\epsilon}\big)$. In our context, the objective function $f(\xbf)$ is convex, as its Hessiasn is $C$ which is positive-semidefinite. To make the objective function become strongly convex, we can add a regularization term to $f(\xbf)$, and by a slight abuse of notation, we obtain a new objective function or we can try this function $f(\xbf) = \lambda( ||\xbf||^2-1) - \frac{1}{2} \xbf^T C \xbf $ -- which can be a strongly convex function because its Hessian is:
\begin{align}
 2 \lambda \Ibb - C
\end{align}
For a sufficiently large $\lambda$, e.g., $\lambda \geq ||C||_o = \lambda_{\max} (C)$ (where $||.||_o$ refers to operator norm) then the Hessian is positive, both upper bounded (by $2-\lambda_{\max}(C) $) and lower bounded (by $2-\lambda_{\min}(C)$). In the following, we choose $\lambda = 1$ for simplicity.

As mentioned previously, our strategy relies on the embedding of a vector $\xbf$ into a density matrix-like operator, $\xbf \xbf^\dagger$. In this convention, the gradient descent algorithm updates as following:
\begin{align}
    (\xbf \xbf^\dagger)_{t+1} \equiv \xbf_{t+1} \xbf_{t+1}^\dagger
    \end{align}
Given that $\xbf_{t+1} = \xbf_t - \eta \bigtriangledown f(\xbf_t)$ from the regular gradient descent, by a simple algebraic procedure, we have that the above operator is:
\begin{align}
    \xbf_t \xbf_{t+1} -\eta \xbf_t \bigtriangledown^\dagger f(\xbf_t) - \eta \bigtriangledown f(\xbf_t) \xbf_t^\dagger + \eta^2 \bigtriangledown f(\xbf_t) \bigtriangledown^\dagger f(\xbf_t)
\end{align}
Because the function is $f(\xbf) =  ||\xbf||^2 - \frac{1}{2} \xbf^T C \xbf$, its gradient is:
\begin{align}
    \bigtriangledown f(\xbf)= 2\xbf  -  \mathcal{C} \xbf  = (2\Ibb_n-\mathcal{C})\xbf
\end{align}

Substituting to the above equation, we obtain:
\textcolor{black}{
\begin{align}
      \xbf_{t+1} \xbf_{t+1}^\dagger &= ((1-2\eta)\Ibb_n+\eta \mathcal{C})\xbf_t\xbf_t^\dagger((1-2\eta)\Ibb_n+\eta \mathcal{C})
\end{align}
}
\textcolor{black}{In the previous section, we have obtained the block encoding of $\mathcal{C}$. As the block encoding of $\Ibb_n$ is simple to prepare (see \cref{def: blockencode}), the block encoding of $(1-2\eta)\Ibb_n$ and $\eta \mathcal{C}$ can be prepared by \cref{lemma: scale}. \cref{lemma: sumencoding} allows us to prepare $\frac{1}{2}((1-2\eta)\Ibb_n+\eta \mathcal{C})$. The factor $2$ can be removed using \cref{lemma: amp_amp}. Then we use \cref{lemma: product} to construct the block encoding of:}
\textcolor{black}{
\begin{align}
       ((1-2\eta)\Ibb_n+\eta \mathcal{C})\xbf_t\xbf_t^\dagger((1-2\eta)\Ibb_n+\eta \mathcal{C})
\end{align}\label{18}
}
which is exactly $\xbf_{t+1}\xbf_{t+1}^\dagger$. Thus, beginning with some initial operator $\xbf_0 \xbf_0^\dagger$, which can be block-encoded by simply using \cref{lemma: improveddme} with an arbitrary unitary $U_0$ that generates the state $\ket{\bf 0} \xbf_0 + \ket{\rm Redundant}$, then we can iterate the above procedure for a total of $T$ times, which produces the block encoding of  $\xbf_T \xbf_T^\dagger$. To obtain the state $\ket{\xbf_T}$, we take such block encoding and apply it to some state $\ket{\alpha}$, according to \cref{def: blockencode}, we obtain the state $ \ket{\bf 0} (\xbf_T \xbf_T^\dagger) \ket{\alpha} + \ket{\rm Garbage}$. Measurement of the ancilla and post-select in $\ket{\bf 0}$ yields the state $\ket{\xbf_T}$.

\textcolor{black}{To analyze the complexity, we recall that the complexity for producing the covariance matrix $\mathcal{C}$ is $\mathcal{O}( \log mn)$.  Thus, the complexity in obtaining the block encoding of $\frac{1}{2}((1-2\eta)\Ibb_n + \eta\mathcal{C})$ is the same, $\mathcal{O}( \log mn)$. An application of \cref{lemma: amp_amp} to transform $\frac{1}{2}((1-2\eta)\Ibb_n + \eta\mathcal{C}) \longrightarrow (1-2\eta)\Ibb_n + \eta\mathcal{C}$ incurs further complexity $\mathcal{O}( \log \frac{1}{\epsilon}\big)$. Let $\mathcal{T}_t$ denote the complexity of obtaining the block encoding of $\xbf_t\xbf_t^\dagger$. In~\cref{18}, the operator $\xbf_t\xbf_t^\dagger$ appears 1 times, the operator  $(1-2\eta)\Ibb_n + \eta\mathcal{C}$ appears 2 times, therefore, the complexity for producing block encoding of $ \xbf_{t+1} \xbf_{t+1}^\dagger$ is }
\textcolor{black}{
$$\mathcal{T}_{t+1} = 2 \mathcal{O}\big(  \log (mn) \log (\frac{1}{\epsilon}) \big)  +  \mathcal{T}_t  $$ Using induction, we have 
$$\mathcal{T}_t =2 \mathcal{O}\big( \log (\frac{1}{\epsilon})\log (mn) \big) + \mathcal{T}_{t-1}  $$ and thus 
\begin{align*}
    \mathcal{T}_{t+1} &=  \big( 4\log (\frac{1}{\epsilon})) \mathcal{O}( \log mn) + \mathcal{T}_{t-1}
\end{align*}
Continuing the process, we have 
\begin{align*}
    \mathcal{T}_{t} &=  \big( 2t\log (\frac{1}{\epsilon})) \mathcal{O}( \log mn)+\mathcal{T}_0
\end{align*} 
where $\mathcal{T}_0$ is the complexity for producing $\xbf_0\xbf_0^\dagger$ , which is $\mathcal{O}(\log n)$ due to an application of \cref{lemma: improveddme}. So for a total of $T$ iteration steps, the complexity is $\mathcal{O}\big( \big( T\log \frac{1}{\epsilon} \big)  \log mn \big)$. Because our objective function $f(\xbf)$ is strongly convex as pointed out before, the value of $T$ can be $\mathcal{O}\big( \log \frac{1}{\epsilon}\big) $, yielding a final complexity $ \mathcal{O}\Big(  \log^2 \big(\frac{1}{\epsilon} \big)  \log mn  \Big) $ for producing $\ket{\xbf_T}$, which is an approximation to $\ket{\lambda_1}$ -- the eigenvector corresponding to the largest eigenvalue of $\mathcal{C}$. }

Now we show how to find the next eigenvector, $\ket{\lambda_2}$. We use the same strategy as in the previous section, where our aim was to find the top eigenvector of $\mathcal{C}- \lambda_1 \ket{\lambda_1}\bra{\lambda_1}$, so we need a tool similar to \cref{lemma: extensionlemmalargestsmallest} . In the \cref{sec: reviewpowermethod}, \cref{sec: extensiongradientdescent} we show the following:
\begin{lemma}\label{lemma: extensiongradientdescent}
    Given that the block encoding of $\xbf_T \xbf_T^\dagger $ can be obtained by the above procedure, there is a quantum procedure that outputs an $\epsilon$-approximated block encoding of $\lambda_1 \ket{\lambda_1}\bra{\lambda_1}$. The complexity of this procedure  is $\mathcal{O}\big(4 \log^2(\frac{1}{\epsilon}) \frac{1}{\epsilon} \cdot \log mn  \big) $
\end{lemma}
The (approximated) block-encoded operator $\lambda_1 \ket{\lambda_1}\bra{\lambda_1} $ can be transformed into $ \frac{1}{2 }\lambda_1 \ket{\lambda_1}\bra{\lambda_1}  $ simply using \cref{lemma: scale}, which can then be used with the already-have block encoding of $\frac{1}{2}\mathcal{C}$ and \cref{lemma: sumencoding}   to construct the block encoding of $\varpropto \big(  \mathcal{C} - \lambda_1 \ket{\lambda_1}\bra{\lambda_1} \big)$. As discussed in previous section, this operator has $\lambda_2$ being the maximum eigenvalue and corresponding eigenvector is $\ket{\lambda_2}$, so we can first convert it to a convex optimization problem as we did from the beginning of \cref{sec: PCAgradientdescent}, and then repeat the  procedure as above, to find $\ket{\lambda_2}$, $\ket{\lambda_3},..., \ket{\lambda_r}$ -- the $r$ principal components. We summary the result of this section in the following theorem.
\begin{theorem}
\label{thm: pcagradientdescent}
    Given a dataset with $m$ samples and $n$ features 
    \begin{align*}
    \mathcal{X} = \begin{pmatrix}
        \xbf^1_1 & \xbf^1_2 & \cdots & \xbf^1_n \\
        \xbf^2_1 & \xbf^2_2 & \cdots & \xbf^2_n \\
        \vdots & \vdots & \ddots & \vdots \\
        \xbf^m_1 & \xbf^m_2 & \cdots & \xbf_n^m
    \end{pmatrix}
\end{align*}
with the covariance matrix $\mathcal{C}$ as defined above. Let the eigenvectors of $\mathcal{C}$ be $\ket{\lambda_1},\ket{\lambda_2},...,\ket{\lambda_n}$ and corresponding eigenvalues be $\lambda_1 > \lambda_2 > ... > \lambda_n$. The $r$ principal components $\ket{\lambda_1},\ket{\lambda_2},...,\ket{\lambda_r} $ of $\mathcal{X}$ can be obtained in complexity 
$$ \mathcal{O}\Big(  \log^{2r-1} (\frac{1}{\epsilon}) \big(\frac{4}{\epsilon}\big)^r \log mn \Big) $$
The eigenvalues $ \lambda_1,\lambda_2,...,\lambda_r$ can be estimated to accuracy $\epsilon$ in complexity 
$$ \mathcal{O}\Big(  \log^{2r-1} (\frac{1}{\epsilon}) \big(\frac{4}{\epsilon}\big)^r \frac{1}{\epsilon}\log mn \Big)  $$
\end{theorem}
Comparing to the complexity of the previous section, we can see that this gradient descent-based approach does not depend on the gap $\Delta$ between eigenvalues, as we expected. However, there is a trade-off on the inverse of error, as this approach exhibits polynomial dependence on $\frac{1}{\epsilon}$. 

\subsection{Lower bound on finding principal components with QSVT}
Majority of the algorithms for finding principal components employ QSVT, including the power method and gradient descent method that we proposed. We prove that the complexity of obtaining principal components by employing QSVT (with polynomial approximation) is lower bounded by $\Omega(\frac{1}{\Delta})$, where $\Delta=\lambda_1-\lambda_2$. 

Polynomial approximation is one of the key components of QSVT. Faster algorithms have been established with lower bounds on polynomial approximation~\cite{sachdeva2014faster}. Polynomial approximation on the absolute value function $f(x)=|x|, (x\in[0,1])$ are used in various problems such as trace distance estimation, property testing. Using the best $d$-degree polynomial we can achieve $\Theta(\frac{1}{d})$ error~\cite{sachdeva2014faster}. With little adjustment, we can deduce the lower bound of step size function approximation. 

\begin{lemma} \label{lem: step}
    Suppose that $g:[0,1]\rightarrow[0,1]$ satisfies $g(x)=0 (x\leq\lambda_1)$, $g(x)=1 (x>\lambda_2)$, and the polynomial $h$ approximate $g$. The minimum degree of $h$ is 
    \begin{equation}
        \Omega(\frac{1}{\lambda_1-\lambda_2})=\Omega(\frac{1}{\Delta}).
    \end{equation}
\end{lemma}
\begin{proof}
    Suppose that $\text{deg}(h) = d$. And define $f(x)=2xg(x-\frac{\lambda_1+\lambda_2}{2})-x$. $f$ satisfies
    \begin{align}
        & f(x) = |x| \space (|x|>\Delta) \\
        & |f(x)| \leq \Delta \space (|x| \leq \Delta).
    \end{align}
    Implying that $f$ approximates the absolute value function with $\Delta$-error. So, $\text{deg}(f) = \Omega(\frac{1}{\Delta})$. Then, $\text{deg}(g) = \Omega(\frac{1}{\Delta})$. Which concludes the proof.
\end{proof}

Any method for finding principal components with QSVT needs the polynomial approximation of the step function with step size $\Delta$. The complexity of QSVT is propotional to the degree of the polynomial and Lemma~\ref{lem: step} implies that the degree requires to be at least $\Omega(\frac{1}{\Delta})$.  


\section{Solving linear algebraic equations}
\label{sec: solvinglinearequation}
The linear system is defined as $A\xbf = \textbf{b}$. Similarly to previous contexts \cite{harrow2009quantum, childs2017quantum}, we assume without loss of generality that $A$ is $s$-sparse Hermitian and its eigenvalues are falling between $(-1,1)$. Suppose that a unique solution exists, it is given by $\xbf = A^{-1} \textbf{b}$. For concreteness, we further define:
\begin{align}
    A = \begin{pmatrix}
        A_{11} & A_{12} &  \cdots & A_{1n} \\
        A_{21} & A_{22} & \cdots & A_{2n} \\
        \vdots & \vdots & \ddots & \vdots \\
        A_{n1} &  A_{n2} & \cdots & A_{nn}
    \end{pmatrix}, \ \textbf{b }= \begin{pmatrix}
        b_1 \\
        b_2\\
        \vdots \\
        b_n
    \end{pmatrix}
\end{align}
In quantum context, the goal is to obtain the state $\ket{\xbf} \varpropto A^{-1} \textbf{b}$. We recall that at the beginning of~\cref{sec: quantumalgorithm}, we showed how to construct the block encoding of $\mathcal{X}^T \mathcal{X}$ (see the discussion above \cref{lemma: improveddme}), and the same technique can be used to construct the block encoding of $A^T A$. More specifically, we first use \cref{lemma: stateprepration} create the state:
\begin{align}
    \ket{\Phi} = \sum_{i=1}^n \sum_{j=1}^n A_{ij }\ket{i}\ket{j} 
\end{align}
in complexity $\mathcal{O}(\log sn) = \mathcal{O}(\log sn)$. The reason for the appearance of $s$ -- the sparsity of $A$, is because by definition, it is the maximum number of non-zero entries in each row or column of $A$. By tracing out the first register that holds $\{ \ket{i}\} $ of the above state, we obtain the density state $A^T A$, which can be block-encoded via \cref{lemma: improveddme}. To proceed, we need \cref{lemma: negative}. Since the matrix $A^T A$ is positive, we can apply the above lemma (with $\mathcal{M} = A^T A$ and $c= \frac{1}{2}$) to obtain the block encoding of $ \frac{1}{2\kappa_M^{c}} (A^T A)^{-c }  $. We mention the following spectral property. Let $\{ \lambda_i, \ket{\lambda_i} \}_{i=1}^n$ denotes the spectrum, including eigenvalues and corresponding eigenvectors of $A$, then $\{ \lambda_i^2, \ket{\lambda_i} \}_{i=1}^n$ is the spectrum of $A^T A$. Therefore, if $A$ is positive semidefinite, or $\lambda_i \geq 0$ for all $i=1,2,...,n$, then $\sqrt{\lambda_i^2} = \lambda_i$, so $ (A^T A)^{-1/2} = A^{-1}$. Additionally, if $\kappa$ is the conditional number of $A$, which is the ratio between the largest and smallest eigenvalue of $A$, then the conditional number $\kappa_M$ of $A^T A$ is $\kappa_M = \kappa^2$. Provided that we can prepare the state $\textbf{b} \equiv \ket{\textbf{b}}$ (assuming to have unit norm for convenience), e.g., via \cref{lemma: stateprepration}, we can then take the block encoding of $\frac{1}{2\kappa_M^{c}} (A^T A)^{-c }  = \frac{1}{2\kappa} A^{-1} $ and apply it to $\ket{\textbf{b}}$. According to \cref{def: blockencode}, we obtain the state:
\begin{align}
    \ket{\bf 0} \frac{1}{2\kappa}A^{-1} \ket{\textbf{b}} + \ket{\rm Garbage}
\end{align}
Measuring the ancilla and post-select on $\ket{\bf 0}$, we obtain the state $\varpropto A^{-1}\ket{\textbf{b}}$. The success probability of this measurement is $ \frac{1}{4\kappa^2}|| A^{-1} \ket{\textbf{b}}||^2 = \mathcal{O} \big( \frac{1}{4\kappa^2}\big) $, which can be improved quadratically faster using the amplitude amplification technique \cite{brassard2002quantum}. The complexity of approach is simply the product of the complexity of producing the block-encoded $A^T A$, of using \cref{lemma: negative} (with $c = 1/2$ and $\mathcal{M} = A^T A$), and of measuring at the final step to obtain $\ket{\xbf}$. Thus, the total complexity is is $\mathcal{O}\Big( \kappa^3 \log(s) \log^2 \frac{\kappa^{3/2}}{\epsilon}  \Big)$.

The above procedure works only when $A$ is positive semidefinite, because in such a case $ (A^T A)^{-1/2} = A^{-1} $. For a general $A$, it might not hold and we can modify the above algorithm as follows. As the eigenvalues $\{\lambda_i\}_{i=1}^n$ of $A$ are between $(-1,1)$, the shifted matrix $\frac{1}{2}\big( \Ibb_n + A\big)$ has eigenvalues $\{ \frac{1}{2}(1+\lambda_i)\}_{i=1}^n $ falling between $(0,1)$, which indicates that this matrix is positive semidefinite. It is also clear that the conditional number of this shifted matrix is upper bounded by $2$, which is very small. The matrix representation of this matrix is:
\begin{align}
    \frac{1}{2}\big( \Ibb_n + A\big) =\frac{1}{2} \begin{pmatrix}
        A_{11}+1 & A_{12} &  \cdots & A_{1n} \\
        A_{21} & A_{22}+1 & \cdots & A_{2n} \\
        \vdots & \vdots & \ddots & \vdots \\
        A_{n1} &  A_{n2} & \cdots & A_{nn}+1
    \end{pmatrix}
\end{align}
which is a slight adjustment of the original matrix $A$. Thus, by using the same procedure that we used to prepare the block encoding of $\mathcal{X}^T \mathcal{X}$ from the beginning of \cref{sec: quantumalgorithm}, we can obtain the block encoding of $\frac{1}{2}\big( \Ibb_n + A\big)^T \frac{1}{2}\big( \Ibb_n + A\big) $. Because $\frac{1}{2}\big( \Ibb_n + A\big) $ is positive semidefinite,  we have pointed out in the previous paragraph the property that $ \Big( \frac{1}{2}\big( \Ibb_n + A\big)^T \frac{1}{2}\big( \Ibb_n + A\big) \Big)^{1/2} = \frac{1}{2}\big( \Ibb_n + A\big) $. Thus, an application of the following \cref{lemma: positive} with $\mathcal{M} =\frac{1}{2}\big( \Ibb_n + A\big)^T \frac{1}{2}\big( \Ibb_n + A\big)  $ and $c= \frac{1}{2}$ allows us to construct the block encoding of $ \frac{1}{4}\big( \Ibb_n + A\big)  $. As noted in the \cref{def: blockencode}, an identity matrix $\Ibb_n$ can be simply block-encoded, we can then use \cref{lemma: scale} to construct the block encoding of $ \frac{1}{4}\Ibb_n$. Then we use \cref{lemma: sumencoding} to construct the block encoding of:
\begin{align}
    \frac{1}{2} \Big( \frac{1}{4}\big( \Ibb_n + A\big) -  \frac{1}{4}\Ibb_n\Big) = \frac{A}{8}
\end{align}
Then applying \cref{lemma: negative} (with $c=1$) yields the block encoding of $\frac{1}{2\kappa} A^{-1}$, which can then be used to obtain the state $\varpropto A^{-1} \ket{\textbf{b}}$ as we discussed in the previous paragraph. 

To analyze the complexity, we recall that the complexity to produce the block encoding of $\frac{1}{2}\big( \Ibb_n + A\big)^T \frac{1}{2}\big( \Ibb_n + A\big)  $ is $\mathcal{O}\big( \log sn \big)$. Then we use \cref{lemma: positive} to transform it into $ \frac{1}{4}\big( \Ibb_n + A\big) $, and the complexity of this step is $\mathcal{O}\big(  \log(sn) \log^2 \frac{1}{\epsilon} \big)$ (we ignore the conditional number of $\frac{1}{2}\big( \Ibb_n + A\big)^T \frac{1}{2}\big( \Ibb_n + A\big) $ because we pointed out before that it is upper bounded by $2$, which is very small). Then we use \cref{lemma: sumencoding} to construct the block encoding of $ \frac{A}{8}$, which incurs a further $\mathcal{O}(1)$ cost because the block encoding of $\Ibb_n$ has $\mathcal{O}(1)$ cost (see~\cref{def: blockencode}), and given that \cref{lemma: sumencoding} use the block encoding of $  \frac{1}{4}\big( \Ibb_n + A\big)$ one time, so the complexity to produce $\frac{A}{8}$ is $ \mathcal{O}\big(  \log(sn) \log^2 \frac{1}{\epsilon} \big)$. The next step is using \cref{lemma: negative} (with $c= 1$ and $\mathcal{M} = A/8$) to transform the block-encoded $\frac{A}{8} $ into $ \frac{A^{-1}}{\kappa}$, which results in complexity:
$$  \mathcal{O}\Big( \kappa^2 \log(sn) \log^2\big( \frac{\kappa^2}{\epsilon}\big) \log^2 \frac{1}{\epsilon}  \Big)$$
\begin{algorithm}[t]
\caption{Quantum Algorithm for Solving Linear Equations}
\label{alg:solvinglinearequation}
\KwIn{Classical knowledge $\{A^1, A^2, \dots, A^n\}$, vector $\ket{\mathbf{b}}$}
\KwOut{$\propto A^{-1} \ket{\mathbf{b}}$}

\If{$A$ is positive semidefinite}{
    Prepare state $\sum_{i=1}^n \ket{i} A^i$ \tcp*{\cref{lemma: stateprepration}}
    Apply improved DME to obtain $A^T A$ \tcp*{\cref{lemma: improveddme}}
    Apply negative power lemma to obtain $A^{-1}$ \tcp*{\cref{lemma: negative}}
    Block-encode $A^{-1}$ and apply to $\ket{\mathbf{b}}$ \tcp*{\cref{def: blockencode}}
}
\Else{
    Transform $A \leftarrow \frac{\mathbb{I}_n + A}{2}$ \;
    Prepare state $\sum_{i=1}^n \ket{i} \frac{1}{2}(\mathbb{I}_n + A)^i$ \tcp*{\cref{lemma: stateprepration}}
    Apply improved DME to obtain $(\mathbb{I}_n + A)^T (\mathbb{I}_n + A)$ \tcp*{\cref{lemma: improveddme}}
    Apply positive power lemma to obtain $\mathbb{I}_n + A$ \tcp*{\cref{lemma: positive}}
    Apply sum encoding to recover $A$ \tcp*{\cref{lemma: sumencoding}}
    Apply negative power lemma to obtain $A^{-1}$ \tcp*{\cref{lemma: negative}}
    Block-encode $A^{-1}$ and apply to $\ket{\mathbf{b}}$ \tcp*{\cref{def: blockencode}}
}
\Return{$\propto A^{-1} \ket{\mathbf{b}}$}
\end{algorithm}
We summarize the the result of this section in the following:
\begin{theorem}[Refined Quantum Linear Solver]
Let the linear system be $A\xbf = \textbf{b}$ where $A$ is an $s$-sparse, Hermitian matrix of size $n \times n$, with conditional number $\kappa$, and  $\textbf{b}$ is unit. Then there is a quantum algorithm outputting the state $\ket{\xbf } \varpropto  A^{-1}\textbf{b} $ in complexity
$$ \mathcal{O}\Big( \kappa^2 \log(sn) \log^2\big( \frac{\kappa^2}{\epsilon}\big) \log^2 \frac{1}{\epsilon}  \Big)$$
In the case $A$ is positive-semidefinite, the complexity is:
$$ \mathcal{O}\Big( \kappa^3 \log (sn) \log^2 \frac{\kappa^{3/2}}{\epsilon}  \Big)$$
\end{theorem}

\section{Direct quantum simulation}
\label{sec: directquantumsimulation}
\begin{algorithm}[t]
\caption{Quantum Simulation Algorithm}
\label{alg:hamiltoniansimulation}
\KwIn{Classical knowledge $\{H^1, H^2, \dots, H^n\}$, simulation time $t$}
\KwOut{$\exp(-i H t)$}

\If{$H$ is positive semidefinite}{
    Prepare state $\sum_{i=1}^n \ket{i} H^i$ \tcp*{\cref{lemma: stateprepration}}
    Apply improved DME to obtain $H^T H$ \tcp*{\cref{lemma: improveddme}}
    Apply positive power lemma to obtain $H$ \tcp*{\cref{lemma: positive}}
    Perform Hamiltonian simulation to obtain $\exp(-i H t)$ \tcp*{Ref.~\cite{low2019hamiltonian}}
}
\Else{
    Transform $H \leftarrow \frac{\mathbb{I}_n + H}{2}$ \;
    Prepare state $\sum_{i=1}^n \ket{i} \frac{1}{2}(\mathbb{I}_n + H)^i$ \tcp*{\cref{lemma: stateprepration}}
    Apply improved DME to obtain $(\mathbb{I}_n + H)^T (\mathbb{I}_n + H)$ \tcp*{\cref{lemma: improveddme}}
    Apply positive power lemma to obtain $\mathbb{I}_n + H$ \tcp*{\cref{lemma: positive}}
    Apply sum encoding to recover $H$ \tcp*{\cref{lemma: sumencoding}}
    Perform Hamiltonian simulation to obtain $\exp(-i H t)$ \tcp*{Ref.~\cite{low2019hamiltonian}}
}
\Return{$\exp(-i H t)$}
\end{algorithm}

As mentioned earlier, the key objective of quantum simulation is to (approximately) construct the evolution operator $\exp(-i Ht)$. Classically, in order to obtain $\exp(-i H t)$, one needs to diagonalize $H$ to find the eigenvalues $\{\lambda_i\}$ and corresponding eigenvectors $\{ \ket{\lambda_i}\} $. The evolution operator can be obtained as $\sum \exp(-i \lambda_i t) \ket{\lambda_i}\bra{\lambda_i}$. Apparently, this approach takes at least linear time, because of the diagonalization step. Furthermore, classically, one needs to know $H$ explicitly in order to perform diagonalization. This fact has inspired us to ask the following question: If we know the entries of $H$ classically, can we perform the quantum simulation? This input model is different from the two models described above, because neither are we provided with an oracle, nor can the Hamiltonian be expressed as linear combination of unitaries. It turns out that we can efficiently simulate the Hamiltonian provided we know the entries classically. The answer is, in fact, a corollary of the refined quantum linear solving algorithm we outlined in the previous section. 

In the above, we have shown how to obtain the block encoding of $\varpropto A$, provided that its columns are classically known. In a similar manner, if the columns of the Hamiltonian $H$ of interest are known, then we can follow the same procedure as above and construct the block encoding of $\varpropto H$, with complexity $\mathcal{O}\big( \log (n) \log^2 \frac{1}{\epsilon} \big)$. We note that similar to existing works, we assume the norm of $H$ is less than $1$. Otherwise, we can consider a rescaled Hamiltonian $\frac{H}{|H|_{\max}}$ where $|H|_{\max}$ is the maximum element of $H$, and then aim to simulate for a longer time $|H|_{\max} t$. To obtain the operator $\exp(-i H t)$, we can apply the results of \cite{low2017optimal,low2019hamiltonian, gilyen2019quantum}. More concretely, we leverage the \cref{lemma: theorem56} and choose the polynomial $P$ to be an approximation of $\exp(-i H t)$ (Jacobi-Anger expansion), then obtain the transformation:
\begin{align}
    \begin{pmatrix}
        H & \cdot \\
        \cdot & \cdot 
    \end{pmatrix} \longrightarrow \begin{pmatrix}
        P(H) & \cdot \\
        \cdot & \cdot 
    \end{pmatrix} \approx \begin{pmatrix}
        \exp(-i H t) & \cdot \\
        \cdot & \cdot 
    \end{pmatrix} 
\end{align}
According to Theorem 58 of \cite{gilyen2019quantum}, this polynomial $P$ has degree $\mathcal{O}\Big( |H|_{\max} t + \frac{\log (1/\epsilon)}{\log\big(e + \log(1/\epsilon)/t)  \big)}  \Big)$. Per \cref{lemma: theorem56}, we obtain the (block-encoded) simulation operator $\exp(-iH t)$ (as above) with complexity $\mathcal{O}\Big( \log (sn) \log^2 \frac{1}{\epsilon}\big(   |H|_{\max} t + \frac{\log (1/\epsilon)}{\log\big(e + \log(1/\epsilon)/t)  \big)}  \big) \Big) $. As been established in \cite{berry2007efficient}, this is optimal with respect to time $t$ and dimension $n$, while being nearly optimal in the inverse of error tolerance. 

\section{Quantum simulation by solving linear equation}
The above method features a direct simulation, where we construct the evolution operator $\exp(-i H t)$ directly, leveraging the technique outlined in previous context. Here we consider an alternative, indirect way, which is reducing Schrodinger's equation into a linear equation, for which we can apply the result from previous section. This reduction strategy has been employed in many previous works \cite{childs2021high, berry2017quantum, childs2020quantum}, where the authors considered more general problems, including solving linear, nonlinear ordinary differential equations and partial differential equations. Recall that the Schrodinger's equation is a first-order ordinary differential equation $\frac{\partial \ket{\psi}}{\partial t} = -i H \ket{\psi}$. Dividing the time interval $[0,t]$ into subintervals $[0,\Delta, 2\Delta, ..., N\Delta \equiv t]$ and defining $\ket{\psi}_k = \big( \psi_1(k \Delta), \psi_2(k \Delta), ..., \psi_n(k \Delta) \big)^T$. A simple approximation of the derivative of, say, $\psi_j(k\Delta)$, reads $\frac{\partial \psi_j }{\partial t } |_{k\Delta} = \frac{1}{2\Delta}\big(  \psi_j(k\Delta+\Delta) - \psi_j(k \Delta-\Delta) \big) \longrightarrow \ket{\psi}_{k+1} - \ket{\psi}_{k-1} = (-2i\Delta) H \ket{\psi}_k$. For $k=0$, which is the starting point, we can use $\frac{\partial \psi_j }{\partial t } |_{0} = \frac{1}{\Delta}  \big( \psi_j (\Delta) - \psi_j(0)  \big)   $. So we obtain the following equations:
\begin{align}
    \begin{cases}
      \ket{\psi}_1 -   \ket{\psi}_0  = (-i\Delta) H \ket{\psi}_0 \\
        \ket{\psi}_{2} - \ket{\psi}_{0} = (-2i\Delta) H \ket{\psi}_1\\
        \ket{\psi}_{3} - \ket{\psi}_{1} = (-2i\Delta) H \ket{\psi}_2 \\
        \vdots \\
        \ket{\psi}_{N} - \ket{\psi}_{N-2} = (-2i\Delta) H \ket{\psi}_{N-1}
    \end{cases}
\end{align}
which forms a linear equation: 
\begin{align}
    \begin{pmatrix}
        i\Delta H& \Ibb_n & 0 & 0 & \cdots & 0 \\
        -\Ibb_n & 2i\Delta H & \Ibb_n & 0 & \cdots & 0 \\
        0 & -\Ibb_n & 2i\Delta H & \Ibb_n & \cdots & 0 \\
        0 & 0 & -\Ibb_n & 2i\Delta H & \cdots & 0 \\
        \vdots & \vdots & \vdots & \vdots & \ddots & \vdots \\
        0 & 0 & 0 & -\Ibb_n & 2i\Delta H & \Ibb_n 
    \end{pmatrix} \begin{pmatrix}
        \ket{\psi}_0 \\
        \ket{\psi}_1 \\
        \ket{\psi}_2  \\
        \ket{\psi}_3 \\
        \vdots  \\
        \ket{\psi}_N
    \end{pmatrix}  = \begin{pmatrix}
        \ket{\psi}_0 \\
        0 \\
        0  \\
        0 \\
        \vdots  \\
        0
    \end{pmatrix}
\end{align}

This is a linear system of size $ nN \times nN$, and apparently we can use our refined quantum linear solver to find the state $\varpropto \sum_{k=0}^N \ket{k} \ket{\psi}_k$. The above method used a simple approximation for the derivative, thus reducing the Schrodinger's equation, which is an ODE to a linear equation. We remind that this strategy was already used in \cite{berry2014high} to solve ordinary differential equations. According to them, a more advanced method, namely, general linear multistep method, yields the following equation at each time step $k\Delta$: $\sum_{l=-K}^K \alpha_l \ket{\psi}_{k+l} =  (-i\Delta) \sum_{l=-K}^K H \beta_l \ket{\psi}_{k+l} $. We remark that for $k < K$, $k-K < 0$ therefore we can't use the multistep, instead, we use the approximation as above $\ket{\psi}_{k+1} - \ket{\psi}_{k-1} = (-2i\Delta) H \ket{\psi}_k $ and only use multisteps for $k\geq K$. We thus form an equation:
\begin{align}
\begin{cases}
     \ket{\psi}_1 -   \ket{\psi}_0  = (-i\Delta) H\ket{\psi}_0 \\
        \ket{\psi}_{2} - \ket{\psi}_{0} = (-2i\Delta) H \ket{\psi}_1\\
        \ket{\psi}_{3} - \ket{\psi}_{1} = (-2i\Delta) H \ket{\psi}_2 \\
        \vdots \\
        \ket{\psi}_{K} - \ket{\psi}_{K-2} = (-2i\Delta) H \ket{\psi}_{K-1} \\
    \sum_{l=-K}^K \alpha_l \ket{\psi}_{K+l} =  (-i\Delta) \sum_{l=-K}^K H \beta_l \ket{\psi}_{K+l} \\
    \sum_{l=-K}^K \alpha_l \ket{\psi}_{K+1+l} =  (-i\Delta) \sum_{l=-K}^K H \beta_l \ket{\psi}_{K+1+l} \\
    \sum_{l=-K}^K \alpha_l \ket{\psi}_{K+2+l} =  (-i\Delta) \sum_{l=-K}^K H \beta_l \ket{\psi}_{K+2+l} \\
    \vdots \\
    \sum_{l=-K}^K \alpha_l \ket{\psi}_{N-K+l} =  (-i\Delta) \sum_{l=-K}^K H \beta_l \ket{\psi}_{N-K+l} 
\end{cases}
\label{29}
\end{align}
and thus form a more complicated linear systems. We refer to \cite{berry2014high} for a more detailed representation of this linear system. Because we discretize the time interval and approximate the derivation, there is an error induced, which means that each of state $\ket{\psi}_1,\ket{\psi}_2,..., \ket{\psi}_N$ above has some deviation to the true solution of the original differential equations, at corresponding time step. According to \cite{berry2014high} (see their Section IV), by choosing $N = \mathcal{O}\Big(  \frac{t^{1+1/K}}{\epsilon^{1/K}} \Big) $, then the accumulated error is $\epsilon$, i.e., for all $k=1,2,...,N$, we have that $\big|\big| \ket{\psi}_k - \ket{\psi}_{k}^{\rm true}    \big|\big| \leq \epsilon$, where $\ket{\psi}_{k}^{\rm true}  $ denotes the true solution. In addition, Theorem 7 of \cite{berry2014high} shows that the conditional number of the above system is $\mathcal{O}(N)$. Therefore, an application of our quantum linear solver yields a quantum simulation algorithm with complexity $\mathcal{O}\Big( \kappa^2 \log (nN) \log^2 \big( \frac{\kappa^{2}}{\epsilon}\big) \log^2 \frac{1}{\epsilon}  \Big) = \mathcal{\Tilde{O}}\Big(  \frac{t^{2+2/K}}{\epsilon^{2/K}} \log(n)   \Big)$, where $\mathcal{\Tilde{O}}$ hides the polylogarithmic terms. 

This approach is clearly not as efficient as the direct simulation approach above, especially with respect to time $t$ and inverse of error $1/\epsilon$. However, it does have some implications. First, we recall that in the original quantum linear solving algorithm \cite{harrow2009quantum} (HHL algorithm), the authors proved that the complexity on conditional number $\kappa$ cannot be better than linear, i.e., a sublinear scaling $\kappa^{1-\gamma}$ is not possible. Here, we provide an alternative and much simpler proof to this statement, based on the fact that our quantum simulation method uses a quantum linear solver as a subroutine. We recall from the above that the conditional number of the linear system defined in~\cref{29} is $\kappa = \mathcal{O}(N)$, and the value of $N$ (the number of time steps) is $N \varpropto t^{1+1/K}$, which is sublinear. Therefore, $\kappa \varpropto t^{1+1/K} $. If a quantum linear solver can produce the solution in $\kappa^{1-\gamma}$, it means that it can solve~\cref{29}, which encodes the evolved state at the time $t = N\Delta$, in complexity $\kappa^{1-\gamma} = N^{1-\gamma} = \big( t^{1+1/K} \big)^{1-\gamma} $. By choosing $\gamma$ properly, then $\big( t^{1+1/K} \big)^{1-\gamma} $ can be sublinear in $t$, which means that we can simulate the dynamics of a given quantum system in sublinear time. This violates the well-known no-forwarding theorem, which states that the complexity of simulating quantum system is $\Omega(t)$. Therefore, a quantum algorithm for solving linear system cannot have sublinear scaling in $\kappa$. 

Second, this approach can be extended to time-dependent regime in a straightforward manner, meanwhile the direct approach above cannot. In the time-dependent regime, the Hamiltonian $H$ becomes time-dependent, and we need to modify the linear system by setting $H$ (in~\cref{29}) with $H_{k\Delta}$ -- which is the Hamiltonian at $k$-th time step. More specifically, we obtain the following:
\begin{align}
\begin{cases}
     \ket{\psi}_1 -   \ket{\psi}_0  = (-i\Delta) H_0\ket{\psi}_0 \\
        \ket{\psi}_{2} - \ket{\psi}_{0} = (-2i\Delta) H_1 \ket{\psi}_1\\
        \ket{\psi}_{3} - \ket{\psi}_{1} = (-2i\Delta) H_2 \ket{\psi}_2 \\
        \vdots \\
        \ket{\psi}_{K} - \ket{\psi}_{K-2} = (-2i\Delta) H_{K-1} \ket{\psi}_{K-1} \\
    \sum_{l=-K}^K \alpha_l \ket{\psi}_{K+l} =  (-i\Delta) \sum_{l=-K}^K H_{K+l} \beta_l \ket{\psi}_{K+l} \\
    \sum_{l=-K}^K \alpha_l \ket{\psi}_{K+1+l} =  (-i\Delta) \sum_{l=-K}^K H_{K+1+l} \beta_l \ket{\psi}_{K+1+l} \\
    \sum_{l=-K}^K \alpha_l \ket{\psi}_{K+2+l} =  (-i\Delta) \sum_{l=-K}^K H_{K+2+l} \beta_l \ket{\psi}_{K+2+l} \\
    \vdots \\
    \sum_{l=-K}^K \alpha_l \ket{\psi}_{N-K+l}  =  (-i\Delta) \sum_{l=-K}^K H_{N-K+l} \beta_l \ket{\psi}_{N-K+l} 
\end{cases}
\end{align}
Solving this linear equation yields the state $\varpropto \sum_{k=0}^N \ket{k} \ket{\psi}_k$ -- which includes the evolved states at different time step, from $0$ up to $N \Delta = t$. 

\section{Finding ground state and excited state energy}
\label{sec: findinggroundstateexcitedstate}
In this context, we remark that for a given Hamiltonian $H$ (with its operator norm $||H|| \leq 1$, which can be guaranteed by scaling, as we mentioned previously), the matrix $\frac{1}{2}\big( \Ibb  - H \big)$ is positive and its spectrum is the `reverse' of $H$. More specifically, if we denote its lowest eigenvalue/eigenvector, i.e., ground state, as $\lambda_{\rm ground}/\ket{\lambda_{\rm ground}}$, then the matrix $ \frac{1}{2}\big( \Ibb  - H \big)$ has its maximum eigenvalue to be $ \frac{1}{2} (1 - \lambda_{\rm ground} )$, and $\ket{\lambda_{\rm ground}}$ as corresponding eigenvector. Similarly, let the second smallest eigenvalue/eigenvector of $H$, or the first excited state to be $\lambda_{\rm excited}, \ket{ \lambda_{\rm excited}}$, then $\frac{1}{2}\big( 1-\lambda_{\rm excited} \big), \ket{\lambda_{\rm excited}} $ is the second largest eigenvalue/eigenvector of $\frac{1}{2}\big(\Ibb -H\big)$. Therefore, we can directly use the result from principal component analysis algorithm outlined in~\cref{sec: PCApowermethod},~\cref{sec: PCAgradientdescent} to find the largest components of $\frac{1}{2}\big( \Ibb -H\big)$ -- which contains the desired information about ground state and excited state energy. The complexity is hence, the same as appearing in~\cref{thm: pcapowermethod},~\cref{thm: pcagradientdescent}. For the ground state energy estimation, specifically, by setting $r=1$, we have that such energy can be estimated with complexity $\mathcal{O}\Big( \log(n) \frac{1}{\Delta \epsilon}   \textcolor{black}{\log \big(\frac{n}{\epsilon}\big)} \log \frac{1}{\epsilon}    \Big) $ (where $\Delta$ is the gap between ground state energy and first excited state energy, and we remind that $\gamma = |\braket{\psi, \lambda_{\rm ground}}|$ is the overlaps between the ground state and the initially random state $\ket{\psi}$ used in the power method, e.g., see \cref{lemma: largestsmallest} ) using power method approach, and with complexity $\mathcal{O}\Big(  \log (\frac{1}{\epsilon}) \big(\frac{4}{\epsilon^2}\big) \log n \Big)$ using gradient descent approach. \\

\begin{algorithm}[H]
\caption{Quantum Algorithm for Finding Ground and Excited States with Corresponding Energies}
\label{alg:groundstateexcitedstate}
\KwIn{Classical knowledge $\{H^1, H^2, \dots, H^n\}$}
\KwOut{$\lambda_{\rm ground} / \ket{\lambda_{\rm ground}}, \quad \lambda_{\rm excited} / \ket{\lambda_{\rm excited}}$}

\If{$H$ is positive semidefinite}{
    Prepare state $\sum_{i=1}^n \ket{i} H^i$ \tcp*{\cref{lemma: stateprepration}}
    Apply improved DME to obtain $H^T H$ \tcp*{\cref{lemma: improveddme}}
    Apply positive power lemma to obtain $H$ \tcp*{\cref{lemma: positive}}
    Run Quantum PCA algorithm \tcp*{\cref{alg:qpca}}
}
\Else{
    Transform $H \leftarrow \frac{\mathbb{I}_n + H}{2}$ \;
    Prepare state $\sum_{i=1}^n \ket{i} \frac{1}{2}(\mathbb{I}_n + H)^i$ \tcp*{\cref{lemma: stateprepration}}
    Apply improved DME to obtain $(\mathbb{I}_n + H)^T (\mathbb{I}_n + H)$ \tcp*{\cref{lemma: improveddme}}
    Apply positive power lemma to obtain $\mathbb{I}_n + H$ \tcp*{\cref{lemma: positive}}
    Apply sum encoding to recover $H$ \tcp*{\cref{lemma: sumencoding}}
    Run Quantum PCA algorithm \tcp*{\cref{alg:qpca}}
}
\Return{$\lambda_{\rm ground} / \ket{\lambda_{\rm ground}}, \quad \lambda_{\rm excited} / \ket{\lambda_{\rm excited}}$}
\end{algorithm}
Comparison to existing works is summarized in the following table.
\begin{table*}[t]
    \centering
    {
    \begin{tabular}{|c|l|}
    \hline
    \textbf{Method} &  \textbf{Ground State Energy Estimation} \\
    \hline
    Power method
    & $\mathcal{O}(\log(n)\log(1/(\epsilon\gamma))\log(1/\epsilon)/(\Delta\gamma\epsilon))  $  \\
    \hline
    Gradient descent method &  $\mathcal{O}(\log(n)\log(1/\epsilon)/\epsilon^2)$ \\
    \hline
    Ref.~\cite{dong2022ground} 
    & $\mathcal{O}(T_U/(\gamma^2\epsilon))$   \\
    \hline
    Ref.~\cite{lin2020near} 
    & $\mathcal{O}(T_U\log(1/\epsilon)\log(1/\gamma)/(\gamma\epsilon))$ \\
    \hline
    \end{tabular}
    }
    \caption{Table summarizing our result and relevant works \cite{dong2022ground, lin2020near}. $T_U$ denotes the complexity of block-encoding Hamiltonian of interest $H$ used in \cite{lin2020near, dong2022ground}, respectively, and $\gamma = | \braket{\psi, \lambda_{\rm ground} }|$ where $\ket{\psi}$ is some initial random state.  }
    \label{tab: groundstateenergy}
\end{table*}

\section{Ground state preparation }
\label{sec: groundstatepreparation}
The application of the quantum PCA algorithm outlined in ~\cref{sec: PCApowermethod}, \cref{sec: PCAgradientdescent} can reveal the ground/excited state energy, and at the same time, can also output their corresponding eigenstates, as indicated in~\cref{thm: pcapowermethod}, ~\cref{thm: pcagradientdescent}. The complexity for producing the ground state by these methods, respectively, are
\begin{itemize}
    \item power method 
    $$ \mathcal{O}\Big( \log(n) \frac{1}{\Delta \gamma }   \textcolor{black}{\log \big(\frac{n}{\epsilon}\big)} \log \frac{1}{\epsilon}    \Big) $$
    \item gradient descent method
    $$ \mathcal{O}\Big(  \log (\frac{1}{\epsilon}) \big(\frac{4}{\epsilon}\big) \log n \Big)  $$
\end{itemize}
Here, we point out another way to prepare the ground state, given that we can construct the block encoding of $H$ from its classical description. The idea is based on imaginary time evolution: 
\begin{align}
   \lim_{t\longrightarrow \infty} \frac{1}{\bra{\psi} e^{-2Ht} \ket{\psi} } e^{-Ht} \ket{\psi} = \ket{\lambda_{\rm ground}}
\end{align}
where $\ket{\psi}$ is the initial random state. The recipes that we need to perform imaginary time evolution are \cref{lemma: theorem56} and \cref{lemma: exponentialapproximation}. We use the polynomial $P$ in \cref{lemma: exponentialapproximation} to transform the block-encoded Hamiltonian $\frac{1}{2}\big(\Ibb - H\big)$ into $P(H) \approx \exp\Big( -\beta (\Ibb - \frac{1}{2}\big(\Ibb - H\big))\Big) = \exp\Big( -\beta \frac{1}{2}\big( \Ibb + H\big)  \Big) $. By choosing $\beta = t$, we obtain an $\epsilon$-approximated block encoding of $\exp\Big( -\frac{1}{2}\big( \Ibb + H\big) t  \Big)$. We note that the operator $ (\Ibb+H)$ has the same eigenvectors as $H$, but only the spectrum is shifted. Therefore, the eigenvector that corresponds to the lowest eigenvalue of $\frac{1}{2} (\Ibb+H) $ is exactly the ground state of $H$, and thus the imaginary evolution on this ``shift'' Hamiltonian eventually produces the ground state of $H$ as desired. Once we obtain the block encoding of $ \exp\big( - \frac{1}{2} (\Ibb+H)t \big)$, we can use it on the randomized state $\ket{\bf 0}\ket{\psi}$ where $\ket{\bf 0}$ is the number of ancilla qubits required for block encoding purpose, and by the property of \cref{def: blockencode} (and~\cref{eqn: action}), we obtain the following state:
\begin{align}
    \ket{\bf 0} \exp\big( - \frac{1}{2} (\Ibb+H)t \big) \ket{\psi } + \ket{\rm Garbage}
    \label{eqn: i2}
\end{align}
Measuring the ancilla qubits and post-select on $\ket{\bf 0}$, we obtain the state $\frac{1}{\bra{\psi} e^{-(\Ibb+H)t} \ket{\psi} }  \exp\big( - \frac{1}{2} (\Ibb+H)t \big) \ket{\psi } \equiv \ket{\widetilde{\lambda_{\rm ground }}}$. According to the analysis provided in the Appendix \ref{sec: proofite}, by choosing $t = \mathcal{O}\big( \frac{1}{\Delta} \log \frac{n}{\epsilon \gamma}  \big) $ (where we remind that $\Delta$ is the gap between ground state energy and first excited state energy, and $\gamma = |\braket{\psi, \lambda_{\rm ground}}|$) it is guaranteed that $ ||\ket{\widetilde{\lambda_{\rm ground }}}  - \ket{\lambda_{\rm ground}} || \leq \epsilon$.

We note that the probability of measuring $\ket{\bf 0}$ in the this step is $\bra{\psi} \exp(-i (\Ibb + H)t) \ket{\psi}$. One can see that this is exponentially small in $t$, and thus implying the inefficiency in $\Delta$ and $n, \gamma$ if we replace $t = \mathcal{O}\big( \frac{1}{\Delta} \log \frac{n}{\epsilon \gamma}  \big)$. To overcome this issue, we propose the following procedure. First, we use Lemma \ref{lemma: improveddme} to block-encode the density state in the Eqn.~\ref{eqn: i2}, which is:
\begin{align}
    \ket{\bf 0} \bra{\bf 0} \otimes \exp\big( -  (\Ibb+H)t \big) \ket{\psi }\bra{\psi} + \ket{\rm Garbage} \bra{\bf 0} \exp\big( - \frac{1}{2} (\Ibb+H)t \big) \ket{\psi } \bra{\psi} + \\ (\ket{\rm Garbage} \bra{\bf 0} \exp\big( - \frac{1}{2} (\Ibb+H)t \big) \ket{\psi } \bra{\psi})^\dagger + \ket{\rm Garbage} \bra{\rm Garbage}
\end{align}
Due to the orthogonality between $\ket{\rm Garbage}$ and $\ket{\bf 0}$ (see Def \ref{def: blockencode}), the state $\ket{\rm Garbage}$ does not contain $\ket{\bf 0}$ in its basis decomposition, therefore, the above density operator is exactly the block-encoding of:
\begin{align}
     \exp\big( - (\Ibb+H)t \big) \ket{\psi } \bra{\psi} \equiv ||\lambda_{\rm ground }||^2 \ket{\lambda_{\rm ground }}\bra{\lambda_{\rm ground }}
\end{align}
where we defined $||\lambda_{\rm ground }||^2 \equiv \bra{\psi} \exp(-i (\Ibb + H)t) \ket{\psi} $. We then use Lemma \ref{lemma: exponentialapproximation} again (albeit with different value of $\beta$, and we denote as $\beta'$ to distinguish with the $\beta$ we used earlier) to transform the above block-encoded operator into:
\begin{align}
    ||\lambda_{\rm ground }||^2 \ket{\lambda_{\rm ground }}\bra{\lambda_{\rm ground }} \longrightarrow \exp \left( - \beta'(1-|\lambda_{\rm ground }||^2) \right) \ket{\lambda_{\rm ground }}\bra{\lambda_{\rm ground }} 
\end{align}
In the Appendix \ref{sec: reviewpowermethod}, we shall show that by choosing $\beta'$ sufficiently small, e.g., $\beta \leq \frac{1}{2(1-|\lambda_{\rm ground }||^2)} \log \frac{1}{1-\epsilon}$, then $\exp \left( - \beta'(1-|\lambda_{\rm ground }||^2) \right)$ is $\epsilon$-close to 1. Thus,  the operator $ \exp \left( - \beta'(1-|\lambda_{\rm ground }||^2) \right) \ket{\lambda_{\rm ground }}\bra{\lambda_{\rm ground }} $ is an $\epsilon$-approximated to $ \ket{\lambda_{\rm ground }}\bra{\lambda_{\rm ground }} $. As the next step, we take such the block-encoding and apply it to the state $\ket{\bf 0}\ket{\psi}$ where $\ket{\psi}$ is the initial state we chose for the imaginary time evolution. According to Definition \ref{def: blockencode}, we obtain the state:
\begin{align}
    \ket{\bf 0} \exp \left( - \beta'(1-|\lambda_{\rm ground }||^2) \right) \ket{\lambda_{\rm ground }} \bra{\lambda_{\rm ground }} \ket{\psi} + \ket{\rm Garbage}
\end{align}
Measuring the ancilla and post-select on seeing $\ket{\bf 0}$, we then obtain the state which is an approximation to the true ground state $\ket{\lambda_{\rm ground}}$. The success probability of such measurement is 
$$\exp \left( - 2\beta'(1-|\lambda_{\rm ground }||^2) \right) |\bra{\lambda_{\rm ground }} \ket{\psi}|^2 \leq \gamma^2 $$
and could be increased to $\mathcal{O}(\gamma)$  based on amplitude amplification. 

Recall that the complexity for obtaining the block encoding of $H$ is $\mathcal{O}\big( \log(n) \log^2 \frac{1}{\epsilon}\big)$. We use \cref{lemma: theorem56} with polynomial $P$, where $P$ is a polynomial having degree $\mathcal{O}\Big(  \sqrt{t \log \frac{1}{\epsilon} } \Big)$ (from \cref{lemma: exponentialapproximation}) to obtain the block encoding of $P(H) = \exp \left( - t (\Ibb + H) \right)$. The total complexity to obtain the block-encoding of $P(H)$ is $\mathcal{O}\Big(   \sqrt{t} \log (n) \log^{5/2} \frac{1}{\epsilon}  \Big) $. We then use such block-encoding with Lemma \ref{lemma: improveddme} to obtain the block-encoding of
$$ ||\lambda_{\rm ground }||^2 \ket{\lambda_{\rm ground }}\bra{\lambda_{\rm ground }} $$
before transforming it into the $\epsilon$-approximated block-encoding of $\ket{\lambda_{\rm ground }}\bra{\lambda_{\rm ground }}$. Lemma \ref{lemma: improveddme} uses $\mathcal{O}(1)$ block-encoding of $P(H)$, and in Lemma \ref{lemma: exponentialapproximation} we use a small value of $\beta'$, therefore, the total complexity for obtaining the $\epsilon$-approximated block-encoding of $ \ket{\lambda_{\rm ground }}\bra{\lambda_{\rm ground }}$ is:
\begin{align}
    \mathcal{O}\Big(   \sqrt{t} \log (n) \log^{7/2} \left( \frac{1}{\epsilon} \right) ||H||_F  \Big)
\end{align}
The last step is to apply such block-encoding to the state $\ket{\bf 0} \ket{\psi}$ and post-measuring on $\ket{\bf 0}$. The success probability is $\mathcal{O}(\gamma)$ and thus the complexity is $\mathcal{O}(1/\gamma) $. By replacing $ t = \mathcal{O}\big( \frac{1}{\Delta} \log \frac{n}{\gamma \epsilon}  \big) $, we arrive at the final complexity:
$$ \mathcal{O}\Big( \frac{||H||_F }{\gamma}  \sqrt{\frac{1}{\Delta} \log \big(\frac{n}{ \gamma \epsilon} }\big) \log (n) \log^{7/2} \frac{1}{\epsilon}  \Big)  $$

To compare with existing results for ground state preparation, we provide the table \ref{tab: groundstatepreparation} summarizing the state-of-the-art complexities.
\begin{table*}[t]
    \centering
    {
    \begin{tabular}{|c|l|}
    \hline
    \textbf{Method} &  \textbf{Ground State Preparation} \\
    \hline
    Power method
    & $\mathcal{O}\Big( \log(n) \frac{1}{\Delta\gamma }   \log \big(\frac{n}{\epsilon}\big) \log \frac{1}{\epsilon}    \Big)  $  \\
    \hline
    Gradient descent method &  $\mathcal{O}\Big(  \log (\frac{1}{\epsilon}) \big(\frac{4}{\epsilon}\big) \log n \Big)$ \\
    \hline
    Imaginary Time Evolution &    $ \mathcal{O}\Big(  \frac{||H||_F }{\gamma} \sqrt{\frac{1}{\Delta} \log \big(\frac{n}{ \gamma^2 \epsilon} }\big) \log (n) \log^{7/2} \frac{1}{\epsilon}  \Big)$ \\
    \hline
    Ref.~\cite{dong2022ground} 
    & $\mathcal{O}\Big( \frac{1}{\gamma^2 \epsilon} T_U \Big)$   \\
    \hline
    Ref.~\cite{lin2020near} 
    & $\mathcal{O}\Big( \frac{1}{\gamma \epsilon}\log\big( \frac{1}{\gamma} \big) \log \big( \frac{1}{\epsilon}\big) T_U  \Big)$ \\
    \hline
    Ref.~\cite{motta2020determining} &  $  \frac{1}{\Delta} \log \big( \frac{n }{\epsilon \gamma }  \big) \exp\Big( \mathcal{O}\big( \rm poly \big(\ln 2\sqrt{2}  \frac{1}{\Delta}  \big(\log \frac{n }{\epsilon \gamma }  \big) \frac{1}{\epsilon}  \big)  \Big) $\\
    \hline
    \end{tabular}
    }
    \caption{Table summarizing our result and relevant works \cite{dong2022ground, lin2020near, motta2020determining}. $T_U$ denotes the complexity of block-encoding Hamiltonian of interest $H$ used in \cite{lin2020near, dong2022ground}, respectively, and $\gamma = | \braket{\psi, \lambda_{\rm ground} }|$.    }
    \label{tab: groundstatepreparation}
\end{table*}

\section{Data Fitting }
\label{sec: datafitting}
As mentioned in the main text, data fitting problem eventually reduces to the problem of finding the parameters that minimize a chosen cost function. More formally, a set of data points $\{ x_i, y_i \}_{i=1}^M$ is given, where $x_i \in \Rbb^N$ is some $N$-dimensional vector and $y_i \in \Rbb$. The fit function is of the form:
\begin{align}
    f(x,\lambda) = \sum_{j=1}^N f_j(x) \lambda_j
\end{align}
where $f_j(x): \Rbb^N \longrightarrow \Rbb$ (a continuous function, and can be nonlinear in $x$) and $\lambda = (\lambda_1,\lambda_2,...,\lambda_N)^T$. The goal is to find the adjustment parameters $\lambda_1,\lambda_2,...,\lambda_N$ that minimize the so-called loss-squared cost function:
\begin{align}
    C = \sum_{i=1}^M |f(x_i,\lambda) - y_i|^2 
\end{align}
As pointed out in \cite{wiebe2012quantum}, the desired parameters $\lambda$ can be found as:
\begin{align}
    \lambda = ( F^\dagger F )^{-1} F^\dagger y
\end{align}
where the matrix $F$ is defined as $F_{ij} = f_j(x_i)$ and $y = (y_1,y_2,...,y_M)^T$. Without loss of generalization, we assume that the operator norm of $F$, $||F||_{\rm o} \leq 1$. 

Our quantum data fitting algorithm is also a corollary of the algorithm to solve linear equations, see, e.g., \cref{alg:solvinglinearequation}, with a few modifications. Because $F$ is not necessarily a Hermitian matrix, we adapt the idea used in \cite{wiebe2012quantum}, that we embed:
$$ F \longrightarrow F' = \begin{pmatrix}
    0 & F^\dagger \\
    F & 0 
\end{pmatrix}$$
Provided that the classical knowledge of $x_i$, the value of $f_j(x_i)$ can then be computed. Therefore, the entries of $F'$ can be classically computed.  The same procedure as above \cref{lemma: improveddme} can be applied to obtain the block encoding of $ F'^\dagger F'$, which is:
$$  \begin{pmatrix}
    0 & F^\dagger \\
    F & 0 
\end{pmatrix} \begin{pmatrix}
    0 & F^\dagger \\
    F & 0 
\end{pmatrix} = \begin{pmatrix}
    F^\dagger F & 0 \\
    0 & F F^\dagger
\end{pmatrix}$$
We note that according to \cref{def: blockencode}, the above operator is indeed the block encoding of $F^\dagger F$. The next step is to use \cref{lemma: negative} with $c= -1$ to invert the block encoded operator $ F^\dagger F$, to obtain the block encoding of $\frac{1}{\kappa_F^2} (F^\dagger F)^{-1}$ -- where we remind that $\kappa_F$ is the conditional number of $F$ and the conditional number of $F^\dagger F $ is $\kappa_F^2$. In addition, we can use \cref{lemma: positive} (with $c=1/2$) to transform the block-encoded operator 
$$ \begin{pmatrix}
    F^\dagger F & 0 \\
    0 & F F^\dagger
\end{pmatrix}$$
into its square root, i.e., 
\begin{align}
    \begin{pmatrix}
    F^\dagger F & 0 \\
    0 & F F^\dagger
\end{pmatrix} \longrightarrow \frac{1}{2}\sqrt{ \begin{pmatrix}
    F^\dagger F & 0 \\
    0 & F F^\dagger
\end{pmatrix} } = \frac{1}{2} \begin{pmatrix}
    0 & F^\dagger \\
    F & 0 
\end{pmatrix} 
\end{align}
We point out that the above block-encoded operator has an equivalent expression:
\begin{align}
    \begin{pmatrix}
    0 & F^\dagger \\
    F & 0 
\end{pmatrix}  = \ket{0}\bra{1}\otimes F^\dagger + \ket{1}\bra{0}\otimes F
\end{align}
The block encoding of $X \otimes \Ibb_M$ is straightforward to obtained, therefore, we can use \cref{lemma: product} to construct the block encoding of $ \left(\ket{0}\bra{1}\otimes F^\dagger + \ket{1}\bra{0}\otimes F\right) \left( X \otimes \Ibb_M\right) = \ket{0}\bra{0}\otimes F^\dagger + \ket{1}\bra{1}\otimes F$, which has the following matrix representation:
$$ \begin{pmatrix}
    F^\dagger  & 0 \\
    0 & F 
\end{pmatrix}$$
It is straightforward to see that the above block-encoded operator is also a block encoding of $F^\dagger$. So, we can use \cref{lemma: product} again to obtain the block encoding of $ \frac{1}{\kappa_F} (F^\dagger F)^{-1} F^\dagger$. Given that the values $\{ y_i\}_{i=1}^M$ are known, we can use the method in \cite{zhang2022quantum} to prepare the state $ \frac{1}{||y||} y$. The final step is to use the block encoding of $\frac{1}{\kappa_F} (F^\dagger F)^{-1} F^\dagger $ and apply it to $ \frac{1}{||y||} y$, and according to~\cref{eqn: action}, it results in:
\begin{align}
    \ket{\bf 0} \frac{1}{\kappa_F} (F^\dagger F)^{-1} F^\dagger \frac{1}{||y||} y + \ket{\rm Garbage} = \ket{\bf 0} \frac{\lambda}{ \kappa_F ||y||}  + \ket{\rm Garbage}
\end{align}
From the above state, by measuring the ancilla and post-select on $\ket{\bf 0}$, we obtain the state $\varpropto (F^\dagger F)^{-1} F^\dagger  y$, which is exactly $\lambda$ up to a normalization factor. We remark that, even though obtaining $\ket{\lambda} \varpropto \lambda$ is possible by measuring the ancilla, it is not always necessary to do so. In practice, we typically desire to \textbf{predicting unseen input} once we obtain the fit parameters $\lambda$. It can be done as follows. Suppose that we are given some unknown data point $\Tilde{x}$ and desire to find out $f(\Tilde{x})$. It means that we need to evaluate $f(\Tilde{x}) = \sum_{i=1}^N \lambda_i \Tilde{x}^i$. First, we use the amplitude encoding method~\cite{zhang2022quantum, mcardle2022quantum} to prepare the state 
$$\ket{\Tilde{x}} = \frac{1}{\sqrt{\sum_{i=1}^N |\Tilde{x}^i|^2}}\sum_{i=1}^N \Tilde{x}^i \ket{i-1}.$$
Next, we append ancilla qubits $\ket{\bf 0}$ to the above state, then we use the well-known SWAP or Hadamard test to estimate the overlaps 
$$  \bra{\bf 0}\bra{\Tilde{x}}  \left( \ket{\bf 0} \frac{\lambda}{ \kappa_F ||y||}  + \ket{\rm Garbage}\right). $$
Because the state $\ket{\rm Garbage}$ is orthogonal to $\ket{\bf 0}\ket{\phi}$ for whatever $\ket{\phi}$  (see Definition \ref{def: blockencode}), the above overlaps is 
$$\frac{1}{ \sqrt{\sum_{i=1}^N |\Tilde{x}^i|^2}\kappa_F ||y||} \sum_{i=1}^N \lambda_i \Tilde{x}^i$$ 
By doing this way, we can avoid the measuring and post-selecting step, thus saving a considerable amount of resources. 

To analyze the complexity, we note that from the classical knowledge of $F'$ (which is drawn from the classical knowledge of $F$), the complexity for producing the block encoding of $F'^\dagger F$ is $\mathcal{O}\left( \log (MN)    \right)$. The block encoding of $F^\dagger F$ is obtained by using two block encodings of $F'^\dagger F $, so the complexity is $\mathcal{O}\left( \log MN \right)$. The next step is inverting $F^\dagger F$, using \cref{lemma: negative} with $c=-1$, so the complexity for obtaining the block encoding of $\frac{1}{\kappa_F^2} (F^\dagger F)^{-1}$ is 
$$ \mathcal{O}\left(  \log (MN) \kappa_F^2 \log^2 \frac{1}{\epsilon} \right)$$
Next, \cref{lemma: positive} is used to obtain the block encoding of:
$$  \frac{1}{2}\sqrt{ \begin{pmatrix}
    F^\dagger F & 0 \\
    0 & F F^\dagger
\end{pmatrix} } = \frac{1}{2}\begin{pmatrix}
    0 & F^\dagger \\
    F & 0 
\end{pmatrix}  $$
So the complexity to obtain the $\epsilon$-approximated block encoding of the above operator  is $\mathcal{O}\left(\log (MN) \log^2 \frac{1}{\epsilon} \right)$. From the block encoding of the above operator to the block encoding of 
$$ \begin{pmatrix}
    F^\dagger  & 0 \\
    0 & F 
\end{pmatrix}$$
utilizes only a single X gate, so the total complexity is still $\mathcal{O}\left(\log (MN) \log^2 \frac{1}{\epsilon} \right) $. The block encoding of $ \frac{1}{\kappa_F} (F^\dagger F)^{-1} F^\dagger $ can be obtained by using one block encoding of $\varpropto (F^\dagger F)^{-1}$ and one block encoding of $F^\dagger$, so the total complexity is 
$$ \mathcal{O}\left( ||F||_F\log (MN) \kappa_F^2 \log^2 \frac{1}{\epsilon} \right)$$
The state $\frac{1}{||y||}y$ can be prepared using \cref{lemma: stateprepration}, with the complexity $\mathcal{O}\left( \log M\right)$. The final step is taking block encoding of $\frac{1}{\kappa_F} (F^\dagger F)^{-1} F^\dagger$ and apply it to $\ket{\bf 0} \frac{y}{||y||}$. The next step is to predict an unseen input $\tilde{x}$, which first involves the preparation of state $\ket{\tilde{x}}$ and then using Hadamard/SWAP test to evaluate the overlaps. The preparation of $\ket{\tilde{x}}$ has complexity $\mathcal{O}\left( \log N\right)$, and the overlaps estimation with additive precision $\epsilon$ has complexity $\mathcal{O}\left( 1/\epsilon\right)$. So the total complexity for obtaining the fit parameters vector $\lambda$ (encoded in a larger state) and eventually predict an unseen input is 
$$\mathcal{O}\left( ||F||_F\log (MN) \kappa_F^2  \frac{1}{\epsilon}\log^2 \frac{1}{\epsilon}  \right) $$
Thus, we arrive at the complexity stated in the main result of \cref{sec: overviewquantumdatafitting}.  

To summarize, we provide the following pseudo-code algorithm.

\begin{algorithm}[H]
\caption{Quantum Data Fitting Algorithm}
\label{alg:datafitting}
\KwIn{Classical knowledge $\{F'^1, F'^2, \dots, F'^n\}$, vector $y$}
\KwOut{$\lambda = (F^\dagger F)^{-1} F^\dagger y$}

\If{$F'$ is positive semidefinite}{
    Prepare state $\sum_{i=1}^n \ket{i} F'^i$ \tcp*{\cref{lemma: stateprepration}}
    Apply improved DME to obtain $F'^\dagger F'$ \tcp*{\cref{lemma: improveddme}}
    Apply negative power lemma to obtain $(F'^\dagger F')^{-1}$ \tcp*{\cref{lemma: negative}}
    Multiply with $F^\dagger y$ using product lemma \tcp*{\cref{lemma: product}}
}
\Else{
    Transform $F' \leftarrow \frac{\mathbb{I}_n + F'}{2}$ \;
    Prepare state $\sum_{i=1}^n \ket{i} \frac{1}{2}(\mathbb{I}_n + F')^i$ \tcp*{\cref{lemma: stateprepration}}
    Apply improved DME to obtain $(\mathbb{I}_n + F')^T (\mathbb{I}_n + F')$ \tcp*{\cref{lemma: improveddme}}
    Apply positive power lemma to obtain $\mathbb{I}_n + F'$ \tcp*{\cref{lemma: positive}}
    Apply sum encoding to recover $F'$ \tcp*{\cref{lemma: sumencoding}}
    Apply negative power lemma to obtain $(F'^\dagger F')^{-1}$ \tcp*{\cref{lemma: negative}}
    Multiply with $F^\dagger y$ using product lemma \tcp*{\cref{lemma: product}}
}
\Return{$\lambda = (F^\dagger F)^{-1} F^\dagger y$}
\end{algorithm}

\section{More Details on prior QPCA algorithms}
\label{sec: moredetailspca}
Here we provide more technical details of the discussion in \cref{sec: overviewPCA}, where we mention previous progress regarding QPCA, specifically \cite{lloyd2014quantum, nghiem2025new}. \\

\noindent
\textbf{Ref.~\cite{lloyd2014quantum}.} This work's initial motivation was actually simulating density matrix, i.e., obtaining $\exp(- i \rho t)$ from multiple copies of density state $\rho \in \mathbb{C}^{n \times n}$ where $n$ is the dimension. In order to obtain the unitary transformation $\exp(- i \rho t)$, the authors in \cite{lloyd2013quantum} used the following property:
\begin{align}
    & \Tr_1 \exp(-iS \Delta t) \big( \rho \otimes \sigma \big) \exp(-iS \Delta t) \\ &= \sigma - i \Delta t [\rho,\sigma] + \mathcal{O}(\Delta t^2) \approx \exp(-i \rho \Delta t) \sigma \exp(-i \rho \Delta t)
\end{align}
where $\Tr_1$ is the partial trace over the first system, $S$ is the swap operator between two system of $\log(n)$ qubits, and $\sigma$ is some ancilla system. Defining $\exp(-i \rho \Delta t) \sigma \exp(-i \rho \Delta t) = \rho_1 $. Repeat the above step:
\begin{align}
     & \Tr_1 \exp(-iS \Delta t)  \big( \rho \otimes \rho_1   \big) \exp(-iS \Delta t) \\ &\approx \exp(-i\rho \Delta t) \big( \exp(-i \rho \Delta t) \sigma \exp(-i \rho \Delta t) \big)  \exp(-i \rho \Delta t) \\
     &= \exp(- i\rho 2\Delta t) \sigma \exp(-i \rho 2\Delta t) 
\end{align}
To obtain $\exp(-i \rho t)$, we repeat the above procedure $N$ times, then we obtain:
\begin{align}
    \exp(-i \rho N \Delta t) \sigma \exp(-i \rho N \Delta t)
\end{align}
The authors in \cite{lloyd2013quantum} shows that to simulate $\exp(-i \rho t)$  to accuracy $\epsilon$, then it requires:
\begin{align}
    N = \mathcal{O} \big(  \frac{t^2}{\epsilon}  \big)
\end{align}
total number of copies and repetition, where $t = N\Delta t$. To find the top eigenvalues/eigenvectors, the authors in \cite{lloyd2014quantum} used quantum phase estimation with $\rho$ as input state. Denote the spectrum of $\rho$ as $\{ \alpha_i , \ket{\phi_i} \}$. The outcome of phase estimation algorithm is a density state:
\begin{align}
    \sum_{i} \alpha_i \ket{\Tilde{\alpha}_i} \bra{ \Tilde{\alpha}_i} \otimes \ket{\phi_i}\bra{\phi_i}
\end{align}
where $ \Tilde{\alpha}_i $ is a binary string approximation of $\alpha_i$. By sampling from the above state, we can obtain the highest eigenvalues / eigenvectors because the probability to obtain the highest eigenvalues is $|\alpha_i|^2$, which means that the higher the value, the higher probability. According to \cite{lloyd2014quantum}, in order to guarantee that the error of eigenvalues estimation is $\epsilon$, we need to choose $t = \mathcal{O}(1/\epsilon^2) $. So the total complexity of this approach is $ \mathcal{O}(1/\epsilon^3)$. 

To apply this approach in the context of principal component analysis, the Ref.~\cite{lloyd2014quantum} assumed that, via some oracle (or quantum random access memory), the ability to prepare a density state $\rho \varpropto \mathcal{C}$ (where $\mathcal{C}$ is the covariance matrix) in logarithmic time $\mathcal{O}(\log mn)$. Then the above procedure yields the top $r$ eigenvalues/ eigenvectors with complexity $\mathcal{O}\Big(  r \frac{1}{\epsilon^3} \log mn  \Big)  $, as one needs to repeat the sampling roughly $r$ times to obtain $r$  different eigenvalues/ eigenvectors. \\

\noindent
\textbf{Ref.~\cite{nghiem2025new}.} The approach of this work is a combination of the density matrix exponentiation technique above and the power method, which was also used in our main text. Instead of using $\exp(-i\rho t)$ with the phase estimation algorithm, the authors of \cite{nghiem2025new} leveraged the following result from \cite{gilyen2019quantum}:
\begin{lemma}[Corollay 71 in \cite{gilyen2019quantum}]
\label{lemma: logarithmicofunitary}
    Suppose that $U = \exp(-iH)$, where $H$ is a Hamiltonian of norm at most $1/2$. Let $\epsilon \in (0,1/2]$, then we can implement an $\epsilon$-approximated block encoding of $\pi H /2$ (see further \cref{def: blockencode}) with $\mathcal{O}(\log(\frac{1}{\epsilon} ))$ uses of controlled-U and its inverse, using $\mathcal{O}(\log(\frac{1}{\epsilon}))$ two-qubit gates and using a single ancilla qubit. 
\end{lemma}
The above lemma allows us to construct the block encoding of $\frac{\pi}{4}\rho$ from $\exp(-i \rho t)$ (by setting $t= 1/2)$. To prepare a covariance matrix without resorting on oracle/QRAM, we recall from the main text that the dataset contains $m$ samples $\xbf^1,\xbf^2, ..., \xbf^m$ where each $\xbf^i \in \Rbb^n$. In the context of \cite{gordon2022covariance} and \cite{nghiem2025new}, they assumed that each data is normalized, i.e., $ ||\xbf^i ||  =1$. Provided $\xbf^i$ is known, the amplitude encoding method \cite{grover2000synthesis,grover2002creating,plesch2011quantum, schuld2018supervised, nakaji2022approximate,marin2023quantum,zoufal2019quantum,prakash2014quantum, zhang2022quantum} can be used to prepare it with an efficient circuit $U_i$ of depth $\mathcal{O}(\log n)$. Suppose that from $\{ \xbf^1,\xbf^2, ..., \xbf^m \} $, we randomly select $\xbf^i$ with probability $1/m$, then we obtain an ensemble $\frac{1}{m} \sum_{i=1}^n \xbf^i (\xbf^i)^\dagger$. Using the above procedure, first simulate $\exp\big(-i \frac{1}{2m} \sum_{i=1}^n \xbf^i (\xbf^i)^\dagger \big)$ (with complexity $\mathcal{O}\big(  \frac{1}{\epsilon}\log n \big)$, then apply \cref{lemma: logarithmicofunitary} to construct the block encoding of $\frac{\pi}{4} \frac{1}{m} \sum_{i=1}^n \xbf^i (\xbf^i)^\dagger$. The resultant complexity is then $\mathcal{O}\big( \frac{1}{\epsilon}\log(\frac{1}{\epsilon}) \log n \big) $.

To construct the block encoding of $\frac{1}{m}\mu \mu^\dagger$, they use \cref{lemma: sumencoding} to construct the block encoding of $ \frac{1}{m} \sum_{i=1}^m U_i$, which contains $\frac{1}{m} \sum_{i}\xbf_i $ as the first column. The complexity of this step is $\mathcal{O}(m \log n)$ because each $U_i$ is used one time. Then they use \cref{lemma: improveddme} to construct the block encoding of $\Big( \frac{1}{m} \sum_{i}\xbf_i\Big) \Big( \frac{1}{m} \sum_{i}\xbf_i\Big)^\dagger \equiv \mu \mu^\dagger $, which can be combined with \cref{lemma: scale} to transform it to $\frac{\pi}{4} \mu \mu^\dagger $. Recall that covariance matrix $\mathcal{C} $ can be expressed as:
\begin{align}
    \mathcal{C} = \frac{1}{m} \sum_{i=1}^n \xbf^i (\xbf^i)^\dagger - \mu\mu^\dagger
\end{align}
Thus one can use the block encoding of $\frac{\pi}{4} \frac{1}{m} \sum_{i=1}^n \xbf^i (\xbf^i)^\dagger,\frac{\pi}{4} \mu \mu^\dagger  $ and \cref{lemma: sumencoding} to construct the block encoding of $ \frac{1}{2} \Big(\frac{\pi}{4} \frac{1}{m} \sum_{i=1}^n \xbf^i (\xbf^i)^\dagger- \frac{\pi}{4} \mu \mu^\dagger   \Big) $, which is $\frac{\pi}{8} \mathcal{C}$. The complexity of this method is $\mathcal{O}\big( \frac{1}{\epsilon}\log(\frac{1}{\epsilon}) \log n + m \log n  \big)$. 

Another method for preparing the covariance matrix, as provided in \cite{nghiem2025new}, is to use $U_i$ with \cref{lemma: improveddme} to construct the block encoding of $\xbf^i (\xbf^i)^\dagger$ for all $i=1,2,...,m$. Then one uses \cref{lemma: sumencoding} to construct the block encoding of $ \frac{1}{m} \sum_{i=1}^m \xbf^i (\xbf^i)^\dagger $. This construction has complexity $\mathcal{O}(m \log n)$. Given that the block encoding of $\mu \mu^\dagger $ is provided above, one can use \cref{lemma: sumencoding} to construct the block encoding of $ \frac{1}{2} \Big(\frac{1}{m} \sum_{i=1}^m \xbf^i (\xbf^i)^\dagger- \mu \mu^\dagger \Big) \equiv \frac{1}{2}  \mathcal{C} $, with total complexity $\mathcal{O}(m \log n)$. Then one can find the top eigenvector/eigenvalue of $\frac{\pi}{4}\rho$ through  \cref{lemma: largestsmallest}. From such an eigenstate, one repeat the above procedure: using copies of $\ket{\lambda_i}\bra{\lambda_i}$ and simulate $\exp(- i\ket{\lambda_i}\bra{\lambda_i} /2)  $, then use \cref{lemma: logarithmicofunitary} to recover $ \frac{\pi}{4}\ket{\lambda_i}\bra{\lambda_i}$. Then one considers finding the maximum eigenvalue/eigenvector of $\mathcal{C } - \lambda_1\ket{\lambda_1}\bra{\lambda_1}$, and continue this process for $r$ eigenvalues/eigenvectors. According to the analysis provided in \cite{nghiem2025new}, the circuit complexity for producing top $r$ eigenvalues/eigenvectors is $\mathcal{O }\Big( m \log(n) \big( \frac{1}{\Delta^2} \log^3 (\frac{n}{\epsilon} ) \frac{1}{\epsilon^2} \big)^r  \Big)$ where $\Delta$ is the gap between two largest eigenvalues.

\section{More Details on Prior Quantum Linear Solving Algorithms}
\label{sec: moredetaillinearsolver}
With similar purpose to the previous section, in the following, we provide more details about existing quantum linear solving algorithms. \\

\noindent
\textbf{Ref.~\cite{harrow2009quantum}.} Under the same notations and conditions as in~\cref{sec: solvinglinearequation}, with a further assumption that there is an oracle/black-box access to entries of $A$ (in an analogous manner to previous simulation contexts \cite{berry2007efficient, aharonov2003adiabatic, berry2012black}), this work first leveraged these simulation algorithms to perform $\exp(-i A t)$. Then they perform the quantum phase estimation with $\exp(-i At)$ and $\ket{\textbf{b}}$ as input state, to obtain:
\begin{align}
    \sum_{i=1}^n \beta_i \ket{\phi_i} \ket{\lambda_i}
\end{align}
where $\{ \lambda_i, \ket{\phi_i} \}$ is eigenvalues/eigenvectors of $A$ and $\{ \beta_i\}$ is the expansion coefficients of $\ket{\textbf{b}}$ in this basis, i.e., $ \ket{\textbf{b}} = \sum_{i=1}^n \beta_i \ket{\phi_i}$. Then they append an ancilla initialized in $\ket{0}$, and rotate the ancilla conditioned on the phase register:
\begin{align}
     \sum_{i=1}^n \beta_i \ket{\phi_i} \ket{\lambda_i} \ket{0} \longrightarrow  \sum_{i=1}^n \beta_i \ket{\phi_i} \ket{\lambda_i}\Big( \frac{1}{\kappa \lambda_i} \ket{0} + \sqrt{1- \frac{1}{\kappa^2 \lambda_i^2}}\ket{1} \Big)
\end{align}
By uncomputing the phase register, or reverse the phase estimation algorithm, and discard that register, we obtain: 
\begin{align}
    \sum_{i=1}^n \beta_i \ket{\phi_i} \Big( \frac{1}{\kappa \lambda_i} \ket{0} + \sqrt{1- \frac{1}{\kappa^2 \lambda_i^2}}\ket{1} \Big)
\end{align}
Measuring the ancilla and post-select on $\ket{0}$, we obtain a state $\varpropto  \sum_{i=1}^n \frac{\beta_i}{\kappa \lambda_i} \ket{\phi_i} = \frac{1}{\kappa} A^{-1} \ket{\textbf{b}} $. The complexity of this algorithm, as analyzed in \cite{harrow2009quantum}, is $\mathcal{\Tilde{O}}\big( \kappa^2 s^2 \log (n) \frac{1}{\epsilon}\big)$ where $\Tilde{O}$ hides the polylogarithmic factor. \\

\noindent
\textbf{Ref.~\cite{childs2017quantum}.} The above HHL algorithm makes use of a quantum phase estimation algorithm, which leads to an unavoidable scaling in $1/\epsilon$. The work of \cite{childs2017quantum} improves upon this aspect by making use of the following approximations: 
\begin{align}
  \text{\rm Fourier approximation:} \   A^{-1} \approx \sum_{j=1}^K \alpha_j \exp(-i A \Delta_j) \\
   \text{\rm Chebyshev approximation:} \ A^{-1} \approx \sum_{j=1}^K \alpha_j T_{j} (A)
\end{align}
By using more precise simulation algorithms \cite{berry2015hamiltonian, berry2015simulating}, the terms $\exp(-i A\Delta_j)$ can be approximated more efficiently. Implementation of Chebyshev polynomials is also known to be efficient via quantum walk technique \cite{childs2010relationship, berry2012black}. The summation $\sum_{j=1}^K \alpha_j \exp(-i A \Delta_j) $, $\sum_{j=1}^K \alpha_j T_{j} (A) $ can be constructed using the technique called linear combination of unitaries \cite{berry2015simulating}. The value of $K$ turns out to be $\mathcal{O}\big( \kappa^2 \log^2 \frac{\kappa}{\epsilon}\big)$. Overall, as provided in Theorem 3 and 4 of \cite{childs2017quantum}, the complexity for constructing $A^{-1}$ and eventually, obtaining $\varpropto A^{-1} \ket{\textbf{b}}$ is $\mathcal{O}\Big( s \kappa^2 \log^{2.5} \big(  \frac{\kappa}{\epsilon}\big)  \big(  \log n + \log^{2.5}\frac{\kappa}{\epsilon} \big)   \Big) $ and $\mathcal{O}\Big( s \kappa^2 \log^{2} \big(  \frac{\kappa}{\epsilon}\big)  \big(  \log n + \log^{2.5}\frac{\kappa}{\epsilon} \big)   \Big) $ for Fourier approximation approach and Chebyshev approximation approach, respectively. \\

\noindent
\textbf{Ref.~\cite{nghiem2025new2}.} This recently introduced approach for solving linear equations is based on reducing the original problem to an optimization problem, which can be solved by gradient descent. More specifically, given a linear system $A\xbf = \textbf{b}$, one can find $\xbf$ by minimizing the following function:
\begin{align}
    f(\xbf) = \frac{1}{2}||\xbf||^2 + \frac{1}{2} || A\xbf - \textbf{b}||^2
\end{align}
This strategy was also used in \cite{huang2019near} to solve linear system, however, they developed a variational algorithm and thus their algorithm is heuristic. The above formulation allows us to use the gradient descent algorithm to find the minima. As the above function is strongly convex, a global minima is also local minima, and thus convergence to such a minima is guaranteed. The gradient descent algorithm works by first initializing a random  vector $\xbf_0$, then iterate the following procedure $T$ times:
\begin{align}
    \xbf \leftarrow \xbf - \eta \bigtriangledown f(\xbf)
\end{align}
where $\eta$ is the learning hyperparmeter. In \cite{nghiem2025new2}, the author performed an embed $\xbf \longrightarrow \xbf \xbf^\dagger$, and in this new framework, the gradient descent's update rule is redefined as:
\begin{align}
    \xbf \xbf^\dagger \leftarrow  \big( \xbf - \eta \bigtriangledown f(\xbf)\big) \big( \xbf - \eta \bigtriangledown f(\xbf)\big)^\dagger 
\end{align}
which turns out to be $ \xbf \xbf^\dagger  - \eta \xbf \bigtriangledown^\dagger f(\xbf) - \eta \bigtriangledown f(\xbf)  \xbf^\dagger + \eta^2 \bigtriangledown f(\xbf) \bigtriangledown^\dagger f(\xbf)$. The gradient of $f(\xbf)$ is:
\begin{align}
    \bigtriangledown f(\xbf) = \xbf  + A^\dagger A \xbf - A^\dagger \textbf{b}
\end{align}
and therefore $\xbf \bigtriangledown^\dagger f(\xbf) = \xbf \xbf^\dagger (\Ibb_n + A^\dagger A ) - \xbf \textbf{b}^\dagger A$. The oracle access to entries of $A$ can be used to construct the block encoding of $\varpropto A$, based on the result of \cite{gilyen2019quantum}. The unitary that generates $\textbf{b}$ can be used to construct the block encoding of $\textbf{b}\textbf{b}^\dagger$. Then by the virtue of \cref{lemma: sumencoding} and \cref{lemma: product}, the block encoding of $\xbf \xbf^\dagger, \varpropto  \Ibb_n + A^\dagger A , \varpropto \xbf \textbf{b}^\dagger  $, and thus eventually can be all combined to yield the block encoding of  $ \xbf \bigtriangledown^\dagger f(\xbf), \bigtriangledown f(\xbf) \xbf^\dagger , \bigtriangledown f(\xbf) \bigtriangledown^\dagger f(\xbf) $. Another application of \cref{lemma: sumencoding} returns the block encoding of $ \varpropto \xbf \xbf^\dagger  - \eta \xbf \bigtriangledown^\dagger f(\xbf) - \eta \bigtriangledown f(\xbf)  \xbf^\dagger + \eta^2 \bigtriangledown f(\xbf) \bigtriangledown^\dagger f(\xbf)$, which completes an update step. Then the whole process is repeated again, to update another time, and continue until $T$ total iterations, we then obtain the block encoding of $\xbf_T \xbf_T^\dagger$. Using this unitary and apply it to a random state $\ket{\phi}$, according to \cref{def: blockencode}, we obtain the state $\ket{\bf 0} \xbf_T \xbf_T^\dagger \ket{\phi}     + \ket{\rm Garbage} $. Measuring the ancilla and post-select on $\ket{\bf 0}$, we obtain the state $\ket{\xbf_T}$, which is a quantum state corresponding to the point of minima of $f(\xbf)$. According to the analysis in \cite{nghiem2025new}, by choosing $T = \log \frac{1}{\epsilon}$, it is guaranteed that $\ket{\xbf_T}$ is $\epsilon$-close to the true minima of $f(\xbf)$, which is also the solution to the linear system. The complexity of this algorithm is $\mathcal{O}\Big( s^2 \frac{1}{\epsilon} \log n  \Big)$.

\section{Review of Method in Ref.~\cite{nghiem2024improved}}
\label{sec: reviewpowermethod}
We review main steps of the improved power method introduced in \cite{nghiem2024improved}, which underlies the \cref{lemma: largestsmallest}. Let $U_A$ denote the unitary block encoding of $A$. Then using \cref{lemma: product} $k$ times, we can construct the block encoding of $A^k$. Let $\ket{x_0}$ denote some initial state, generated by some known circuit $U_0$ (assuming to have $\mathcal{O}(1)$ depth). Defined $x_k = A^k \ket{x_0}$ and the normalized state $\ket{x_k} = \frac{x_k}{||x_k||}$. According to \cref{def: blockencode}, if we use the block encoding of $A^k$ to apply it to $\ket{x_0}$, we obtain the state:
\begin{align}
   \ket{\phi_1} =  \ket{\bf 0} A^k \ket{x_0} + \ket{\rm Garbage}
\end{align}
\cref{lemma: improveddme} allows us to construct the block encoding of $\ket{\phi_1}\bra{\phi_1}$, which is:
\begin{align}
    \ket{\phi_1}\bra{\phi_1} = \ket{\bf 0}\bra{\bf 0}\otimes x_k x_k^\dagger + (...)
    \label{d2}
\end{align}
where $(...)$ refers to the irrelevant terms. The above operator is exactly the block encoding of $x_k x_k^\dagger = ||x_k||^2 \ket{x_k} \bra{x_k} $, according to the \cref{def: blockencode}. We quote the following two results from \cite{gilyen2019quantum}: 

\begin{lemma}[Corollary 64 of \cite{gilyen2019quantum}  ]
\label{lemma: exponential}
   Let $\beta \in \mathbb{R}_+$ and $\epsilon \in (0,1/2]$. There exists an efficiently constructible polynomial $P \in \mathbb{R}[x]$ such that 
   $$ \Big|\!\Big| e^{ -\beta ( 1-x ) } - P(x)  \Big|\!\Big|_{x\in[-1,1]} \leq \epsilon. $$
   Moreover, the degree of $P$ is $\mathcal{O}\Big( \sqrt{\max[\beta, \log(\frac{1}{\epsilon})] \log(\frac{1}{\epsilon}}) \Big).$
\end{lemma}

\begin{lemma}\label{lemma: qsvt}[\cite{gilyen2019quantum} Theorem 56]
Suppose that $U$ is an
$(\alpha, a, \epsilon)$-encoding of a Hermitian matrix $A$. (See Definition 43 of~\cite{gilyen2019quantum} for the definition.)
If $P \in \mathbb{R}[x]$ is a degree-$d$ polynomial satisfying that
\begin{itemize}
\item for all $x \in [-1,1]$: $|P(x)| \leq \frac{1}{2}$,
\end{itemize}
then, there is a quantum circuit $\tilde{U}$, which is an $(1,a+2,4d \sqrt{\frac{\epsilon}{\alpha}})$-encoding of $P(A/\alpha)$ and
consists of $d$ applications of $U$ and $U^\dagger$ gates, a single application of controlled-$U$ and $\mathcal{O}((a+1)d)$
other one- and two-qubit gates.
\end{lemma}
Define $\gamma =  ||x_k||^2 $ for simplicity. We remark that even though the above lemma requires $A$ to be Hermitian, however, for non-Hermitian $A$, it still works on the singular values of $A$ instead of eigenvalues (see Theorem 17 and Corollary 18 of \cite{gilyen2019quantum}). We use the above lemmas to perform the following transformation on the block-encoded operator: 
\begin{align}
    \gamma \ket{x_k}\bra{x_k} \longrightarrow e^{-\beta(1-\gamma)} \ket{x_k}\bra{x_k}
\end{align}
Recall that we are given $U_0$ that generates the state $\ket{x_0}$, \cref{lemma: improveddme} allows us to block-encode the operator $\ket{x_0}\bra{x_0}$. Now we take the above block encoding and apply it to $\ket{x_0}$, and according to \cref{def: blockencode}, we obtain the following state:
\begin{align}
    \ket{\bf 0} \Big( e^{-\beta(1-\gamma)} \ket{x_k}\bra{x_k} \Big) \ket{x_0}  \\ + \ket{\rm Garbage} = \ket{\bf 0} \braket{x_k,x_0} e^{-\beta(1-\gamma)} \ket{x_k} + \ket{\rm Garbage}
    \label{eqn: d4}
\end{align}
where we have used the orthogonality of $\ket{x_0}$ and $ \{\ket{v_m} \} $. Measuring the first register and post-select on $\ket{\bf 0}$, yields the state $\ket{x_k}$ on the remaining register. The success probability of this measurement is $|\braket{x_k,x_0}|^2 e^{-2\beta(1-\gamma)}$, which can be improved quadratically better using amplitude amplification. As pointed out in \cite{nghiem2024improved}, by choosing $\beta$ sufficiently small, the value of $e^{-2\beta(1-\gamma)} $ is lower bounded by some constant, e.g., $1/2$, thus the probability can be lower bounded by $ \frac{1}{2}|\braket{x_k,x_0}|  $. From $\ket{x_k}$, we use the block encoding of $A$ to apply and obtain the state:
\begin{align}
    \ket{\bf 0} A \ket{x_k} + \ket{\rm Garbage}
\end{align}
Taking another copy of $\ket{x_k}$ and append another ancilla $\ket{\bf 0}$, we then observe that the overlaps: 
\begin{align}
    \bra{\bf 0}\bra{x_k} \big( \ket{\bf 0} A \ket{x_k} + \ket{\rm Garbage}\big) = \bra{x_k}A \ket{x_k}
\end{align}
which is an approximation to the largest eigenvalue of $A$. In order to achieve an additive error $\epsilon$, i.e., 
\begin{align}
    |\bra{x_k}A\ket{x_k} - A_1| \leq \epsilon \\
    || \ket{x_k} - \ket{A_1} || \leq \epsilon
\end{align}
according to \cite{friedman1998error, golub2013matrix}, the value of $k$ needs to be of order $\mathcal{O}\big( \frac{1}{\Delta}\log \frac{n}{\epsilon} \big)$. In the above procedure, we use the block encoding of $A$ $k$ times, and then use \cref{lemma: theorem56} to transform to a polynomial of degree $\mathcal{O}( \log \frac{1}{\epsilon} )$ (per \cref{lemma: exponential}). The overlaps above can be estimated via Hadamard test or SWAP test, incurring a further $\frac{1}{\epsilon}$ complexity for an estimation of precision $\epsilon$. So the total complexity for estimating largest eigenvalue $A_1$, up to $\epsilon$ error is 
$$\mathcal{O}\Big( \frac{1}{\Delta |\braket{x_k,x_0}| \epsilon} T_A \big(\log \frac{n}{\epsilon}\big) \log\frac{1}{\epsilon}\Big) $$
and the complexity for obtaining $\ket{x_k}$, which is an approximation to $\ket{A_1}$ is $ \mathcal{O}\Big( \frac{1}{\Delta |\braket{x_k,x_0}|} T_A \big(\log \frac{n}{\epsilon}\big) \log \frac{1}{\epsilon} \Big)$. The above summary completes the details for \cref{lemma: largestsmallest}, which we left in the main text. \\

In the following, we show how to obtain the operator $A_1 \ket{A_1}\bra{A_1}$ in the \cref{lemma: extensionlemmalargestsmallest}. Recall from~\cref{eqn: d4} above that we obtained the state: 
\begin{align}
    \ket{\bf 0} \braket{x_k,x_0} e^{-\beta(1-\gamma)} \ket{x_k} + \ket{\rm Garbage} \equiv \ket{\phi}
\end{align}
\cref{lemma: improveddme} allows us to construct the block encoding of $\ket{\phi}\bra{\phi}$, which is: 
\begin{align}
    \ket{\bf 0}\bra{\bf 0} \otimes |\braket{x_k,x_0}|^2 e^{-2\beta(1-\gamma)} \ket{x_k}\bra{x_k} + (...)
\end{align}
where $(...)$ denotes irrelevant term. According to~\cref{def: blockencode}, the above operator is the block encoding of $|\braket{x_k,x_0}|^2e^{-2\beta(1-\gamma)} \ket{x_k}\bra{x_k} $, and the factor $ |\braket{x_k,x_0}|^2$ can be removed using \cref{lemma: amp_amp}. Now we analyze the term $e^{-2\beta(1-\gamma)} $ and show that for a sufficiently small $\beta$, we have $1- e^{-2\beta(1-\gamma)} \leq \epsilon$. Recall that we defined $\gamma =  ||x_k||^2 $, so apparently $-1 \leq \gamma \leq 1$, which implies $ 1- \gamma \geq 1$. We have that:
\begin{align}
    1- e^{-2\beta(1-\gamma)} &\leq \epsilon \\
    \longrightarrow 1- \epsilon &\leq e^{-2\beta(1-\gamma)}  \\
    \longrightarrow \log (1-\epsilon) &\leq  -2\beta (1-\gamma) \\
    \longrightarrow \log \frac{1}{1-\epsilon} &\geq 2\beta (1-\gamma)
\end{align}
which indicates that $\beta \leq \frac{1}{2(1-\gamma)}\log \frac{1}{1-\epsilon} $. So by choosing a sufficiently small value of $\beta$, we have that $e^{-2\beta(1-\gamma)} \leq 1-\epsilon $, thus implying:
\begin{align}
    || \ket{x_k}\bra{x_k} - e^{-2\beta(1-\gamma)} \ket{x_k}\bra{x_k} || \leq |1-  e^{-2\beta(1-\gamma)} | \leq \epsilon
\end{align}
So the block-encoded operator $e^{-2\beta(1-\gamma)} \ket{x_k}\bra{x_k}  $ is $\epsilon$-approximated to $\ket{x_k}\bra{x_k}$, which is again an $\epsilon$-approximation of $\ket{A_1}\bra{A_1}$ provided $k$ is chosen properly, as mentioned in the previous paragraph. By additivity, $ e^{-2\beta(1-\gamma)} \ket{x_k}\bra{x_k} $ is $2\epsilon$-approximation to $ \ket{A_1}\bra{A_1}$. From the block encoding of $e^{-2\beta(1-\gamma)} \ket{x_k}\bra{x_k} $, we can use \cref{lemma: product} to construct the block encoding of $ A e^{-2\beta(1-\gamma)} \ket{x_k}\bra{x_k} \approx A \ket{A_1}\bra{A_1} = A_1 \ket{A_1}\bra{A_1} $, thus completing the \cref{lemma: extensionlemmalargestsmallest}. 

\section{Proof of \cref{lemma: extensiongradientdescent}   }
\label{sec: extensiongradientdescent}
We remind the reader that the goal is to obtain the block encoding of $\lambda_1 \ket{\lambda_1}\bra{\lambda_1}$, and we have the block encoding of $\xbf_T \xbf_T^\dagger$, which is equivalent to $||\xbf_T||^2 \ket{\xbf_T}\bra{\xbf_T}$.  This block-encoded operator is essentially similar to what we had in~\cref{d2}  , therefore, we can follow exactly the same procedure as in previous section (everything below~\cref{d2}), and end up obtaining an $\epsilon$-approximated block encoding of $\mathcal{C} \ket{\xbf_T}\bra{\xbf_T}$. As worked out in the main text, by choosing $T = \mathcal{O}(\log\frac{1}{\epsilon})$, it is guaranteed that $|| \ket{\xbf_T} - \ket{\lambda_1}||\leq \epsilon$. Therefore, by additivity of error, we have that  $|| \mathcal{C} \ket{\xbf_T}\bra{\xbf_T}  - \lambda_1 \ket{\lambda_1}\bra{\lambda_1} || \leq 2\epsilon$.

\section{Proof of convergence guarantee for imaginary time evolution}
\label{sec: proofite}
In this section we show the efficiency of imaginary time evolution algorithm. Let $H$ be the Hamiltonian on $n$ qubits, and $\ket{\Phi_0}, \ket{\Phi_1},..., \ket{\Phi_{2^n-1}}$ are eigenvectors of $H$ with corresponding eigenvalue $E_0,E_1,....,E_{2^n-1}$, assumed to have an ordering $E_0 < E_1 \leq E_2 \leq ... \leq E_{2^n-1}$.  Suppose that we begin with an initial state $\ket{\Phi}$ having the following decomposition $ \sum_{i=0}^{2^n-1} a_i \ket{\Phi_i}$. Under the action of $\exp(-t H)$, we have:
\begin{align}
   \Phi_t \equiv  e^{-t H} \ket{\Phi} &= \sum_{i=0}^{2^n-1} a_i e^{-t H}\ket{\Phi_i} \\
    &= \sum_{i=0}^{2^n-1} a_i e^{-t E_i} \ket{\Phi_i} \\
    &= e^{-t E_0} \Big(  \sum_{i=0}^{2^n-1} a_i e^{-t (E_i-E_0)} \ket{\Phi_i} \Big)
\end{align}
As $E_0$ is the lowest eigenvalue, we have that for all $i$, $E_i-E_0$ is greater than zero. Thus, in the limit $t\longrightarrow \infty$, all the terms $e^{-t (E_i-E_0)} $ vanishes, leaving the state $\Phi_t \approx a_0 e^{-t E_0} \ket{\Phi_0}$. Normalization yields:
\begin{align}
   \ket{\Phi_t} \equiv  \frac{ \Phi_t  }{ ||\Phi_t||} = \ket{\Phi_0}
\end{align}
which is exactly ground state. Now we consider at time $t$, the normalized state $\ket{\Phi_t}$ is:
\begin{align}
   \ket{\Phi_t} \equiv \frac{ \Phi_t  }{ ||\Phi_t||} = \frac{e^{-t E_0} \Big(  \sum_{i=0}^{2^n-1} a_i e^{-t (E_i-E_0)} \ket{\Phi_i} \Big)}{ e^{-tE_0} \sqrt{\big( \sum_{i=0}^{2^n-1} |a_i|^2 e^{-2t (E_i-E_0)}    \big) } } = \frac{\Big(  \sum_{i=0}^{2^n-1} a_i e^{-t (E_i-E_0)} \ket{\Phi_i} \Big)}{ \sqrt{\big( \sum_{i=0}^{2^n-1} |a_i|^2 e^{-2t (E_i-E_0)}    \big) } }
\end{align}
We have the overlaps:
\begin{align}
    \braket{\Phi_0,\Phi_t} &= \frac{a_0}{  \sqrt{\big( \sum_{i=0}^{2^n-1} |a_i|^2 e^{-2t (E_i-E_0)}    \big) } } \\
    &= \frac{1}{ \sqrt{ 1+  \sum_{i=1}^{2^n-1} \frac{|a_i|^2}{|a_0|^2}e^{-2t (E_i-E_0)}   } } \\
    & \geq \frac{1}{ 1+  \sum_{i=1}^{2^n-1} \frac{|a_i|^2}{|a_0|^2}e^{-2t (E_i-E_0)}    }  \\
    & \geq 1- \sum_{i=1}^{2^n-1} \frac{|a_i|^2}{|a_0|^2}e^{-2t (E_i-E_0)}    \\
    & \geq \frac{1}{2} - \sum_{i=1}^{2^n-1} \frac{|a_i|^2}{|a_0|^2}e^{-2t (E_i-E_0)} 
\end{align}
If we demand this state is $\epsilon$ close to the ground state $\ket{\Phi_0}$, then we need that $|\braket{\Phi_0,\Phi_t}| \geq \frac{1-\epsilon}{2}$, which implies:
\begin{align}
    \sum_{i=1}^{2^n-1} \frac{|a_i|^2}{|a_0|^2}e^{-2t (E_i-E_0)}  \leq \frac{\epsilon }{2}
\end{align}
Let $a = \max_i \frac{|a_i|^2}{|a_0|^2}  $. Because for all $i \geq 2$, we have $E_i - E_0 > E_1- E_0$, which indicates $e^{-2t (E_i-E_0)} < e^{-2t (E_1-E_0)}   $. Thus the left-hand side of the above equation is upper bounded by:
\begin{align}
    \sum_{i=1}^{2^n-1} \frac{|a_i|^2}{|a_0|^2}e^{-2t (E_i-E_0)}  \leq \sum_{i=1}^{2^n-1} a e^{-2t (E_1-E_0)} = (2^n-1) a e^{-2t (E_1-E_0)}  
\end{align}
Demanding $(2^n-1) a e^{-2t (E_1-E_0)}  = \frac{\epsilon }{2} $, we have:
\begin{align}
            e^{-2t (E_1-E_0)}  = \frac{\epsilon}{2 a (2^n-1)} \longrightarrow t = \frac{1}{(E_1-E_0)}  \log\Big( 2a\cdot  \frac{2^n-1 }{ \epsilon} \Big) 
\end{align}
which indicates that $t$ is highly efficient with respect to $n$ and $\epsilon$. Thus, it justifies the efficiency of imaginary evolution algorithm.

\end{document}